\documentclass{sig-alternate-2013}
\usepackage{url,xspace,cite,subfigure}
\usepackage{comment}
\usepackage{subfigure}

\def\full{1}        % set 1 for a full tech report version
\def\shownotes{0}   % set 1 for version with author notes
\def\anon{0}        % set 1 to anonymize
\ifnum\full=1 \usepackage{mdwlist} \fi

\newtheorem{theorem}{Theorem}[section]
\newtheorem{definition}[theorem]{Definition}
\newtheorem{lemma}[theorem]{Lemma}

\newtheorem{obs}[theorem]{Observation}

\newtheorem{cor}[theorem]{Corollary}
\ifnum\shownotes=1
\newcommand{\authnote}[2]{{ $\ll$\textsf{\footnotesize #1 notes: #2}$\gg$}}
\else
\newcommand{\authnote}[2]{}
\fi

\providecommand{\first}{$1^{st}$\xspace}
\providecommand{\third}{$3^{rd}$\xspace}
\providecommand{\second}{$2^{nd}$\xspace}

\providecommand{\vs}{vs. }
\providecommand{\ie}{\emph{i.e.,} }
\providecommand{\eg}{\emph{e.g.,} }
\providecommand{\cf}{\emph{cf.,} }
\providecommand{\resp}{\emph{resp.,} }
\providecommand{\etal}{\emph{et al.}}   %Removed trailing space here; usually want non-breaking space with following reference
\providecommand{\etc}{\emph{etc.}}      % No trailing space here either
\providecommand{\mypara}[1]{\smallskip\noindent\emph{#1} }
\providecommand{\myparab}[1]{\smallskip\noindent\textbf{#1} }

\providecommand{\LP}{\textbf{LP}\xspace}
\providecommand{\LP}[1]{\textbf{LP}\textbf{#1}\xspace}

\providecommand{\SP}{\textbf{SP}\xspace}
\providecommand{\SecP}{\textbf{SecP}\xspace}

\providecommand{\TB}{\textbf{TB}\xspace}
\providecommand{\Ex}{\textbf{Ex}\xspace}
\providecommand{\legit}{legitimate\xspace}
\providecommand{\attacked}{attacked\xspace}

\providecommand{\PR}{\textsf{PR}}
\providecommand{\BPR}{\textsf{BPR}}

\providecommand{\Nxt}{\textsf{Nxt}}
\providecommand{\BR}{\textsf{BR}}

\newfont{\mycrnotice}{ptmr8t at 7pt}
\newfont{\myconfname}{ptmri8t at 7pt}
\let\crnotice\mycrnotice%
\let\confname\myconfname%

\permission{Permission to make digital or hard copies of all or part of this work for personal or classroom use is granted without fee provided that copies are not made or distributed for profit or commercial advantage and that copies bear this notice and the full citation on the first page. Copyrights for components of this work owned by others than the author(s) must be honored. Abstracting with credit is permitted. To copy otherwise, or republish, to post on servers or to redistribute to lists, requires prior specific permission and/or a fee. Request permissions from permissions@acm.org.}
\conferenceinfo{SIGCOMM'13,}{August 12--16, 2013, Hong Kong, China. \\
{\mycrnotice{Copyright is held by the owner/author(s). Publication rights licensed to ACM.}}}
\CopyrightYear{2013}
\crdata{978-1-4503-2056-6/13/08}
\title{BGP Security in Partial Deployment}
\subtitle{Is the Juice Worth the Squeeze?
\ifnum\full=1
\begin{large}  \\ Full version \ifnum\anon=0 from \today \fi \end{large}
\fi
}
\ifnum\anon=0
\numberofauthors{3}
\author{
\alignauthor Robert Lychev*\\
\affaddr{Georgia Tech \\ Altanta, GA, USA \\ \url{rlychev@cc.gatech.edu}} \\
\alignauthor Sharon Goldberg \\
\affaddr{Boston University \\ Boston, MA, USA \\ \url{goldbe@cs.bu.edu}} \\
\alignauthor Michael Schapira \\
\affaddr{Hebrew University \\ Jerusalem, Israel \\ \url{schapiram@huji.ac.il}} \\
}

\else
\author{[Paper \#68]}
\fi

\begin{document}

\maketitle

\begin{abstract}
As the rollout of secure route origin authentication with the RPKI slowly gains traction among network operators, there is a push to standardize secure path validation for BGP (\ie S*BGP: S-BGP, soBGP, BGPSEC, etc.). Origin authentication already does much to improve routing security.  Moreover, the transition to S*BGP is expected to be long and slow, with S*BGP coexisting in ``partial deployment'' alongside BGP for a long time.   We therefore use theoretical and experimental approach to study the security benefits provided by partially-deployed S*BGP, vis-a-vis those already provided by origin authentication.
Because routing policies have a profound impact on routing security, we use a survey of 100 network operators to find the policies that are likely to be most popular during partial S*BGP deployment. We find that S*BGP provides only meagre benefits over origin authentication when these popular policies are used.  We also study the security benefits of other routing policies, provide prescriptive guidelines for partially-deployed S*BGP, and show how interactions between S*BGP and BGP can introduce new vulnerabilities into the routing system.

\end{abstract}

\noindent \textbf{Categories and Subject Descriptors:}  C.2.2 [Computer-
Communication Networks]: Network Protocols


%\myparab{General Terms:} Security;

%refs of deployment times
%\Snote{refs: IPv6 \url{http://ipv6forum.com.au/timeline.php}; DNSSEC: \url{http://www.internetdagarna.se/arkiv/2008/www.internetdagarna.se/images/stories/pdf/domannamn/Steve_Crocker_administrationofdnssec.pdf}}

\section{Introduction}\label{sec:intro}

Recent high-profile routing failures~\cite{pakistan,china,moratel,DEFCON} have highlighted major vulnerabilities in BGP, the Internet's interdomain routing protocol.  To remedy this, secure  origin authentication~\cite{pfx-validate-09,OrAuth,BGPsurvey} using the RPKI~\cite{rpki} is gaining traction among network operators, and there is now a push to standardize a path validation protocol (\ie S*BGP~\cite{SBGP,soBGP,BGPSEC}).  Origin authentication is relatively lightweight,  requiring neither changes to the BGP message structure nor online cryptographic computations. Meanwhile, path validation with S*BGP could require both~\cite{BGPSEC}.  The deployment of origin authentication is already a significant challenge~\cite{fccReport}; here we ask, is the deployment of S*BGP path validation worth the extra effort? (That is, is the juice worth the squeeze?)

%Given that the deployment of origin authentication is already a significant challenge~\cite{fccReport}, we ask, is path validation worth the extra effort?

%However, the deployment of route origin authentication is already a significant challenge~\cite{fccReport}; here we ask, is the incremental deployment of a path validation protocol really worth the extra effort?

To answer this question, we must contend with the fact that any deployment of S*BGP is likely to coexist with legacy insecure BGP for a long time. (IPv6 and DNSSEC, for example, have been in deployment since at least 1999 and 2007 respectively.) In a realistic \emph{partial deployment} scenario, an autonomous system (AS) that has deployed S*BGP will sometimes need to accept insecure routes sent via legacy BGP; otherwise, it would lose connectivity to the parts of the Internet that have not yet deployed S*BGP~\cite{BGPSEC}.
Most prior research has ignored this issue, either by assuming that ASes will never accept insecure routes~\cite{JenYannis,adoptability}, by studying only the \emph{full deployment} scenario where every AS has already deployed S*BGP~\cite{BGPattack,BGPsurvey}, or by focusing on creating incentives for ASes to adopt S*BGP in the first place~\cite{adopt,adoptability}.

We consider the security benefits provided by partially-deployed S*BGP vis-a-vis those already provided by origin authentication.
Fully-deployed origin authentication is lightweight and already does much to improve security, even against attacks it was not designed to prevent (\eg propagation of bogus AS-level paths)~\cite{BGPattack}.
We find that, given the routing policies that are likely to be most popular during partial deployment,  S*BGP can provide only meagre improvements to security over what is already possible with origin authentication; we find that other, less popular policies can sometimes provide tangible security improvements. (``Popular'' routing policies were found using a survey of 100 network operators~\cite{surveyEmail}.)   However, we also show that security improvements can come at a risk; complex interactions between BGP and S*BGP can introduce new instabilities and  vulnerabilities into the routing system.

\subsection{Security with partially-deployed S*BGP.}\label{sec:intro:sbgp}

With BGP, an AS learns AS-level paths to destination ASes (and their IP prefixes) via routing announcements from neighboring ASes; it then selects one path per destination by applying its local \emph{routing policies}.  Origin authentication ensures that the destination AS that announces a given IP prefix is really authorized to do so. S*BGP ensures that the AS-level paths learned actually exist in the network.

In S*BGP partial deployment, security will be profoundly affected by the routing policies used by individual ASes, the AS-level topology, and the set of ASes that are \emph{secure} (\ie have deployed S*BGP). Suppose a secure AS has a choice between a \emph{secure route} (learned via S*BGP) and an \emph{insecure route} (learned via legacy BGP) to the same destination.  While it seems natural that the AS should always prefer the secure route over the insecure route, a network operator must balance security against economic and performance concerns.  As such, a \emph{long} secure route through a \emph{costly} provider might be less desirable than a \emph{short} insecure route through a \emph{revenue-generating} customer. Indeed, the BGPSEC standard is careful to provide maximum flexibility, stating the relationship between an AS's routing policies and the security of a route ``is a matter of local policy''~\cite{BGPSEC}.

%These local policy decisions, however, can have global implications, \eg if a secure AS chooses an insecure route this might force its secure neighbors to also use an insecure route, regardless of their own routing policies, for lack of other alternatives.
%
%These local policy decisions, however, can have global implications: if a secure AS chooses an insecure route, it's secure neighbors may be forced to choose an insecure routes for lack of other options.

While this flexibility
is a prerequisite for assuring operators that S*BGP will not disrupt existing traffic engineering or network management
\ifnum\full=1
polices
\footnote{Practitioners commonly resist deployment of a new protocol because it ``breaks'' their networks; witness the zone enumeration issue in DNSSEC~\cite{NSECrfc} or the fact that IPv6 is sometimes disabled because it degrades DNS performance~\cite{happyEyeballs}.},
\else
policies,
\fi
 it can have dire consequences on security.
%
%First off, it means that an AS that has deployed S*BGP may not learn any routes via S*BGP, because its secure neighbor ASes prefer to route along insecure paths.  Worse yet,
Attackers can exploit routing policies that prioritize economic and/or length considerations above security.  In a \emph{protocol downgrade attack}, for example, an attacker convinces a secure AS with a secure route to downgrade to a \emph{bogus} route sent via legacy BGP, simply because the bogus route is shorter, or less costly (Section~\ref{sec:pda}).

\subsection{Methodology \& paper roadmap.}\label{sec:intro:method}

\ifnum\full=0
This paper summarizes our results.  Extended analysis, robustness tests, and proofs are in the full version~\cite{full}.
\else
%This is the full version of \cite{partialSec}.
\fi

\myparab{Three routing models.} In Section~\ref{sec:model} we develop models for routing with partially-deployed S*BGP, based on classic models of AS business relationships and BGP~\cite{gr,ggr,gsw,hustonModel1,hustonModel12}.  Our \emph{security \first model} supposes that secure ASes \emph{always} prefer secure routes over insecure ones; while this is most natural from a security perspective, a survey of 100 network operators~\cite{surveyEmail} suggests that it is least popular in partial deployment.  In our \emph{security \second model}, a secure route is preferred only if no \emph{less-costly} insecure route is available. The survey confirms that our \emph{security \third model} is most popular in partial deployment~\cite{surveyEmail}; here a secure route is preferred only if there is no \emph{shorter or less-costly} insecure route.
\ifnum\full=0
This paper works within these models; the full version assesses robustness to assumptions made in these models.
\else
In Appendix~\ref{apx:robust:policies} we analyze the robustness of our results to assumptions made in these models.
\fi

\myparab{Threat model \& metric.} Sections~\ref{sec:threat}-\ref{sec:metric} introduce our threat model, and a metric to quantify security within this threat model; our metric measures the \emph{average} fraction of ASes using a legitimate route when a destination is attacked.

\myparab{Deployment invariants.} The vast number of choices for the set $S$ of ASes that adopt S*BGP makes evaluating security challenging. Section~\ref{sec:partitions} therefore presents our (arguably) most novel methodological contribution; a framework that bounds the \emph{maximum} improvements in security possible for each routing model, \emph{for any} deployment scenario $S$.

\myparab{Deployment scenarios.} How close do real S*BGP deployments $S$ come to these bounds? While a natural objective would be to determine the ``optimal'' deployment $S$, we prove that this is NP-hard.  Instead, Sections~\ref{sec:results:metric}-\ref{sec:wrapup} use simulations on empirical AS-level graphs to quantify security in scenarios suggested in the literature~\cite{FCCcommit,JenYannis,adopt,adoptability}, and determine root causes for security improvements (or lack thereof).

\myparab{Algorithms \& experimental robustness. } We designed parallel simulation algorithms to deal with the large space of parameters that we explore, \ie attackers, destinations, deployment scenarios $S$, and routing policies,
\ifnum\full=0
 (full version).
\else
 (Appendix~\ref{apx:algos}~and~\ref{apx:sim}).
\fi
We also controlled for empirical pitfalls, including
%%%
(a) variations in routing policies
\ifnum\full=1
(Appendix~\ref{apx:robust:policies})
\else
(full version)
\fi
%%%
(b) the fact that empirical AS-level graphs tend to miss many peering links at Internet eXchange Points (IXPs)~\cite{10lessons,ixp,ixpSIGCOMM12},
\ifnum\full=1
(Section~\ref{sec:policies}, Appendix~\ref{apx:ixp})
\else
(Section~\ref{sec:policies})
\fi
%%%
(c) a large fraction of the Internet's traffic originates at a few ASes~\cite{LabovitzSIGCOMM} (Sections~\ref{sec:policies},~\ref{sec:partitions:robust},~\ref{sec:metric:CPs},~\ref{sec:T1suck:metric}).   %%
While our analysis cannot predict exactly how individual ASes would react to routing attacks, we do report on strong aggregate trends.

\ifnum\full=1
\myparab{Proofs.} Proofs of our theorems are in Appendix~\ref{apx:algos}-\ref{apx:hardness_results}.
\fi

\subsection{Results.}\label{sec:intro:results}

\ifnum\full=1
Our simulations, empirically-validated examples, and theoretical analyses indicate the following:
\else
Our theoretical and experimental analyses indicate that:
\fi

\myparab{Downgrades are a harsh reality.}  We find that protocol downgrade attacks (Sections~\ref{sec:intro:sbgp},~\ref{sec:pda}) can be extremely effective; so effective, in fact, that they render deployments of S*BGP at large Tier 1 ISPs almost useless in the face of attacks (Sections~\ref{sec:T1suck}~and~\ref{sec:T1suck:metric}).

\myparab{New vulnerabilities. } We find that the interplay between topology and routing policies can cause some ASes to fall victim to attacks they would have avoided if S*BGP had \emph{not} been deployed.  Fortunately, these troubling phenomena occur less frequently than phenomena that protect ASes from attacks during partial deployment (Section~\ref{sec:wrapup}).

\myparab{New instabilities.}   We show that undesirable phenomena (BGP Wedgies~\cite{wedgies}) can occur if ASes prioritize security inconsistently (Section~\ref{sec:stability}).

\myparab{Prescriptive deployment guidelines.}  Other than suggesting that (1) ASes should prioritize security in the same way in order to avoid routing instabilities,  our results (2) confirm that deploying  lightweight \emph{simplex} S*BGP~\cite{BGPSEC,adopt} (instead of full-fledged S*BGP) at stub ASes at the edge of the Internet does not harm security (Section~\ref{sec:simplex}).  Moreover, while~\cite{adopt,JenYannis,adoptability} suggest that Tier 1s should be  early adopters of S*BGP, our results do not support this; instead, we suggest that (3) Tier 2 ISPs should be among the earliest adopters of S*BGP (Section~\ref{sec:T1suck},~\ref{sec:islands},~\ref{sec:T1suck:metric}).

\myparab{Is the juice worth the squeeze?} We use our metric to compare S*BGP in a partial deployment $S$ to the baseline scenario where no AS is secure (\ie $S=\emptyset$ and only origin authentication is in place).
We find that large partial deployments of S*BGP provide excellent protection against attacks when ASes use routing policies that prioritize security \first (Section~\ref{sec:islands}); however, \cite{surveyEmail} suggests that network operators are less likely to use these routing policies.  Meanwhile, the policies that operators most favor (\ie security~\third) provide only meagre improvements over origin authentication (Section~\ref{sec:upperLower}).
This is not very surprising, since S*BGP is designed to prevent path-shortening attacks and when security is \third, ASes prefer (possibly-bogus) short insecure routes over longer secure routes.

However, it is less clear what happens in security is \second, where route security is prioritized over route length.
Unfortunately, even when S*BGP is deployed at 50\% of ASes, %(including all Tier 1s and Tier 2s),
the benefits obtained in the security \second model lag significantly behind those available when security is \first.  While some destinations can obtain tangible benefits when security is \second, for others (especially Tier 1s) the security \second model behaves much like the security \third model (Section~\ref{sec:bigDeps}). We could only find clear-cut evidence of strong overall improvement in security when ASes prioritize security \first.

\ifnum\full=1 \newpage \fi
\section{Security \& Routing Policies}\label{sec:model}

S*BGP allows an AS to validate the correctness of the AS-level path information it learns from its neighbors~\cite{BGPsurvey}. (S-BGP~\cite{SBGP} and BGPSEC~\cite{BGPSEC} validate that every AS on a path sent a routing announcement for that path; soBGP~\cite{soBGP} validates that all the edges in a path announcement physically exist in the AS-level topology.  As we shall see in Section~\ref{sec:threat}, our analysis applies to all these protocols.) However, for S*BGP to prevent routing attacks, validation of paths alone is not sufficient. ASes also need to use information from path validation to make their routing decisions.  We consider three alternatives for incorporating path validation into routing decisions, and analyze the security of each.

%and present a framework for analyzing routing with S*BGP.

%\ifnum\full=1 \newpage \fi
\subsection{Dilemma: Where to place security?}\label{sec:bgp-decision}

An AS that adopts S*BGP must be able to process and react to insecure routing information, so that it can still route to destination ASes that have not yet adopted S*BGP. The BGPSEC standard is such that a router only learns a path via BGPSEC if every AS on that path has adopted BGPSEC; otherwise, the path is learnt via legacy BGP. (The reasoning for this is in \cite{sriramBGPSECchoices} and Appendix A of~\cite{adopt}):

\myparab{Secure routes.} We call an AS that has adopted S*BGP a \emph{secure AS}, and a path learned via S*BGP (\ie a path where every AS is secure) a \emph{secure path} or \emph{secure route}; all other paths are called \emph{insecure}.

\smallskip\noindent
If a secure AS can learn both secure and insecure routes, what role should security play in route selection?
\ifnum\full=1
To blunt routing attacks, secure routes should be preferred over insecure routes. But how should \emph{expensive} or \emph{long} secure routes be ranked relative to \emph{revenue-generating} or \emph{short} insecure routes?
\fi

\subsection{S*BGP routing models.}\label{sec:policies}

\begin{table}
\begin{center}
\begin{scriptsize}
\begin{tabular}{|c|p{2.6in}|}
  \hline
  Tier 1 & 13 ASes with high customer degree \& no providers\\\hline
  Tier 2 & 100 top ASes by customer degree \& with providers \\\hline
  Tier 3 & Next 100 ASes by customer degree \& with providers\\\hline
  \ifnum\full=1
  CPs &  17 Content provider ASes listed in Figure~\ref{fig:t1_cp_sec_r_breakdown} \\\hline
  \else
  CPs & 17 Content providers: AS 15169, 8075, 20940, 22822, 32934, 15133, 16265, 16509, 2906, 23286, 40428, 714, 10310, 38365, 14907 13414, and 4837.\\\hline
  \fi
  Small CPs & Top 300 ASes by peering degree \\
            & (other than Tier 1, 2, 3, and CP) \\\hline
  Stubs-x &  ASes with peers but no customers\\%  (8186 $|$ 8686)\\ %
  \hline
  Stubs & ASes with no customers \& no peers \\%(24771 $|$ 24179) \\
  \hline
  SMDG  & Remaining non-stub ASes \\%(5859 $|$ 5628)\\
  \hline
\end{tabular}
\vspace{-3mm}\caption{Tiers.} %Notation: ($\#$ ASes in UCLA graph $|$ $\#$ ASes in graph with IXP edges)}
\label{tab:tiers}
\vspace{-5mm}
\end{scriptsize}
\end{center}
\end{table}

While it is well known that BGP routing policies differ between ASes and are often kept private, we need a concrete model of ASes' routing policies so as to analyze and simulate their behaviors during attacks. The following models of routing with S*BGP are variations of the well-studied models from~\cite{gr,ggr,gsw,adopt,PaulF,hustonModel1,hustonModel12}.

\myparab{AS-level topology.} The AS-level topology is represented by an undirected graph $G=(V,E)$; the set of vertices $V$ represents ASes and the set of links (edges) $E$ represents direct BGP links between neighboring ASes. We will sometimes also refer to the ``tiers'' of ASes~\cite{amogh} in Table~\ref{tab:tiers}; the list of 17 content providers (CPs) in Table~\ref{tab:tiers}
\ifnum\full=1
(or see Table~\ref{tab:CP} and Figure~\ref{fig:t1_cp_sec_r_breakdown})
\fi
was culled from recent empirical work on interdomain traffic volumes~\cite{labovitzHyperGiant,labovitzPreso,LabovitzSIGCOMM,sandvine,alexa}.

\myparab{ASes' business relationships.} Each edge in $E$ is annotated with a business relationship: either (1) \emph{customer-to-provider}, where the customer purchases connectivity from its provider (our figures depict this with an arrow from customer to provider), or (2) \emph{peer-to-peer}, where two ASes transit each other's customer traffic for free (an undirected edge).

\myparab{Empirical AS topologies.} All simulations and examples described in this paper were run on the UCLA AS-level topology from 24 September 2012~\cite{cyclops}.
\ifnum\full=1
We preprocessed the graph by (1) renaming all 4-byte ASNs in more convenient way, and (2) recursively removing all ASes that had no providers that had low degree (and were not Tier 1 ISPS).  The resulting graph had 39056 ASes, 73442 customer-provider links and 62129 peer-to-peer links.
\fi
Because empirical AS graphs often miss many of peer-to-peer links in Internet eXchange Points (IXP)~\cite{10lessons,ixp,ixpSIGCOMM12}, we constructed a second graph where we augmented the UCLA graph with over 550K peer-to-peer edges between ASes listed as members of the same IXP
\ifnum\full=1
(on September 24, 2012)
\fi
on voluntary online sources (IXPs websites, EuroIX, Peering DB, Packet Clearing House, \etc).
\ifnum\full=1
Our list contained 332 IXPs and 10,835 mappings of member ASes to IXPs; after connecting \emph{every} pair of ASes that are present in the same IXP (and were not already connected in our original UCLA AS graph) with a peer-to-peer edge, our graph was augmented with 552933 extra peering links.
\fi
Because \emph{not} all ASes at an IXP peer with each other~\cite{ixpSIGCOMM12}, our augmented graph is an upper bound on the number of missing links in the AS graph.  When we repeated our simulations on this second graph, we found that all the aggregate trends we discuss in subsequent sections still hold, which suggests they are robust to missing IXP edges.
\ifnum\full=1
(Results in Appendix~\ref{apx:ixp}.)
\else
(Results in full version.)
\fi

\myparab{S*BGP routing.}  ASes running BGP compute routes to each \emph{destination} AS $d \in V$ independently.  For every destination AS $d \in V$, each \emph{source} AS $s \in V \backslash \{d\}$ repeatedly uses its local \emph{BGP decision process} to select a single ``best'' route to $d$ from routes it learns from neighboring ASes.  $s$ then announces this route to a subset of its neighbors according to its \emph{local export policy}. An AS $s$ \emph{learns a route} or has an \emph{available route} $R$ if $R$ was announced to $s$ by one of its neighbors;  AS $s$ \emph{has} or \emph{uses a route} $R$ if it chooses $R$ from its set of available routes. AS $s$ has customer (\resp peer, provider) route if its neighbor on that route is a customer (\resp peer, provider); see \eg AS 29518 in Figure~\ref{fig:wedgies} left.

\subsubsection{Insecure routing policy model.} \label{sec:insec_policies}

When choosing between many routes to a destination $d$, each \emph{insecure} AS executes the following (in order):

\myparab{Local pref (\LP):} Prefer customer routes over peer routes. Prefer peer routes over provider routes.

\myparab{AS paths (\SP):} Prefer shorter routes over longer routes.

\myparab{Tiebreak (\TB):} Use intradomain criteria (\eg geographic location, device ID) to break ties among remaining routes.

\smallskip\noindent
After selecting a single route as above, an AS announces that route to a subset of its neighbors:

\myparab{Export policy (\Ex):}  In the event that the route is via a customer, the route is exported to all neighbors. Otherwise, the route is exported to customers only.

\smallskip\noindent
The relative ranking of the \LP, \SP, and \TB are standard in most router implementations~\cite{bestpath}.
The \LP and \Ex steps are based on the classical economic model of BGP routing~\cite{gr,ggr,hustonModel1,hustonModel12}.  \LP captures ASes' incentives to send traffic along revenue-generating customer routes, as opposed to routing through peers (which does not increase revenue), or routing through providers (which comes at a monetary cost). \Ex captures ASes's willingness to transit traffic only when paid to do so by a customer.

\myparab{Robustness to \LP model. } While this paper reports results for the above \LP model, we also test their robustness to other models for \LP; results are in
\ifnum\full=0
the full version.
\else
Appendix~\ref{apx:robust:policies}.
\fi

\subsubsection{Secure routing policy models.}

\noindent
Every \emph{secure} AS also adds this step to its routing policy.

\myparab{Secure paths (\SecP):} Prefer a secure route over an insecure route.

\noindent
We consider three models for incorporating the \SecP step:

\myparab{Security \first.}  The \SecP is placed before the \LP step; this model supposes security is an AS's highest priority.

\myparab{Security \second.} The \SecP step comes between the \LP and \SP steps; this model supposes that an AS places economic considerations above security concerns.

\myparab{Security \third. } The \SecP step comes between \SP and \TB steps; this model, also used in \cite{adopt}, supposes security is prioritized below business considerations and AS-path length.

\subsubsection{The security \first model is unpopular.}\label{sec:survey}

While the security \first model is the most ``idealistic'' from the security perspective, it is likely the least realistic. During incremental deployment, network operators are expected to cautiously incorporate S*BGP into routing policies, placing security \second or \third, to avoid disruptions due to (1) changes to traffic engineering, and (2) revenue lost when expensive secure routes are chosen instead of revenue-generating customer routes.  The security \first model might be used only once these disruptions are absent (\eg when most ASes have transitioned to S*BGP), or to protect specific, highly-sensitive IP prefixes.  Indeed, a survey of 100 network operators \cite{surveyEmail}
%on the NANOG and other mailing lists
found that 10\% would rank security \first, 20\% would rank security \second and 41\% would rank security \third. (The remaining operators opted not to answer this question.)

\begin{figure}
\begin{center}
  \includegraphics[width=1.7in]{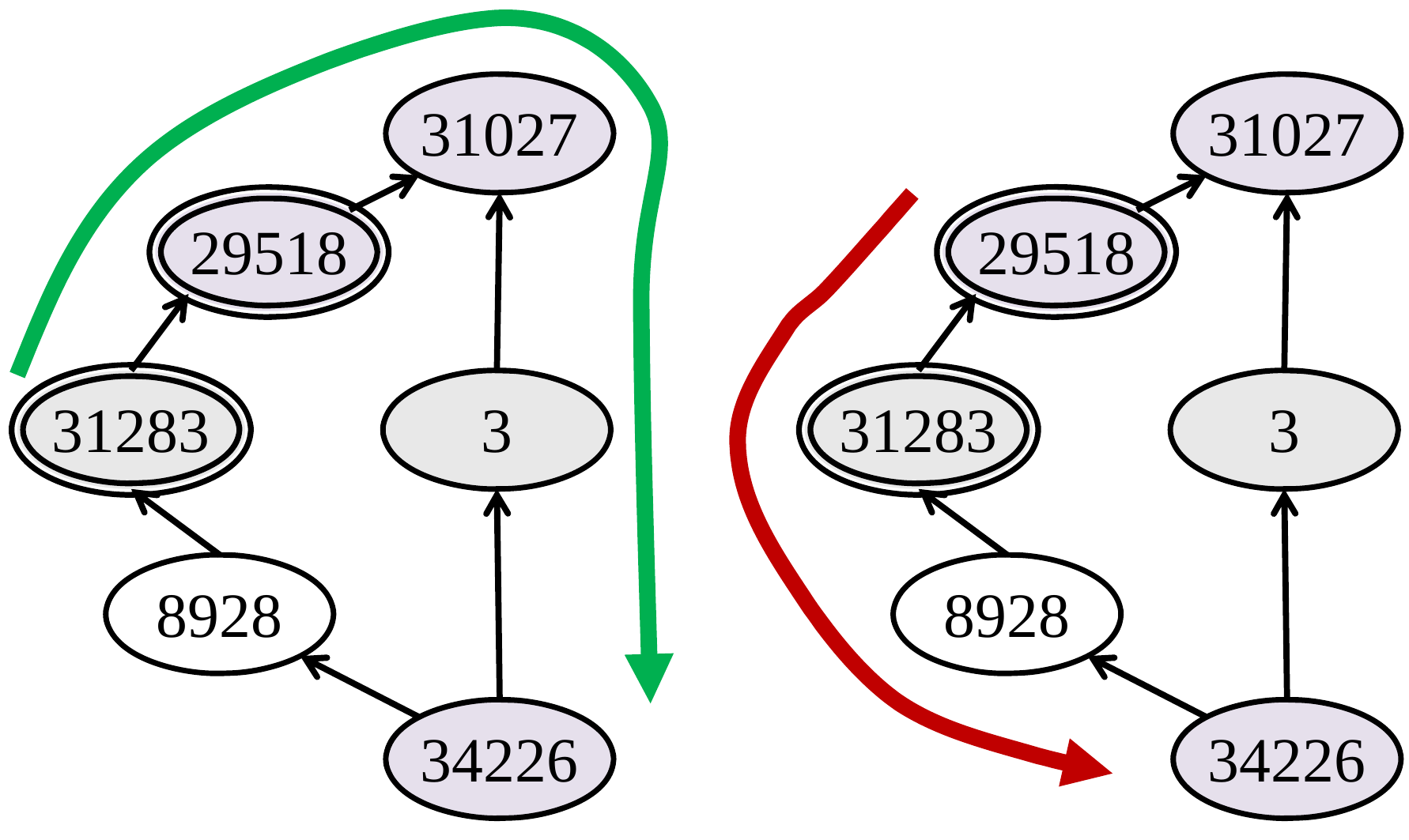}
  \hspace{2mm}
  \includegraphics[width=.8in]{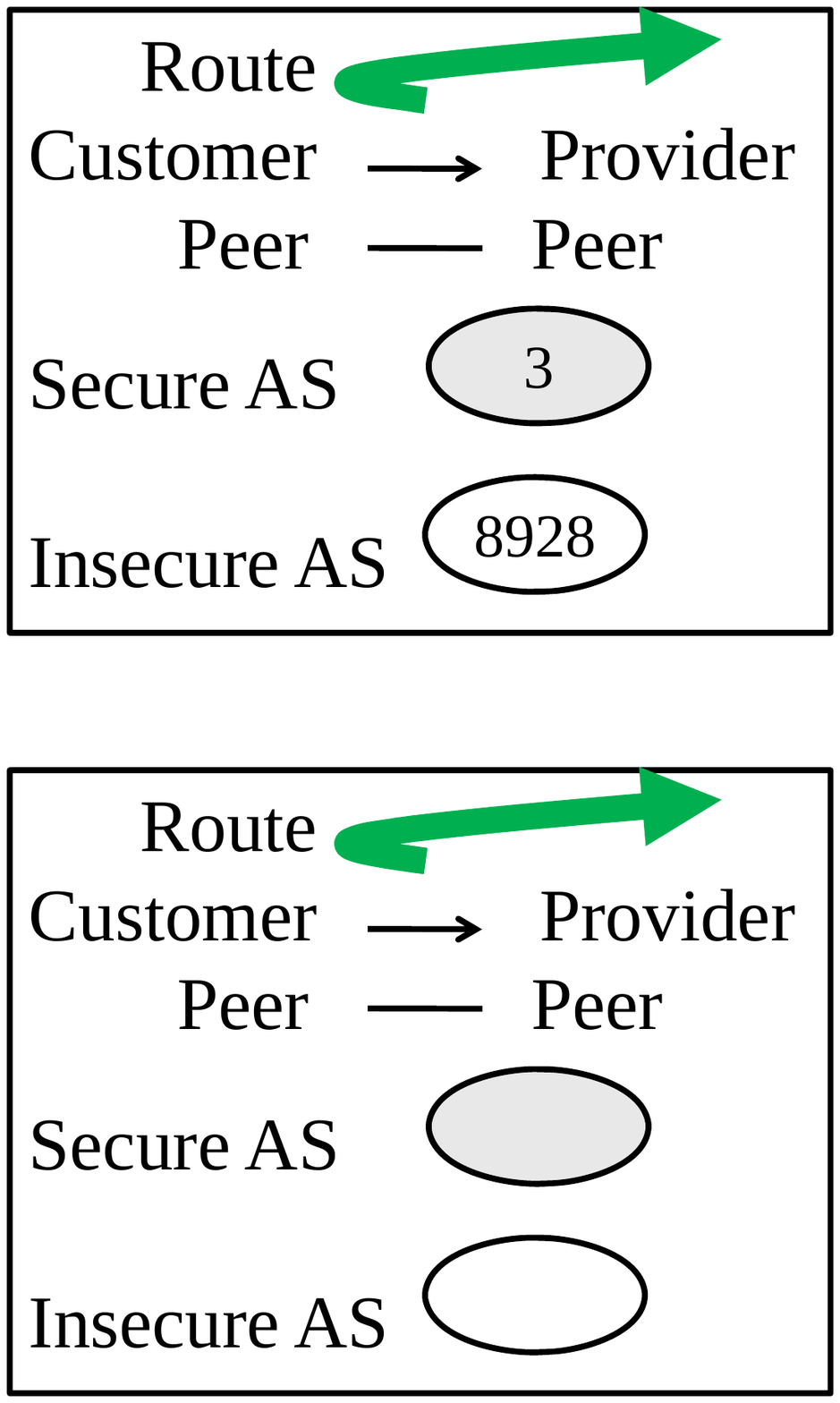}
  \vspace{-3mm}\caption{S*BGP Wedgie.}
  \label{fig:wedgies}
  \vspace{-8mm}
  \end{center}
\end{figure}

\subsection{Mixing the models?}\label{sec:stability}

It is important to note that in each of our S*BGP routing models, the prioritization of the \SecP step in the route selection process is consistent across ASes. The alternative---lack of consensus amongst network operators as to where to place security in the route selection process---can lead to more than just confusion; it can result in a number of undesirable phenomena that we discuss next.

\subsubsection{Disagreements can lead to BGP Wedgies.}

\myparab{Figure~\ref{fig:wedgies}.} Suppose that all ASes in the network, except AS 8928, have deployed S*BGP. The Swedish ISP AS 29518 places security below \LP in its route selection process, while the Norwegian ISP AS 31283 prioritizes security above all else (including \LP). Thus, while AS 29518 prefers the customer path through  AS 31283,  AS 31283 prefers the secure path through its provider AS 29518.
The following undesirable scenario, called a ``BGP Wedgie''~\cite{wedgies} can occur. Initially, the network is in an intended \emph{stable routing state}\footnote{A routing state, \ie the route chosen by each AS $s \in V \backslash \{d\}$ to destination $d$, is \emph{stable} if any AS $s$ that re-runs its route selection algorithm does not change its route~\cite{gsw}.}, in which AS 31283 uses the secure path through its provider AS 29518 (left).
Now suppose the link between AS 31027 and AS 3 fails. Routing now converges to a \emph{different} stable state, where AS 29518 prefers the customer path through AS 31283 (right). When the link comes back up, BGP does not revert to the original stable state, and the system is stuck in an unintended routing outcome.

``BGP Wedgies''~\cite{wedgies} cause unpredictable network behavior that is difficult to debug.  (Sami~\etal \cite{sss} also showed that the existence of two stable states, as in Figure~\ref{fig:wedgies}, implies that persistent routing oscillations are possible.)

\subsubsection{Agreements imply convergence.}

In
\ifnum\full=1 Appendix~\ref{apx:converge} \else the full version \fi
we prove that when all ASes prioritize secure routes the same way, %a unique stable routing configuration exists and
convergence to a single stable state is guaranteed, \emph{regardless} of which ASes adopt S*BGP:

\begin{theorem} \label{thm:convergence}
S*BGP convergence to a unique stable routing state is guaranteed in all three S*BGP routing models even under partial S*BGP deployment.
\end{theorem}

\noindent
This holds even in the presence of the attack of Section~\ref{sec:attack:dets}, \cf \cite{LGSPODC12}.   This suggests a prescriptive guideline for S*BGP deployment: ASes should all prioritize security in the same way.  (See Section~\ref{sec:guidelines} for more guidelines.)  The reminder of this paper supposes that ASes follow this guideline.

%This motivates study of scenarios where all ASes prioritize security in the same way (as in our three models). Next, we ask: which model is the best?

\section{Threat Model}\label{sec:threat}

To quantify ``security'' in each of our three models, we first need to discuss what constitutes a routing attack.  We focus on a future scenario where RPKI and origin authentication are deployed, and the challenge is engineering global S*BGP adoption. We therefore disregard attacks that are prevented by origin authentication, \eg prefix- and subprefix-hijacks~\cite{BGPsurvey,PaulF,pakistan,china,AS7007} (when an attacker originates a prefix, or more specific subprefix, when not authorized to do so).
Instead, we focus on attacks that are effective even in the presence of origin authentication, as these are precisely the attacks that S*BGP is designed to prevent.

Previous studies on S*BGP security~\cite{BGPattack,JenYannis,adoptability} focused on the endgame scenario, where S*BGP is fully deployed, making the crucial assumption that \emph{any secure AS that learns an insecure route from one of its neighbors can safely ignore that route}. This assumption is invalid in the context of a partial deployment of S*BGP, where S*BGP coexists alongside BGP. In this setting, some destinations may only be reachable via insecure routes. Moreover, even a secure AS may prefer to use an insecure route for economic or performance reasons (as in our security \second or \third models).  Therefore, propagating a bogus AS path using legacy insecure BGP~\cite{DEFCON,BGPattack} (an attack that is effective against fully-deployed origin authentication) can \emph{also} work against some \emph{secure} ASes when S*BGP is partially deployed.

\subsection{The attack.}\label{sec:attack:dets}

We focus on the scenario where a single attacker AS $m$ attacks a single destination AS $d$; all ASes except $m$ use the policies in Section~\ref{sec:policies}.  The attacker $m$'s objective is to maximize the number of source ASes that send traffic to $m$, rather than $d$. This commonly-used objective function~\cite{PaulF,BGPattack,sigBGP} reflects $m$'s incentive to attract (and therefore tamper~/~eavesdrop~/~drop) traffic from as many source ASes as possible. (We deal with the fact that ASes can source different amounts of traffic~\cite{LabovitzSIGCOMM} in Sections~\ref{sec:partitions:robust},~\ref{sec:metric:CPs},~\ref{sec:T1suck:metric}.)

\myparab{Attacker's strategy.  } The attacker $m$ wants to convince ASes to route to $m$, instead of the legitimate destination AS $d$ that is authorized to originate the prefix under attack.  It will do this by sending bogus AS-path information using legacy BGP.
What AS path information should $m$ propagate? A straightforward extension of the results in~\cite{BGPattack} to our models shows it is NP-hard for $m$ to determine a bogus route to export to each neighbor that maximizes the number of source ASes it attracts. As such, we consider the arguably simplest, yet very disruptive~\cite{PaulF,BGPattack}, attack: the attacker, which is not actually a neighbor of the destination $d$, pretends to be directly connected to $d$.   Since there is no need to explicitly include IP prefixes in our models, this translates to a single attacker AS $m$ announcing the bogus AS-level path ``$m,d$'' using legacy BGP to \emph{all} its neighbor ASes. %\footnote{\label{foot:sbgp} This attack is equally effective against partially-deployed soBGP, S-BGP and BGPSEC.  With soBGP, the attacker claims to have an edge to $d$ that does not exist in the graph. With S-BGP or BGPSEC the attacker claims to have learned a path ``$m,d$'' that $d$ never announced. Hence, we will analyze all these under the umbrella of S*BGP.}.
 Since the path is announced via legacy BGP, recipient ASes will not validate it with S*BGP, and thus will not learn that it is bogus.
(This attack is equally effective against partially-deployed soBGP, S-BGP and BGPSEC.  With soBGP, the attacker claims to have an edge to $d$ that does not exist in the graph. With S-BGP or BGPSEC the attacker claims to have learned a path ``$m,d$'' that $d$ never announced.)

\subsection{Are secure ASes subject to attacks?}\label{sec:pda}

\begin{figure}
\begin{center}
  \includegraphics[width=1.1in]{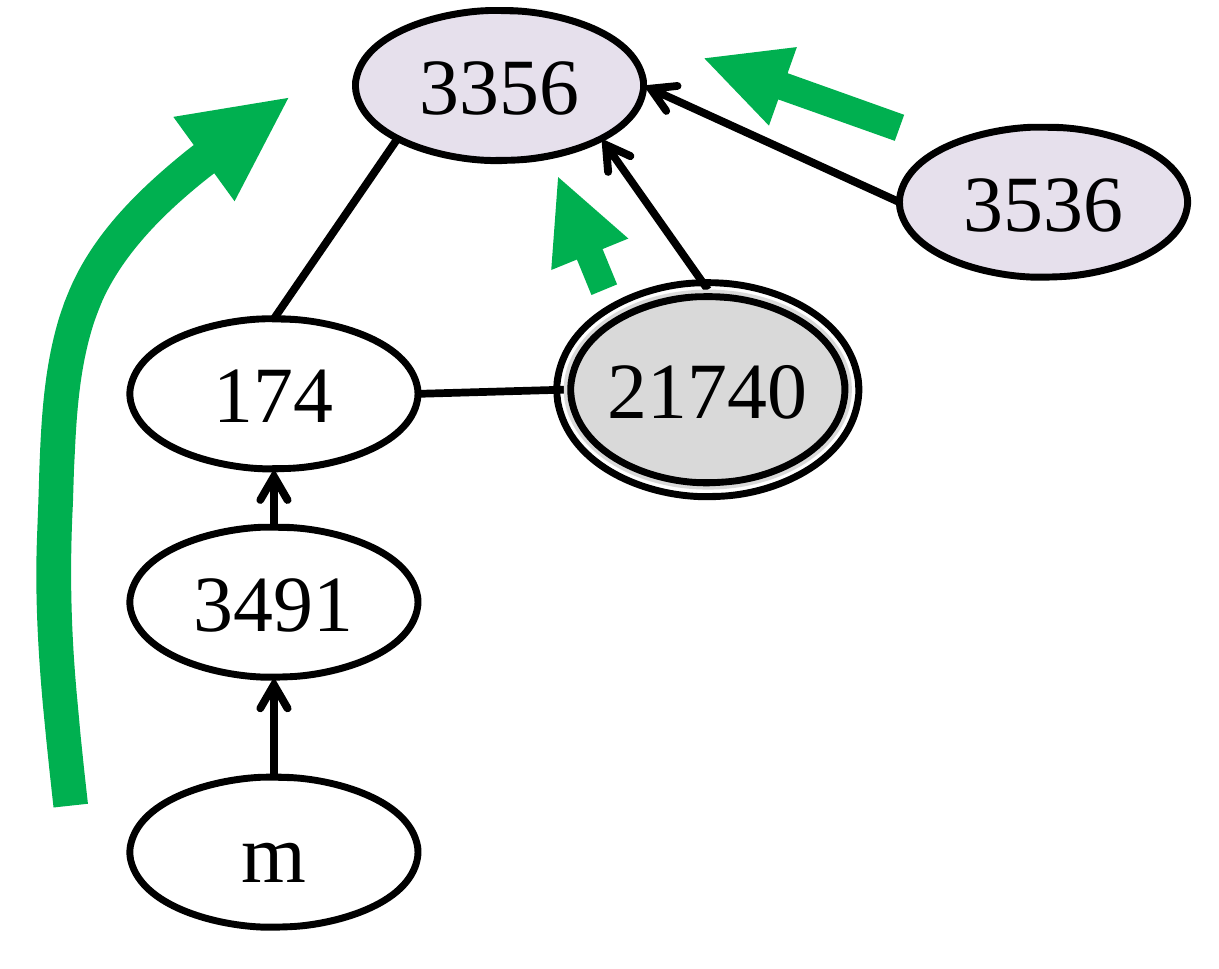}
  \includegraphics[width=1.1in]{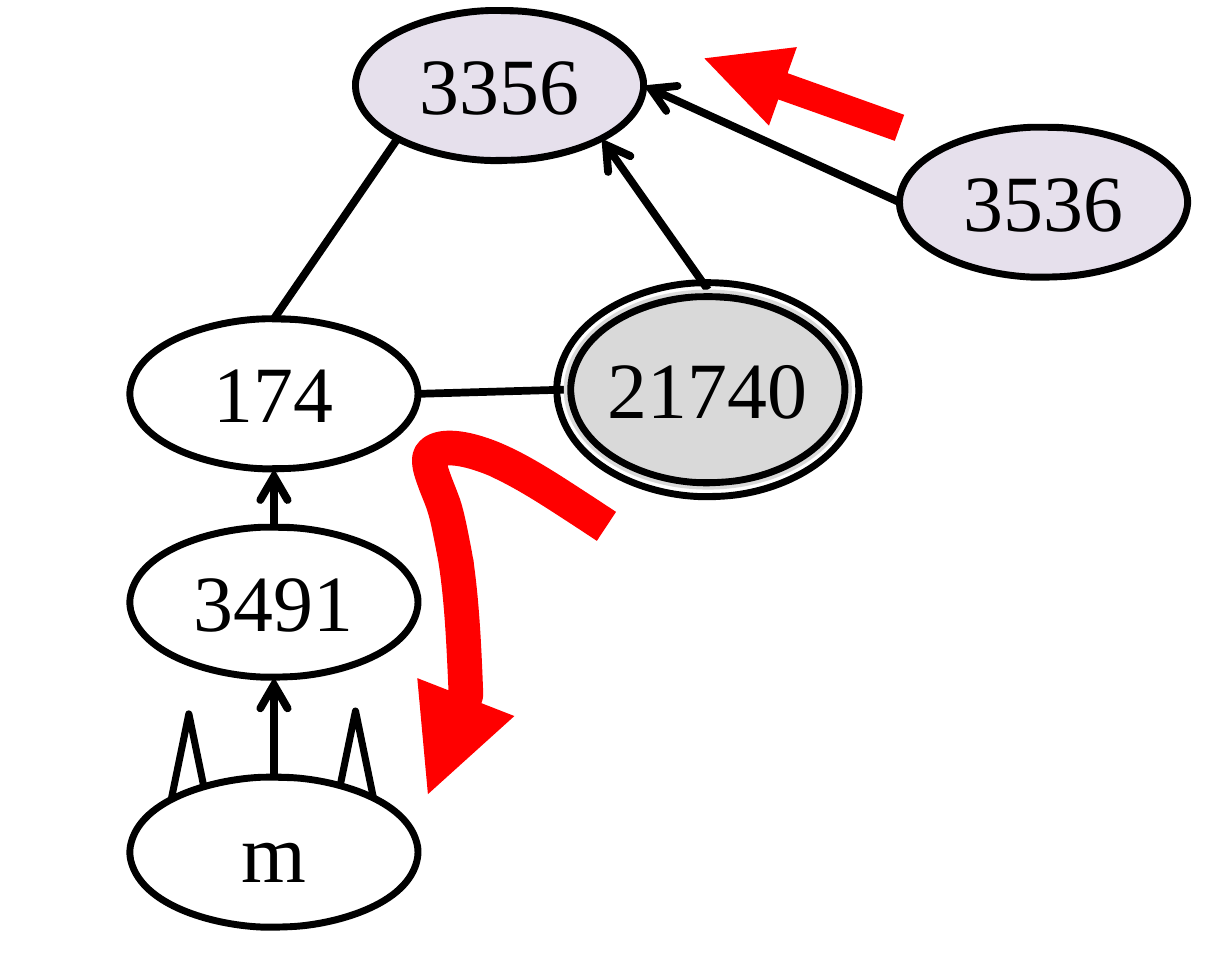}
  \vspace{-3mm}\caption{Protocol downgrade attack; Sec \second.}
  \vspace{-5mm}\label{fig:pda2nd}\end{center}
\end{figure}

%Naturally, an insecure AS is subject to the attack we described above.

Ideally, we would like a secure AS with a secure route to be protected from a routing attack. Unfortunately, however, this is  not always the case.  We now discuss a troubling aspect of S*BGP in partial deployment~\cite{BGPSECthreat}:

\myparab{Protocol downgrade attack.}  In a protocol downgrade attack, a source AS that uses a secure route to the legitimate destination under normal conditions, downgrades to an insecure bogus route \emph{during} an attack.

\smallskip\noindent
The best way to explain this is via an example:

\myparab{Figure~\ref{fig:pda2nd}.}  We show how AS 21740, a webhosting company, suffers a protocol downgrade attack, in the security \second  (or \third) model.  Under normal conditions (left), AS~21740 has a secure provider route directly to the destination Level 3 AS 3356, a Tier 1 ISP. (AS~21740 does \emph{not} have a peer route via AS 174 due to \Ex.)
During the attack (right), $m$ announces that it is directly connected to Level3, and so AS 21740 sees a bogus, insecure 4-hop peer route, via his peer AS 174.  Importantly, AS 21740 has no idea that this route is bogus; it looks just like any other route that might be announced with legacy BGP. In the security \second (and \third) model, AS 21740 prefers an insecure \emph{peer} route over a secure \emph{provider} route, and will therefore downgrade to the bogus route.% during the attack.

\ifnum\full=1
In Section~\ref{sec:T1suck:metric}, we show that protocol-downgrade attacks can be a serious problem, rendering even large partial  deployments of S*BGP ineffective against attacks.
\fi

\iffalse
\myparab{Damping insecure routes?}  Looking at the example, one might wonder why AS 21740 was so quick to drop his old secure route in favor of the shorter insecure route during the attack.  Indeed, one could apply an idea similar to route-flap damping to prevent such an attack, where AS 21740 might always prefer an \emph{old} secure route over an \emph{new} insecure route.  However, route-flap damping is notoriously problematic in its own right~\cite{flapping}. We do not expect network operators to be willing to introduce a complicated policy like ``insecure-route damping" in partial deployment of S*BGP, since insecure routes will commonly be learned alongside secure routes in normal conditions, \ie not due to a routing attack. % We therefore do not discuss this further.
\fi

\myparab{Downgrades are avoided in the security \first model.} Protocol downgrade attacks can happen in the security \second and \third models, but not when security is \first:

\begin{theorem}\label{thm:noPDA1st}
In the security \first model, for every attacker AS $m$, destination AS $d$, and AS $s$ that, in normal conditions, has a secure route to $d$ that does not go through $m$, $s$ will use a secure route to $d$ even during $m$'s attack.
\end{theorem}

\noindent
\ifnum\full=1
The proof is in Appendix~\ref{apx:downgrades}.
\fi
While the theorem holds only if the attacker $m$ is not on AS $s$'s route, this is not a severe restriction because, otherwise, $m$ would attract traffic from $s$ to $d$ even without attacking. %Hence, in the security \first model the attacker's best hope is to attract a large fraction of the ASes that cannot route to the destination along a secure route. 
\section{Invariants to Deployment}\label{sec:partitions}

Given the vast number of possible configurations for a partial deployment of S*BGP, we present a framework for exploring the security benefits of S*BGP vis-a-vis origin authentication, \emph{without making any assumptions about which ASes are secure}. To do this, we show how to quantify security (Section~\ref{sec:metric}), discuss how to determine an \emph{upper bound} on security available with \emph{any} S*BGP deployment for any routing model (Section~\ref{sec:doomedImm}),  finally compare it to the security available with origin authentication (Section~\ref{sec:baseline},~\ref{sec:upperLower}).

\subsection{Quantifying security: A metric.}\label{sec:metric}

We quantify improvements in ``security'' by determining the fraction of ASes that avoid attacks (per Section~\ref{sec:attack:dets}). The attacker's goal is to attract traffic from as many ASes as possible; our metric therefore measures the average fraction of ASes that do \emph{not} choose a route to the attacker.

\myparab{Metric.} Suppose the ASes in set $S$ are secure and consider an attacker $m$ that attacks a destination $d$.  Let $H(m,d,S)$ be the number of ``happy'' source ASes that choose a legitimate route to $d$ instead of a bogus route to $m$. (See Table~\ref{tab:glossary}).  Our metric is:
%
%\begin{small}
\begin{equation*}\label{eq:metric}
H_{M,D}(S)=\tfrac1{|D|(|M|-1)(|V|-2)}\sum_{m\in M}\sum_{d\in D\backslash\{m\}} H(m,d,S)
\end{equation*}
%\end{small}
%
\noindent
Since we cannot predict where an attack will come from, or which ASes it will target, the metric averages over all attackers in a set $M$ and destinations in a set $D$; we can choose $M$ and $D$ to be any subset of the ASes in the graph, depending on (i) where we expect attacks to come from, and (ii) which destinations we are particularly interested in protecting.  When we want to capture the idea that all destinations are of equal importance, we average over all destinations; note that ``China's 18 minute mystery" of 2010~\cite{china} fits into this framework well, since the hijacker targeted prefixes originated by a large number of (seemingly random) destination ASes.   However, we can also zoom in on important destinations $D$ (\eg content providers~\cite{LabovitzSIGCOMM,pakistan,moratel}) by averaging over those destinations only.
We can, analogously, zoom in on certain types of attackers $M$ by averaging over them only.
\ifnum\full=1
Averaging over fixed sets $D$ and $M$ (that are independent of $S$) also allows us to compare security across deployments $S$ and routing policy models.
\fi

\begin{table}
\begin{small}
\begin{tabular}{|l|p{2.5in}|}
  \hline
  happy & Chooses a legitimate secure/insecure route to $d$. \\
  \hline
  unhappy & Chooses a bogus insecure route to $m$. \\
\hline\hline
  immune & Happy \emph{regardless of which ASes are secure}. \\\hline
%  protected & A protectable source that is happy when some set of ASes $S$ is secure. \\\hline
  doomed & Unhappy \emph{regardless of which ASes are secure}. \\\hline
  protectable & Neither immune nor doomed.   \\\hline
\end{tabular}
\end{small}
\vspace{-3mm}
\caption{Status of source $s$ when $m$ attacks $d$.}\label{tab:glossary}
\vspace{-3mm}
\end{table}

\myparab{Tiebreaking \& bounds on the metric.}  Recall from Section~\ref{sec:policies} that our model fully determines an AS's routing decision up to the tiebreak step \TB of its routing policy.  Since computing $H_{M,D}(S)$ only requires us to distinguish between ``happy'' and ``unhappy'' ASes, the tiebreak step matters only when a source AS $s$ has to choose between (1) an \emph{insecure} route(s) to the legitimate destination $d$ (that makes it happy), and (2) an \emph{insecure}  bogus route(s) to $m$ (that makes it unhappy).  Importantly, $s$ has no idea which route is bogus and which is legitimate, as both of them are insecure.  Therefore, to avoid making uninformed guesses about how ASes choose between equally-good \emph{insecure} routes, we will compute upper and lower bounds on our metric; to get a lower bound, we assume that every AS $s$ in the aforementioned situation will always choose to be unhappy (\ie option (2)); the upper bound is obtained by assuming $s$ always chooses to be happy (\ie  (1)).
\ifnum\full=1
See also Appendix~\ref{apx:partitions}.
\fi

\ifnum\full=1
\myparab{Algorithms.} Our metric is determined by computing routing outcomes, each requiring time $O(|V|)$, over all possible $|M||D|$ attacker and destination pairs.  We sometimes take $M=D=V$
so that our computations approach $O(|V|^3)$; the parallel algorithms we developed for this purpose are presented in Appendix~\ref{apx:algos},~\ref{apx:sim}.
\fi

\subsection{Origin authentication gives good security.}\label{sec:baseline}

At this point, we could compute the metric for various S*BGP deployment scenarios, show that most source ASes are ``happy'', argue that S*BGP has improved security, and conclude our analysis.
This, however, would not give us the full picture, because it is possible that most of the happy ASes would have been happy \emph{even if S*BGP had not been deployed}.  Thus, to understand if the juice is worth the squeeze, we need to ask how many more attacks are prevented by a particular S*BGP deployment scenario, relative to those already prevented by RPKI with origin authentication.  More concretely, we need to compare the fraction of happy ASes \emph{before and after the ASes in $S$ deploy S*BGP}. To do this, we compare the metric for a deployment scenario $S$ against the ``baseline scenario'', where RPKI and origin authentication are in place, but no AS has adopted S*BGP, so that the set of secure ASes is $S=\emptyset$.

In \cite{BGPattack}, the authors evaluated the efficacy of origin authentication against attacks that it was not designed to prevent --- namely, the ``$m,d$'' attack of Section~\ref{sec:attack:dets}.  They randomly sampled pairs of attackers and destinations and plotted the distribution of the fraction of ``unhappy'' source ASes (ASes that route through the attacker, see Table~\ref{tab:glossary}). Figure 3 of \cite{BGPattack} shows that attacker is able to attract traffic from less than half of the source ASes in the AS graph, on average.
We now perform a computation and obtain a result that is similar in spirit; rather than randomly sampling pairs of attackers and destinations as in \cite{BGPattack}, we instead compute a \emph{lower bound} on our metric over all possible attackers and destinations. We find that  $H_{V,V}(\emptyset)\geq60\%$ on the basic UCLA graph, and $H_{V,V}(\emptyset)\geq62\%$ on our IXP-augmented graph.

It is striking that both our and \cite{BGPattack}'s result indicate more than half of the AS graph is \emph{already} happy even \emph{before} S*BGP is deployed. To understand why this is the case, recall that with origin authentication, an attacking AS $m$ must announce a bogus path ``$m,d$'' that is one hop longer than the path ``$d$'' announced by the legitimate destination AS $d$.  When we average over all $(m, d)$ pairs and all the source ASes, bogus paths through $m$ will appear longer, on average, than legitimate paths through $d$. Since path length plays an important role in route selection,  on average, more source ASes choose the legitimate route.

\begin{figure}
\begin{center}
    \includegraphics[width=2.3in]{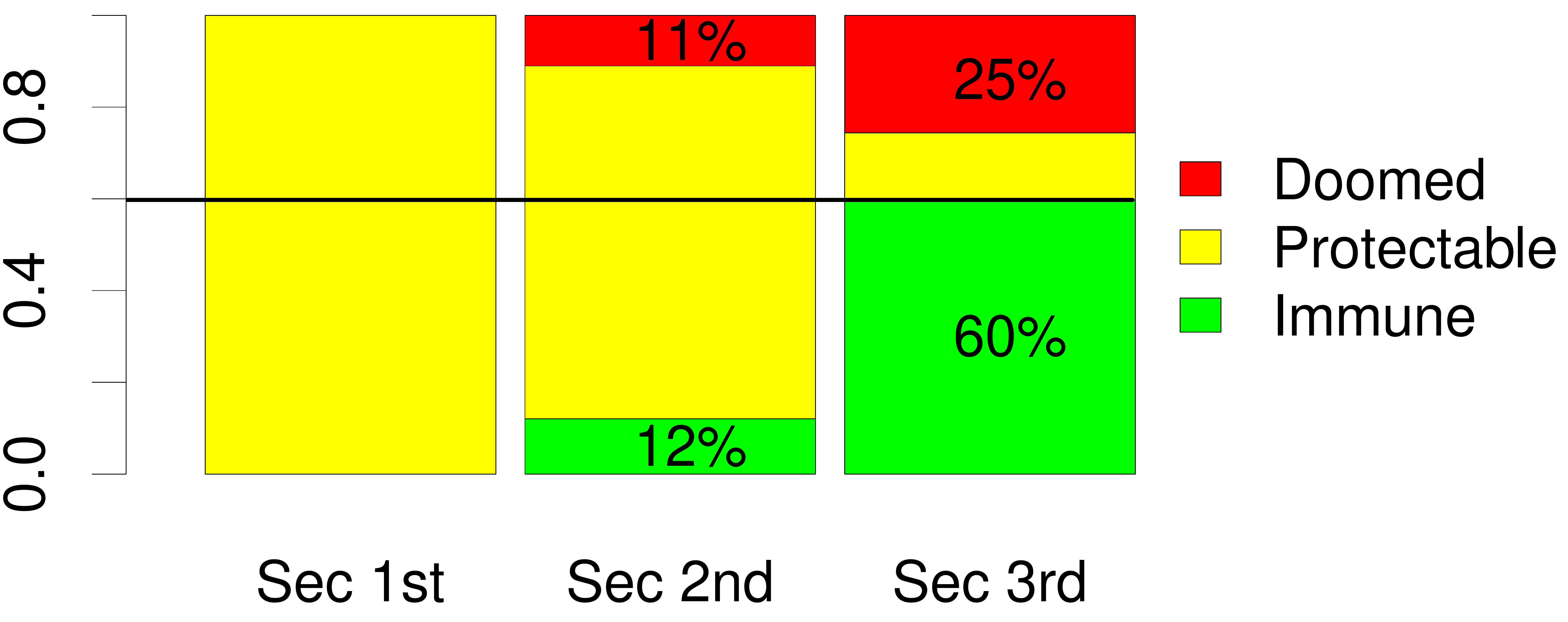}
    \vspace{-3mm}\caption{Partitions}
    \vspace{-7mm}\label{fig:partitions}\end{center}
\end{figure}

\subsection{Does S*BGP give better security?}

%Given that  many ASes already avoid attacks when origin authentication is in place,
How much further can we get with a partial deployment of S*BGP? We now obtain bounds on the improvements in security that are possible for a given routing policy model, but \emph{for any} set  $S$ of secure ASes.

We can obtain these bounds thanks to the following crucial observation: ASes can be partitioned  into three distinct categories with respect to each attacker-destination pair $(m,d)$.  Some ASes are \emph{doomed} to route through the attacker regardless of which ASes are secure.  Others  are \emph{immune} to the attack  regardless of which ASes are secure. Only the remaining ASes are \emph{protectable}, in the sense that whether or not they route through the attacker depends on which ASes are secure (see Table~\ref{tab:glossary}).

To bound our metric $H_{M,D}(S)$  for a given routing policy model (\ie security \first, \second, or \third) and \emph{across all partial-deployment scenarios} $S$, we first partition source ASes into categories ---  doomed, immune, and protectable ---  for each $(m,d)$ pair and each routing policy model.  By computing the average fraction of immune ASes across all $(m,d) \in M \times D$ for a given routing model, we get a lower bound on $H_{M,D}(S)$  $\forall S$ and that routing model. We similarly get an upper bound on $H_{M,D}(S)$ by computing the average fraction of ASes that are \emph{not} doomed.

\subsubsection{Partitions: Doomed, protectable \& immune.}\label{sec:doomedImm}

\noindent  We return to Figure~\ref{fig:pda2nd} to explain our partitioning:

\myparab{Doomed.}  A source AS $s$ is \emph{doomed} with respect to pair $(m,d)$ if $s$  routes through $m$ no matter which set $S$ of ASes is secure.   AS 174 %and AS 21740
in Figure~\ref{fig:pda2nd} is doomed when security is \second (or \third).  If security is \second (or \third), AS 174 \emph{always} prefers the bogus customer route to the attacker over a (possibly secure) peer path to the destination AS 3356,  for every $S$.%; as a result, AS 21074 will always prefer the bogus peer route to AS 174 than its provider route to the destination.

\myparab{Immune.} A source AS $s$ is \emph{immune} with respect to pair $(m,d)$ if $s$ will route through $d$ no matter which set $S$ of ASes is secure.  AS 3536 in Figure~\ref{fig:pda2nd} is one example; this single-homed stub customer of the destination AS 3356 can \emph{never} learn a bogus route in any of our security models.  When security is \second or \third, another example of an immune AS is AS 10310 in Figure~\ref{fig:bensNdamages2nd}; its customer route to the legitimate destination AS 40426 is always more attractive than its provider route to the attacker in these models.

\myparab{Protectable.} AS $s$ is protectable  with respect to pair $(m,d)$ if it can either choose the legitimate route to $d$, or the bogus one to $m$, depending on $S$.  With security \first, AS~174 in Figure~\ref{fig:pda2nd} becomes protectable.  If it has a secure route to the destination AS 3356,  AS~174 will choose it and be happy; if not, it will choose the bogus route to $m$.

\subsubsection{Which ASes are protectable?}

The intuition behind the following partitioning of ASes is straightforward.  The subtleties involved in proving that an AS is doomed/immune are discussed in
\ifnum\full=1
Appendix~\ref{apx:partitions}.
\else
the full version.
\fi

\mypara{Security \first.} Here, we suppose that all ASes are protectable; the few exceptions (\eg the single-homed stub of Figure~\ref{fig:pda2nd}) have little impact on the count of protectable ASes.

\mypara{Security \second.} Here, an AS is doomed if it has a route to the attacker with better local preference \LP than every available route to the legitimate destination; (\eg the bogus \emph{customer} route offered to AS 174 in Figure~\ref{fig:pda2nd} has higher \LP than the legitimate \emph{peer} route).  An immune AS has a route to the destination that has higher \LP than every route to the attacker.  For protectable AS, its best available routes to the attacker and destination have \emph{exactly the same} \LP.

\mypara{Security \third.} Here, a doomed AS has a path to $m$ with (1) better \LP OR (2) equal \LP and shorter length \SP, than every available path to $d$. The opposite holds for an immune AS. A protectable AS has best available routes to $m$ and $d$ with equal \LP \emph{and} path length \SP.

\subsection{Bounding security for all deployments.}\label{sec:upperLower}

For each routing model, we found the fraction of doomed/ protectable / immune source ASes for each attacker destination pair $(m,d)$, and took the average over all $(m,d)\in V\times V$.  We used these values to get upper- and lower bounds on $H_{V,V}(S)$ \emph{for all deployments $S$}, for each routing model.

\myparab{Figure~\ref{fig:partitions}:}   The colored parts of each bar represent the average fraction of immune, protectable, and doomed source ASes, averaged over all $O(|V|^2)$ possible pairs of attackers and destinations.  Since $H_{V,V}(S)$ is an average of the fraction of happy source ASes over all pairs of attackers and destinations, the upper bound on the metric $H_{V,V}(S)$ $\forall S$ is the average fraction of source ASes that are \emph{not} doomed.
The upper bound on the metric $H_{V,V}(S)$ $\forall S$ is therefore: $\approx 100\%$ with security \first, $89\%$ with security \second, and $75\%$ with security \third.  (The same figure computed on our IXP-edge-augmented graph looks almost exactly the same, with the proportions being $\approx 100\%$, $90\%$ and $77\%$.) %Recall, however, that our real interest is to compare S*BGP to the baseline of origin authentication only, where $S=\emptyset$.
Meanwhile, the heavy solid line is the lower bound on the metric $H_{V,V}(\emptyset)$ in the baseline setting where $S=\emptyset$ and  there is only origin authentication; in Section~\ref{sec:baseline} we found that $H_{V,V}(\emptyset)=60\%$ (and $62\%$ for the IXP-edge-augmented graph). Therefore, we can bound the maximum change in our security metric $H_{V,V}(S)$ $\forall S$ for each routing policy model by computing the distance between the solid line and the boundary between the fraction of doomed and protectable ASes.  We find:

\myparab{Security \third: Little improvement.} Figure~\ref{fig:partitions} shows that the maximum gains over origin authentication that are provided by the security \third model are quite slim --- at most 15\% --- \emph{regardless} of which ASes are secure.  (This follows because the upper bound on the metric $H_{V,V}(S)\leq 75\%$ for any $S$ while the lower bound on the baseline setting is $H_{V,V}(\emptyset)\geq 60\%$.) Moreover, these are the \emph{maximum} gains $\forall S$; in a realistic S*BGP deployment, the gains are likely to be much smaller.  This result is disappointing, since the security \third model is likely to be the most preferred by network operators (Section~\ref{sec:survey}), but it is not especially surprising.  S*BGP is designed to prevent path shortening attacks; however, in the security \third model ASes prefer short (possibly bogus) insecure routes over a long secure routes, so it is natural that this model realizes only minimal  security benefits.

%\subsubsection{Security \second: More improvement?}\label{sec:partitions:2}

\myparab{Security \second: More improvement.}  Meanwhile,  route security is prioritized above route length with the security~\second model, so we could hope for better security benefits.  Indeed, Figure~\ref{fig:partitions} confirms that the \emph{maximum gains} over origin authentication are better: $89- 60 =29\%$.  But can these gains be realized in realistic partial-deployment scenarios?
\ifnum\full=1
We answer this in question in Section~\ref{sec:results:metric}.
\fi

\myparab{Decreasing numbers of immune ASes? } The fraction of immune ASes in the security \second (12\%) and \first ($\approx 0\%$) models is (strangely) lower than the fraction of happy ASes in the baseline scenario (60\%). How is this possible?  In Section~\ref{sec:Colldamages} we explain this counterintuitive observation by showing that \emph{more} secure ASes can sometimes result in \emph{less} happy ASes; these ``collateral damages'', that occur only in the security \first and \second models, account for the decrease in the number of immune ASes.

 %This negative side-effect of partially-deployed S*BGP is also an issue when security is \first (where, as we can see from Figure~\ref{fig:partitions}, there are almost no immune ASes at all!).

\subsection{Robustness to destination tier.}\label{sec:partitions:robust}

Thus far, we have been averaging our results over all possible attacker-destination pairs in the graph. However, some destination ASes might be particularly important to secure, perhaps because they source important content (\eg the content provider ASes (CPs)) or transit large volumes of traffic (the Tier 1 ASes). As such, we broke down the metric over destinations in each  \emph{tier} in Table~\ref{tab:tiers}.

\myparab{Figure~\ref{fig:partitions:dest:sl}.  } We show the partitioning into immune~/~protectable~/~doomed ASes %(the three colored bars)
in the security \third model, but this time averaged individually over all destinations in each tier, and all possible attackers $V$.  The thick horizontal line over each vertical bar again shows the corresponding lower bound on our metric $H_{V,\text{Tier}}(\emptyset)$ when no AS is secure.  Apart from the Tier 1s (discussed next), we observe similar trends as in Section~\ref{sec:upperLower}, with the improvement in security ranging from $8-15\%$ for all tiers;
\ifnum\full=1
the same holds for the security \second model, shown in \textbf{Figure~\ref{fig:partitions:dest:ss}}.
\fi
\ifnum\full=0
the same holds for the security \second model. (Figure in full version).
\fi

\begin{figure}
\begin{center}
\includegraphics[width=2.7in]{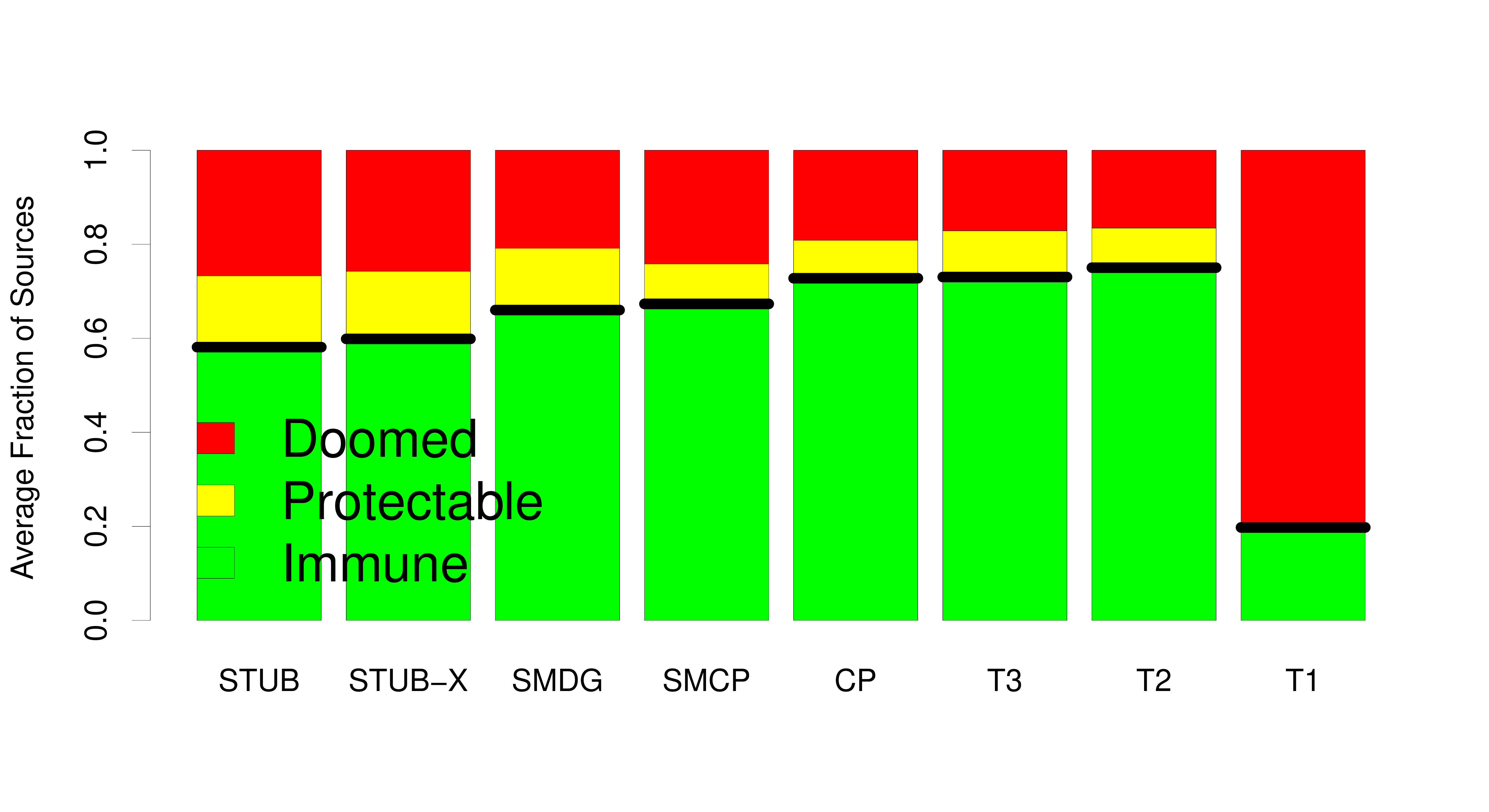}
    \vspace{-1mm}\caption{Partitions by destination tier. Sec~\third.}
    \vspace{-5mm}\label{fig:partitions:dest:sl}\end{center}
\end{figure}

\subsection{It's difficult to protect Tier 1 destinations.}\label{sec:T1suck}

Strangely enough, Figure~\ref{fig:partitions:dest:sl} shows that when Tier 1 destinations are attacked in the security \third model, the vast majority ($\approx 80\%$) of ASes are doomed, and only a tiny fraction are protectable; the same holds when security is \second
\ifnum\full=0
(not shown).
\else
(Figure~\ref{fig:partitions:dest:ss}).
\fi
Therefore, in these models, S*BGP can do little to blunt attacks on Tier 1 destinations.

How can it be that Tier 1s, the largest and best connected (at least in terms of customer-provider edges) ASes in our AS graph, are the most vulnerable to attacks? Ironically, it is the Tier 1s' very connectivity that harms their security. Because the Tier 1s are so well-connected, they can charge most of their neighbors for Internet service.  As a result, most ASes reach the Tier 1s via costly provider paths that are the least preferred type of path according to the \LP step in our routing policy models.  Meanwhile, it turns out that when a Tier 1 destination is attacked, most source ASes will learn a bogus path to the attacker that is \emph{not} through a provider, and is therefore preferred over the (possibly secure) provider route to the T1 destination in the security \second or \third models.
In fact, this is exactly what lead to the protocol downgrade attack on the Tier 1 destination AS 3356 in Figure~\ref{fig:pda2nd}.
%
%There is nothing that can be done about this in the security \second or \third models, since routes with better local preference \LP are always preferred over secure routes.
We will later (Section~\ref{sec:T1suck:metric}) find that this is a serious hurdle to protecting Tier 1 destinations.

%In Section~\ref{sec:T1suck:metric} we will find that this is a serious hurdle when we consider securing Tier 1s as early S*BGP adopters.

%%%%%%%%%%%%%%%%%%%%%%%%%%%%%%%%%%%%%%%%%%%%%%%%%%%

\ifnum\full=1
\subsection{Which attackers cause the most damage?}\label{sec:partitions:attacker}
Next, we break things down by the type of the attacker, to get a sense of type of attackers that S*BGP is best equipped to defend against.

\myparab{Figure~\ref{fig:partitions:attacker:sl}.}   We bucket our counts of doomed, protectable, and immune ASes for the security \third model by the attacker type in Table~\ref{tab:tiers}, for all $|V|^2$ possible attacker-destination pairs. As the degree of the attacker increases, it's attack becomes more effective; the number of immune ASes steadily decreases, and the number of doomed ASes correspondingly increases, as the the tier of the attacker grows from stub to Tier 2.  Meanwhile, the number of protectable ASes remains roughly constant across tiers. The striking exception to this trend is that the the Tier 1 attacker is significantly less effective than even the lowest degree (stub) attackers.     While at this observation might seem unnatural at first, there is a perfectly reasonable explanation: when a Tier 1 attacks, its bogus route will look like a provider route from the perspective of most other source ASes in the graph.  Because the \LP step of our routing model depreferences provider routes relative to peer and customer routes, the Tier 1 attacker's bogus route will be less attractive than any legitimate route through a peer or provider, and as such most ASes will be immune to the attack. The same observations hold when security is \second.

\begin{figure}
\begin{center}
\includegraphics[width=2.7in]{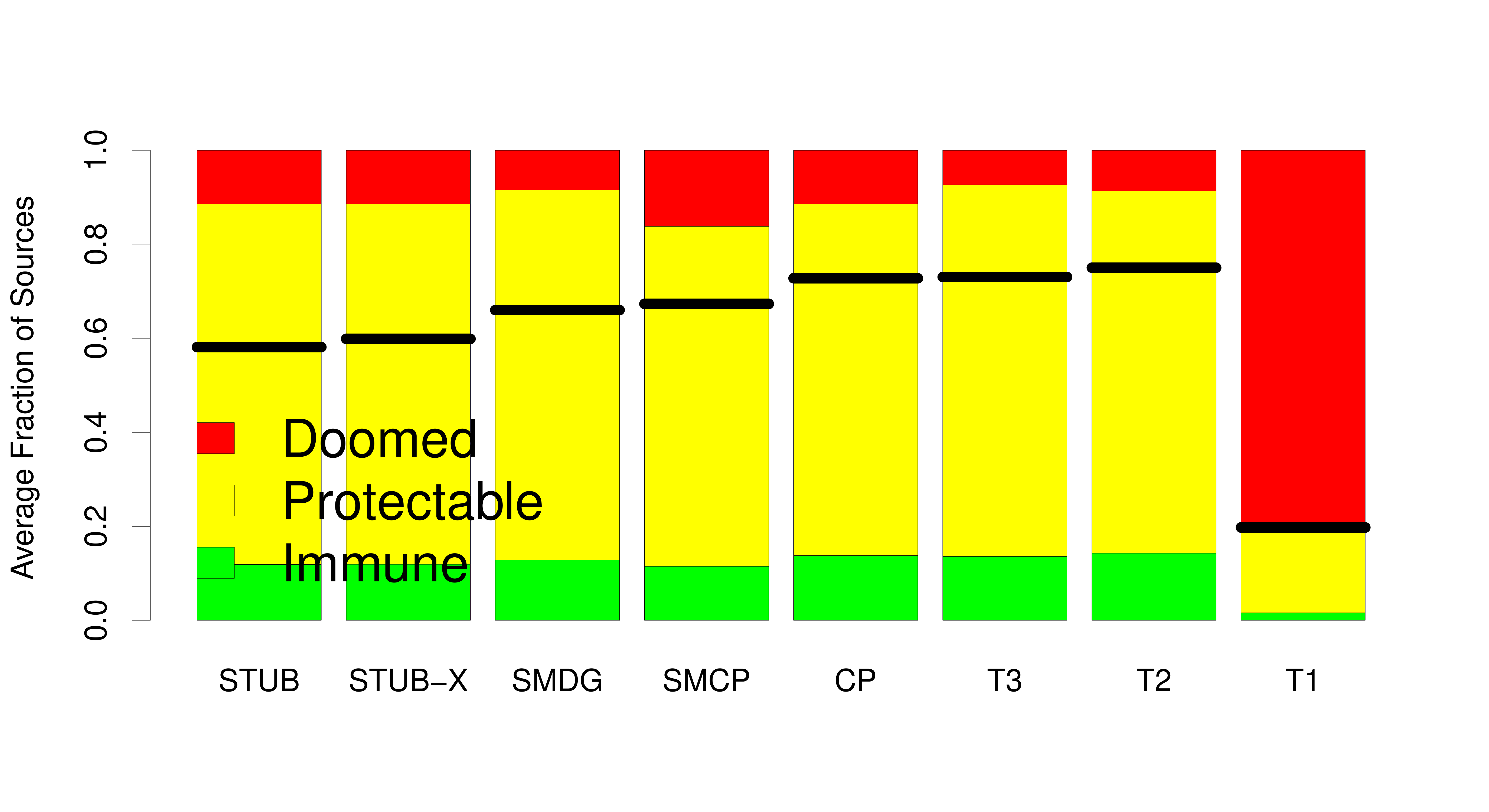}
    \vspace{-3mm}\caption{Partitions by destination tier. Sec \second.}
    \vspace{-5mm}\label{fig:partitions:dest:ss}\end{center}
\end{figure}

\begin{figure}
\begin{center}
    \includegraphics[width=2.7in]{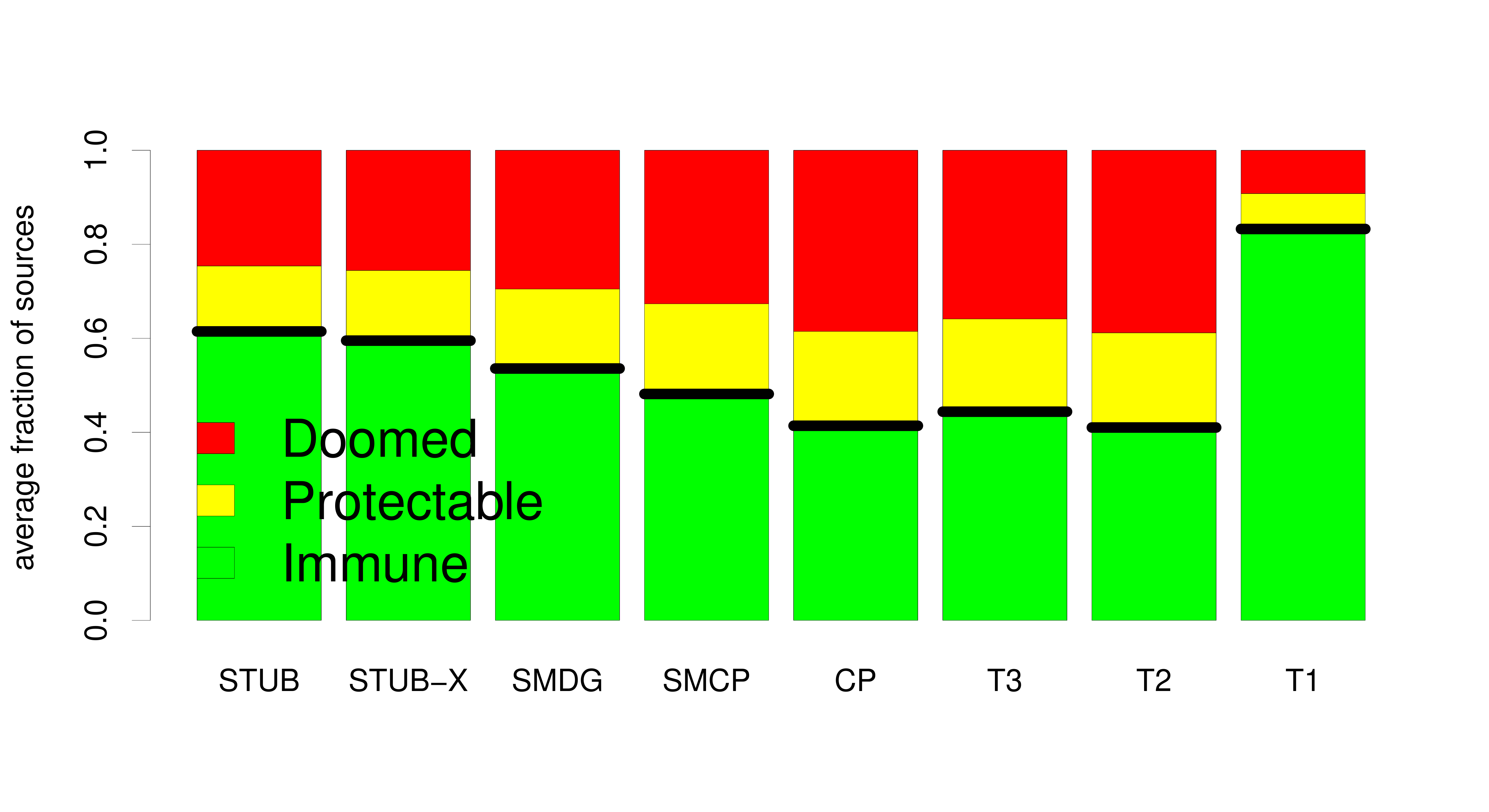}
    \vspace{-3mm}\caption{Partitions by attacker tier. Sec \third.}
    \vspace{-5mm}\label{fig:partitions:attacker:sl}\end{center}
\end{figure}

\myparab{Tier 1s can still be protected as sources.}  However, before we completely give up on the Tier 1s obtaining any benefit from S*BGP, we reproduced Figures~\ref{fig:partitions:dest:sl}~-~\ref{fig:partitions:dest:ss}
but this time, bucketing the results by the tier of source. (Figure omitted.)  We found that each source tier, \emph{including} the Tier 1s, has roughly the same average number of doomed (25\%), immune (60\%), and protectable (15\%) ASes.  It follows that, while S*BGP cannot protect Tier 1 destinations from attack, S*BGP still has the potential to prevent a Tier 1 sources from choosing a bogus route.

\myparab{Robustness of results.}  We repeated this analysis on our IXP-augmented graph (Appendix~\ref{apx:ixp}) and using different routing policies (Appendix~\ref{apx:stubborn_partitions}). Please see the appendices for details.

\fi

\iffalse  %%% version with old figure
\myparab{Doomed.}  A source AS $s$ is \emph{doomed} for attacker-destination pair $(m,d)$ if $s$  route through $m$ no matter which set $S$ of ASes is secure.   AS 3491 in Figure~\ref{fig:pda3rd} is doomed when security is \third. AS 3491 it will \emph{always} prefer the short bogus route to the attacker over the longer route to destination, regardless if which ASes are secure.  Even if he has a secure route under normal conditions, he will still downgrade to an insecure route during an attack.

\myparab{Immune.} A source AS $s$ is \emph{immune} for attacker-destination pair $(m,d)$ because $s$ will route through $d$ no matter which set $S$ of ASes is secure.  AS 10310 in Figure~\ref{fig:pda3rd} is an example of an immune AS when security is \third, because his customer route to the legitimate destination AS 40426 is always more attractive than his provider route to the attacker.

\myparab{Protectable.} A protectable AS $s$ for attacker destination pair $(m,d)$ can either choose the legitimate route to t$d$, or the bogus one to $m$, depending on the set $S$ of secure ASes.  Finally, observe that when security is \second or \first, AS 3491 in Figure~\ref{fig:pda3rd} becomes protectable; if AS 3491 has a secure route to the destination, he will chose it and be happy, and if not, he will chose the bogus route to the attacker.
\fi 
\section{Deployment Scenarios}\label{sec:results:metric}\label{sec:ixpedge}

In Section~\ref{sec:upperLower} we presented upper bounds on the improvements in security from S*BGP deployment for choice of secure ASes $S$. We found that while only meagre improvements over origin authentication are possible in the security \third model, better results are possible in the security \second and \first models.  However, achieving the bounds in Section~\ref{sec:upperLower} could require full S*BGP deployment at every AS. What happens in more realistic deployment scenarios?  First, we find that the security \second model often behaves disappointingly like the security \third model.  We also find that  Tier 1 destinations remain most vulnerable to attacks when security is \second or \third. We conclude the section by presenting prescriptive guidelines for  partial S*BGP deployment.

\myparab{Robustness to missing IXP edges.}
\ifnum\full=0
In the full version, we repeat the analysis in Section~\ref{sec:bigDeps}-\ref{sec:guidelines} on IXP-edge-augmented AS graph, and see almost identical trends.
\else
We repeated the analysis in Section~\ref{sec:bigDeps}-\ref{sec:guidelines} over the AS graph augmented with IXP peering edges and saw almost identical trends. We see a slightly higher baseline of happy ASes when $S=\emptyset$ (Section~\ref{sec:upperLower}), which almost always causes the improvement in the metric (over the baseline scenario) to be slightly smaller for this graph. (Plots in Appendix~\ref{apx:ixp}.)
\fi

\subsection{It's hard to decide whom to secure.}\label{sec:hard}

We first need to decide which ASes to secure. Ideally, we could choose the smallest set of ASes that maximizes the value of the metric. To formalize this, consider the following computational problem, that we call ``Max-k-Security'': Given an AS graph, % $G=(V,E)$,
a specific attacker-destination pair $(m,d)$, and a parameter $k>0$, find a set $S$ of secure ASes of size $k$ that maximizes the total number of happy ASes. Then:

\begin{theorem}\label{thm:hard}
Max-k-Security is NP-hard in all three routing policy models.
\end{theorem}
The proof is
\ifnum\full=0
in the full version.
\else
in Appendix~\ref{apx:hardness_results}.
\fi
This result can be extended to the problem of choosing the set of secure ASes that maximize the number of happy ASes over \emph{multiple} attacker-destination pairs (which is what our metric computes).

\subsection{Large partial deployments.}\label{sec:bigDeps}

\begin{figure}
 %\begin{center}
 \subfigure[]{\includegraphics[width=.243\textwidth]{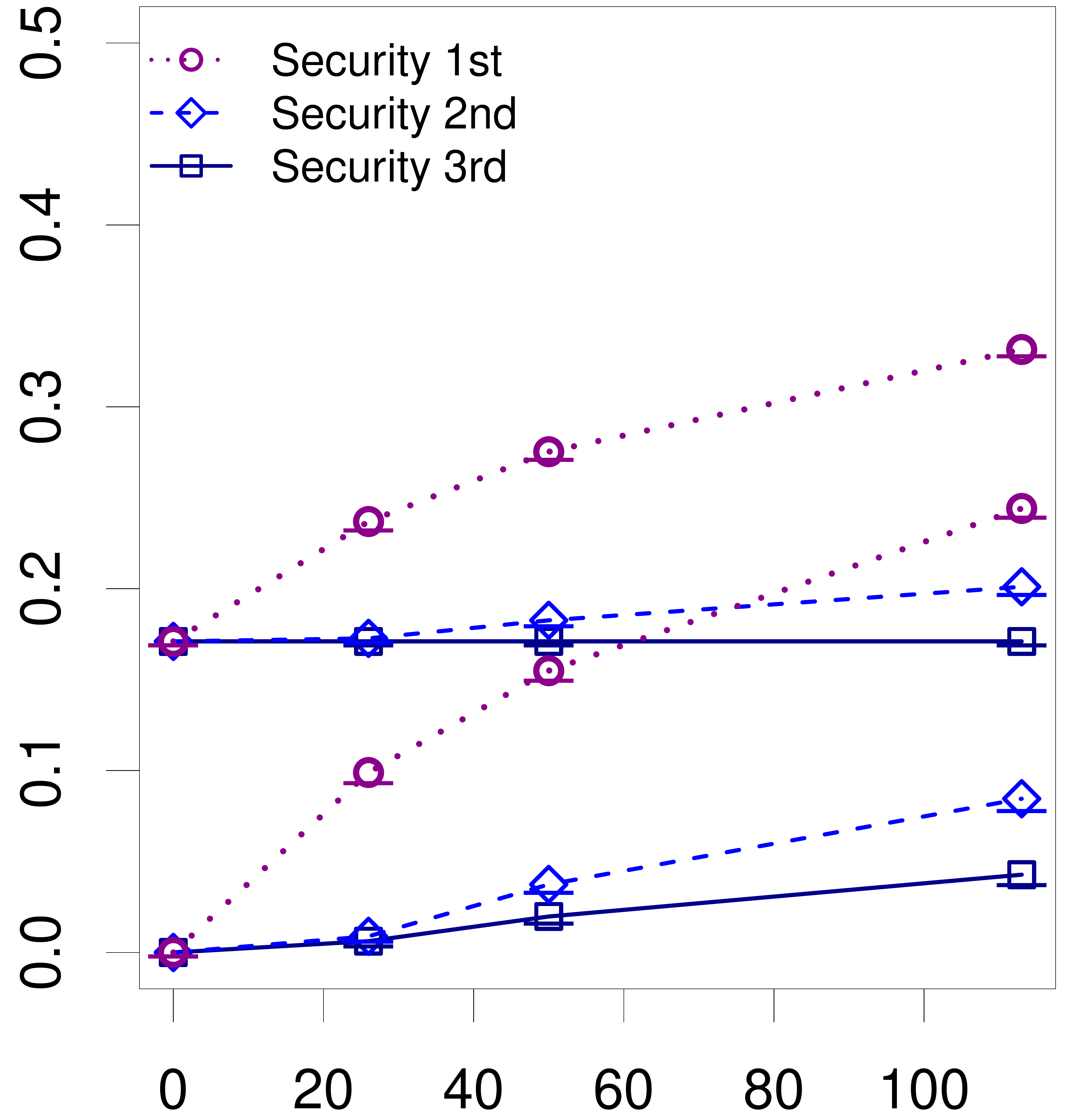} \label{fig:metric:bigdeps}}
 \subfigure[]{\includegraphics[width=.214\textwidth]{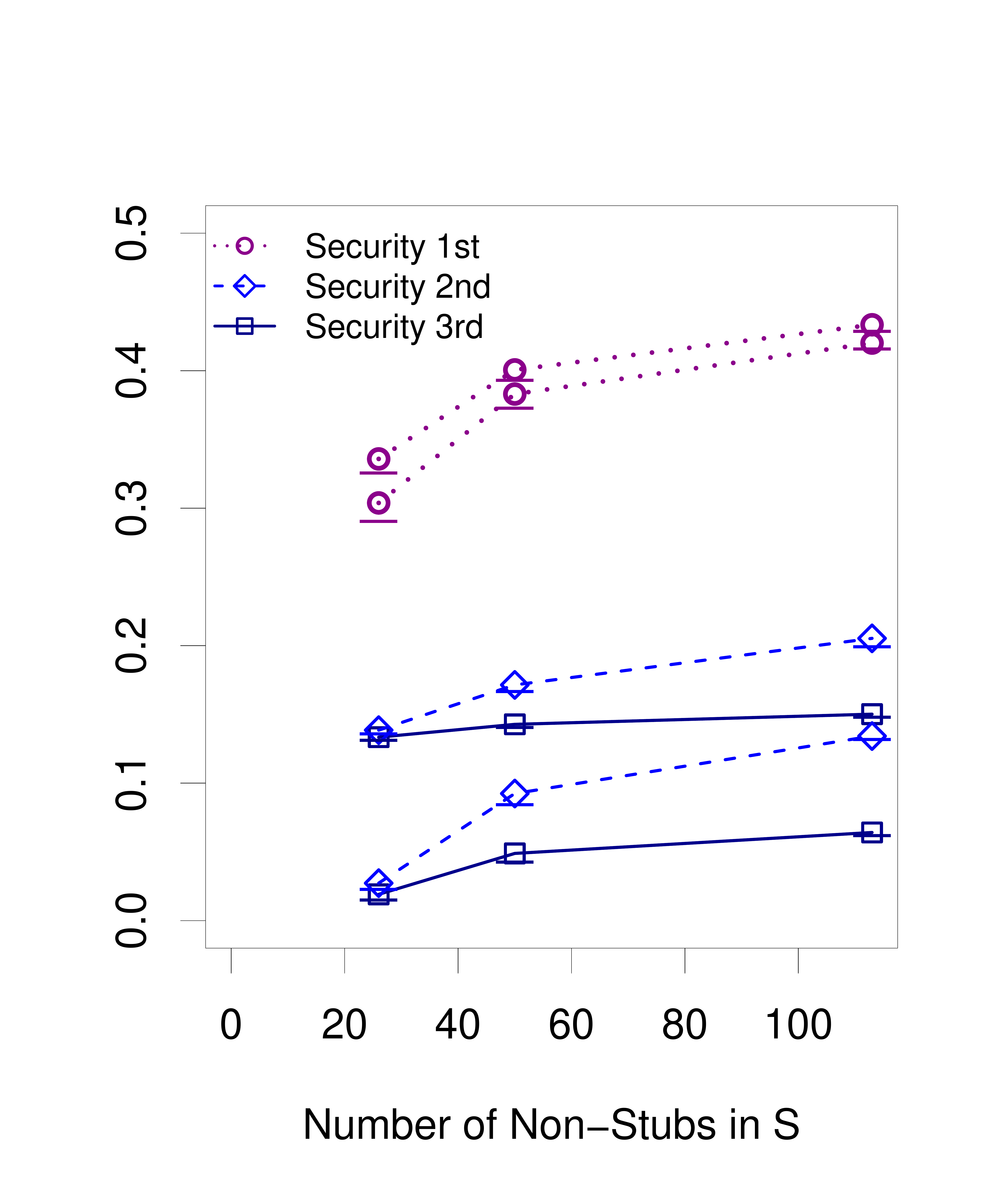} \label{fig:metric:Sonly}}
 \vspace{-5mm}
 \caption{Tier 1+2 rollout: For each step $S$ in rollout, upper and lower bounds on \subref{fig:metric:bigdeps} $H_{M',V}(S) - H_{M',V}(\emptyset)$ and \subref{fig:metric:Sonly} $H_{M',V}(S)-H_{M',d}(\emptyset)$ averaged over all $d \in S$. The $x$-axis is the number of non-stub ASes in $S$.  The ``error bars'' are explained in Section~\ref{sec:simplex}.}
  \vspace{-3mm}
 %\end{center}
 \end{figure}

Instead of focusing on choosing the optimum set $S$ of ASes to secure (an intractable feat), we will instead consider a  few partial deployment scenarios among high-degree ASes $S$, as suggested in practice~\cite{FCCcommit} and in the literature~\cite{adopt,JenYannis,adoptability}.%, and evaluate the extent to which they approach the upper bounds presented in Section~\ref{sec:upperLower}.

\myparab{Non-stub attackers.} We now suppose that the set of attackers is the set of non-stub ASe in our graph $M'$ (\ie not ``Stubs'' or ``Stubs-x'' per Table~\ref{tab:tiers}). Ruling out stub ASes is consistent with the idea that stubs cannot launch attacks if their providers perform prefix filtering~\cite{BGPattack,BGPsurvey}, a functionality that can be achieved via IRRs~\cite{IRRpt} or even the RPKI~\cite{roaRPSL}, and does not require S*BGP.

\subsubsection{Security across all destinations.}\label{sec:bigdep:all}

Gill~\etal~\cite{adopt} suggest bootstrapping S*BGP deployment by having secure ISPs deploy S*BGP in their customers that are stub ASes. %(Note: 85\% of the ASes in our AS graph are stubs.)
We therefore consider this ``rollout'':

\myparab{Tier 1 \& Tier 2 rollout. }  Other than the empty set, we consider three different secure sets.  We secure $X$ Tier 1's and $Y$ Tier 2's and all of their stubs, where {\small $(X,Y) \in\{(13,13),(13,37),(13,100)\}$}; this corresponds to securing about 33\%, 40\%, and 50\%  of the AS graph.   % With such large deployments, we would hope to see very large improvements in security.

\smallskip
\noindent
The results are shown in \textbf{Figure~\ref{fig:metric:bigdeps}}, which plots, for each routing policy model, the increase in the upper- and lower bound on $H_{M',V}(S)$ (Section~\ref{sec:metric}) for each set $S$ of secure ASes in the rollout ($y$-axis), versus the number of non-stub ASes in $S$ ($x$-axis). We make a few  important observations:

\myparab{Tiebreaking can seal an AS's fate.} Even with a large deployment of S*BGP, the improvement in security is highly dependent on the vagarities of the intradomain tiebreaking criteria used to decide between \emph{insecure} routes. (See also Section~\ref{sec:metric}'s discussion on tiebreaking.) Even when we secure 50\% of ASes in the security \first model (the last step of our rollout), there is still a gap of more than 10\% between the lower and upper bounds of our metric.  Thus, in a partial S*BGP deployment, there is a large fraction of ASes that are balanced on a knife's edge between an insecure legitimate route and an insecure bogus route; only the (unknown-to-us) intradomain routing policies of these ASes can save them from attack. This is inherent to any partial deployment of S*BGP, even in the security~\first model.

\myparab{Meagre improvements even when security is \second.}  As expected, the biggest improvements come in the security \first model, where ASes make security their highest priority and deprecate all economic and operational considerations.  When security is \first and 50\% of the AS graph is secure (at the last step in the rollout), the improvement over the baseline scenario is significant; about $24\%$.
While we might hope that the security \second model would present improvements that are similar to those achieved when security is \first, this is unfortunately not the case. In both the security \second and \third models we see similarly disappointing increases in our metric. %, with improvements of only $8.4\%$ and  $4.2\%$, respectively, at the end of the rollout. %Indeed, a common thread to most of our results is that security \second behaves much like the security \third model.
We explain this observation in  Section~\ref{sec:postmortem}.

\subsubsection{Focus on the content providers?}\label{sec:metric:CPs}

Since much of the Internet's traffic originates at the content providers (CPs), we might consider the impact of S*BGP deployment on CPs only.
We considered the same rollout as above, but with all 17 CPs secure, and computed the metric \emph{over CP destinations only}, \ie $H_{M',CP}(S)$.
\ifnum\full=0
Results, shown in the full version, were similar to Figure~\ref{fig:metric:bigdeps}.
\else
The results, presented in Figure~\ref{fig:metric:bigdeps:CP}, are very similar to those in  Figure~\ref{fig:metric:bigdeps}:  improvements of at least $26\%$ $9.4\%$, and $4\%$  for security \first, \second, and \third  respectively.  We note, however, that CP destinations have a higher fraction of happy sources than other destinations on average, (see Figure \ref{fig:partitions:dest:sl}).
\fi

\ifnum\full=1
\begin{figure}
    \begin{center}
      \includegraphics[width=2in]{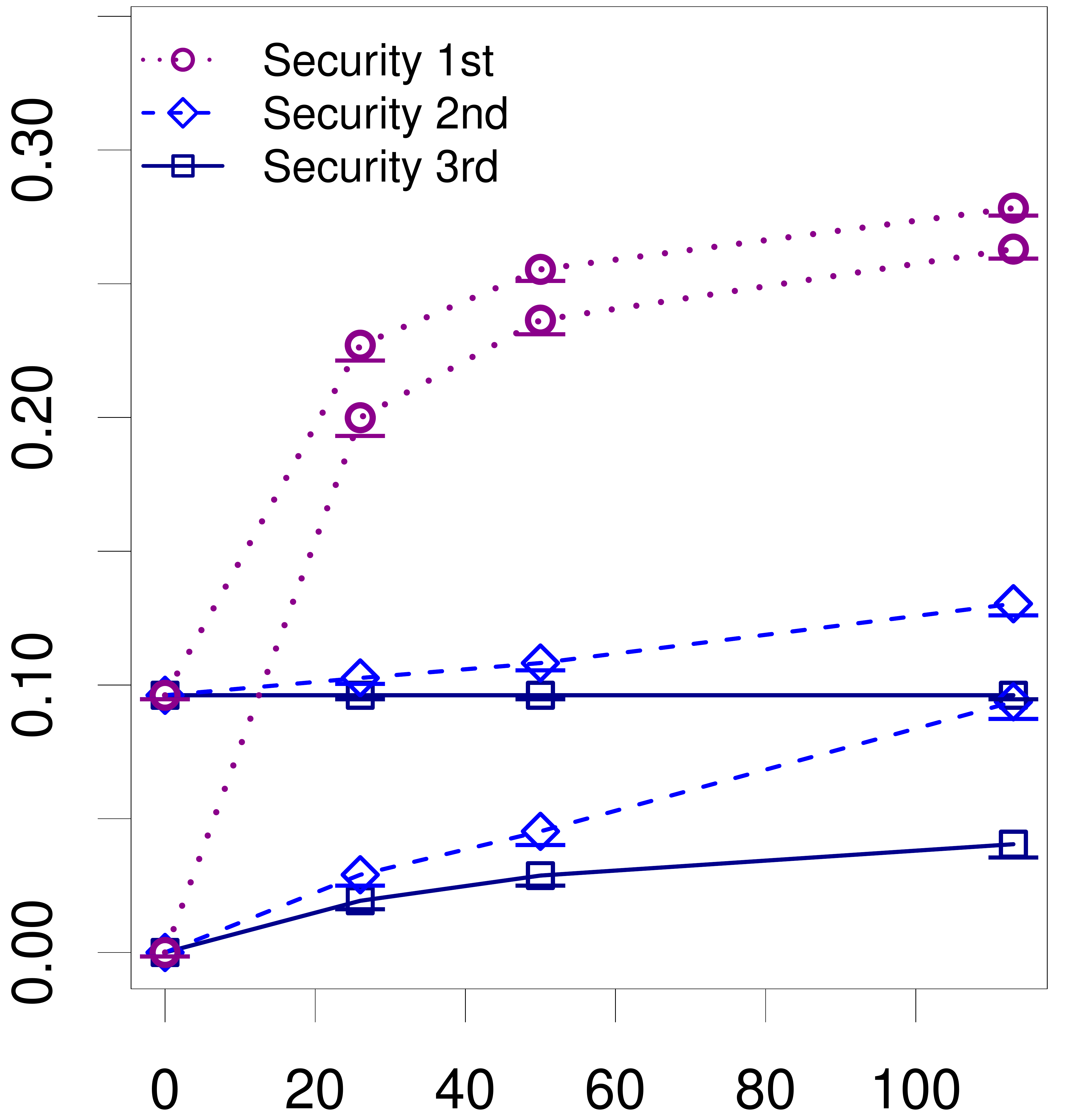} \vspace{-3mm}
      \caption{Tier 1+2+CP rollout: $H_{M',CP}(S) - H_{M',C}(\emptyset)$ for each step in the rollout. The $x$-axis is the number of non-stub, non-CP ASes in $S$.}
    \vspace{-5mm}\label{fig:metric:bigdeps:CP}
    \end{center}
\end{figure}
\fi

\subsubsection{Different destinations see different benefits.}\label{sec:uneven}\label{sec:islands}

%\begin{figure}
 %\includegraphics[width= ]{ }\\
 % \caption{The value of $H_{M',d}(S)-H_{M',d}(\emptyset)$ averaged over all $d \in S$ for each step in T1s + T2 rollout.}\label{fig:metric:Sonly}
 %\end{figure}
Thus far, we have looked at the impact of S*BGP in aggregate across all destinations $d \in V$ (or $d\in CP$).  Because secure routes can only exist to secure destinations, we now look at the impact of S*BGP on \emph{individual secure destinations} $d\in S$, by considering $H_{M',d}(S)$.

\myparab{Figure~\ref{fig:metric:Sonly}.} We plot the upper and lower bounds on the \emph{change} in the metric, \ie $H_{M',d}(S)-H_{M',d}(\emptyset)$, averaged across \emph{secure destinations only}, \ie $d \in S$.
As expected, we find large improvements when security is \first, and small improvements when security is \third.  Interestingly, however, when security is \second the metric does increase by $13-20\%$ by the last step in the rollout; while this is still significantly smaller than what is possible when security is \first, it does suggest that at least some secure destinations benefit more when security is \second, rather than \third.

\smallskip\noindent
For more insight, we zoom in on this last step in our rollout:% (where 13 Tier 1s, 100 Tier 2s, and all of their stubs are secure ($>50\%$ of ASes)):

\begin{figure}
 \includegraphics[width=.45\textwidth]{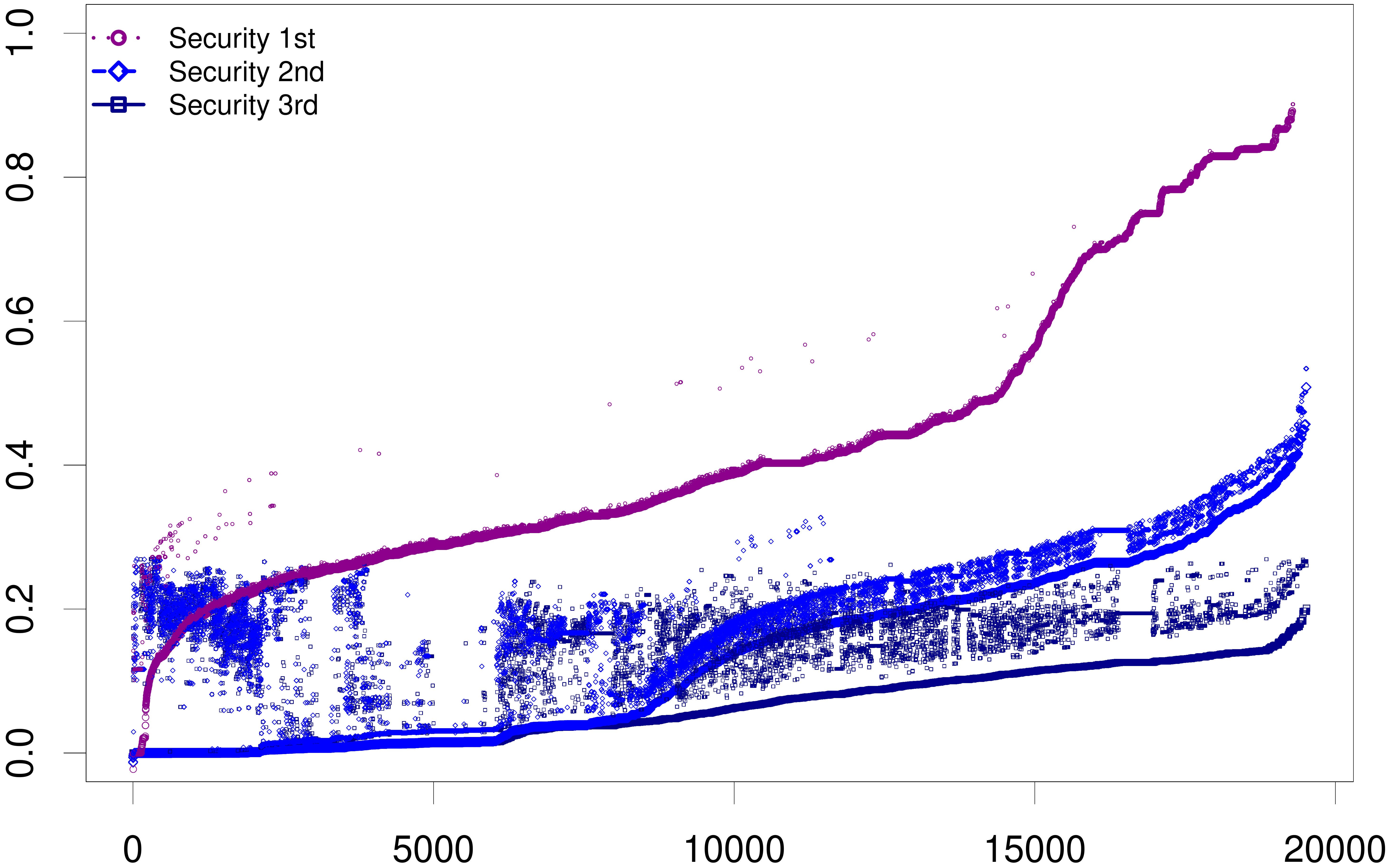}
 \vspace{-2mm}
 \caption{Non-decreasing sequence of $H_{M',d}(S)-H_{M',d}(\emptyset)$ $\forall$ $d \in S$. $S$ is all T1s, T2s, and their stubs.}\label{fig:seq}
 \vspace{-3mm}
 \end{figure}

\myparab{Figure~\ref{fig:seq}.}  For the last step in our rollout, we plot upper and lower bounds on the \emph{change} in the metric, \ie $H_{M',d}(S)-H_{M',d}(\emptyset)$, for each \emph{individual} secure destination $d\in S$.  For each of our three models, the lower bound for each $d\in S$ is plotted as a non-decreasing sequence; these are the three ``smooth'' lines. The corresponding upper bound for each $d\in S$ was plotted as well. For security \first, the upper and lower bounds are almost identical, and for security \second and \third, the upper bounds are the ``clouds'' that hover over the lower bounds. A few observations:

\myparab{Security \first provides excellent protection. } We find that when security is \first, a \emph{secure} destination can reap the full benefits of S*BGP even in (a large) partial deployment. To see this, we computed the \emph{true value} of $H_{M',d}(S)$  for all secure destinations $d\in S$, and found that it was between $96.8-97.9\%$ on average (across all $d\in S$).

\myparab{Security \second and \third are similar for many destinations.}  Figure~\ref{fig:seq} also reveals that many destinations obtain roughly the same benefits from S*BGP when security is \second as when security is \third.  Indeed, $93 \%$ of 7500 secure destinations that see $<4\%$ (lower-bound) improvement in Figure~\ref{fig:seq} when security is \third, do the same when security is \second as well.
What is the reason for this?  There are certain types of protocol downgrade attacks that succeed \emph{both} when security is \second and when security is \third (\ie when the bogus path has better \LP than the legitimate path, see \eg Figure~\ref{fig:pda2nd}). In Section~\ref{sec:postmortem} we shall show that protocol downgrade attacks are the most significant reason for the metric to degrade; therefore, for destinations where these ``\LP-based'' protocol downgrade attacks are most common, the security \second model looks much like the security \third model.

\myparab{Tier 1s do best when security is \first, and worst when it is \second or \third.} %Figure~\ref{fig:seq} reveals that, when security is \first, some secure destinations experience impressive improvements in their security metric $H_{M',d}(S)$, relative to the baseline setting $H_{M',d}(\emptyset)$.
When security is \first, our data also shows that the secure destinations that obtain the largest ($>40\%$) increases in their security metric $H_{M',d}(S)$ (relative to the baseline setting $H_{M',d}(\emptyset)$) include: (a) all 13 Tier 1s, and (b) $\geq 99\%$ of  ``Tier 1 stub'' destinations (\ie stub ASes such that all their providers are Tier 1 ASes).
On the other hand, these same destinations experience the \emph{worst} improvements when security is \second or \third (\ie a lower bound of $<3\%$).

To explain this, recall from Section~\ref{sec:T1suck} that when security is \second or \third, most source ASes that want to reach a Tier 1 destination are \emph{doomed}, because of protocol downgrade attacks like the one shown in Figure~\ref{fig:pda2nd}.  This explains the meagre benefits these destinations obtain when security is \second or \third.
On the other hand, protocol downgrade attacks fail when security is \first. Therefore, in the security \first model, the Tier 1 destinations (and by extension, Tier 1 stub destinations) obtain excellent security when S*BGP is partially deployed; moreover, they see most significant gains simply because they were so highly vulnerable to attacks in the absence of S*BGP (Figure~\ref{fig:partitions:dest:sl}, Section~\ref{sec:T1suck}).

\myparab{Security \second helps some secure destinations.}  Finally, when security is \second, about half of the secure destinations $d\in S$ see benefits that are discernibly better than what is possible when security is \third, though not quite as impressive as those when security is \first.  These destinations include some Tier 2s and their stubs, but never any Tier 1s.

\noindent Similar observations hold for earlier steps in the rollout.

\ifnum\full=1

\subsubsection{What happens when the Tier 1s are not secure?}\label{sec:results:T2:nonstub}

\begin{figure}
 \includegraphics[width=.45\textwidth]{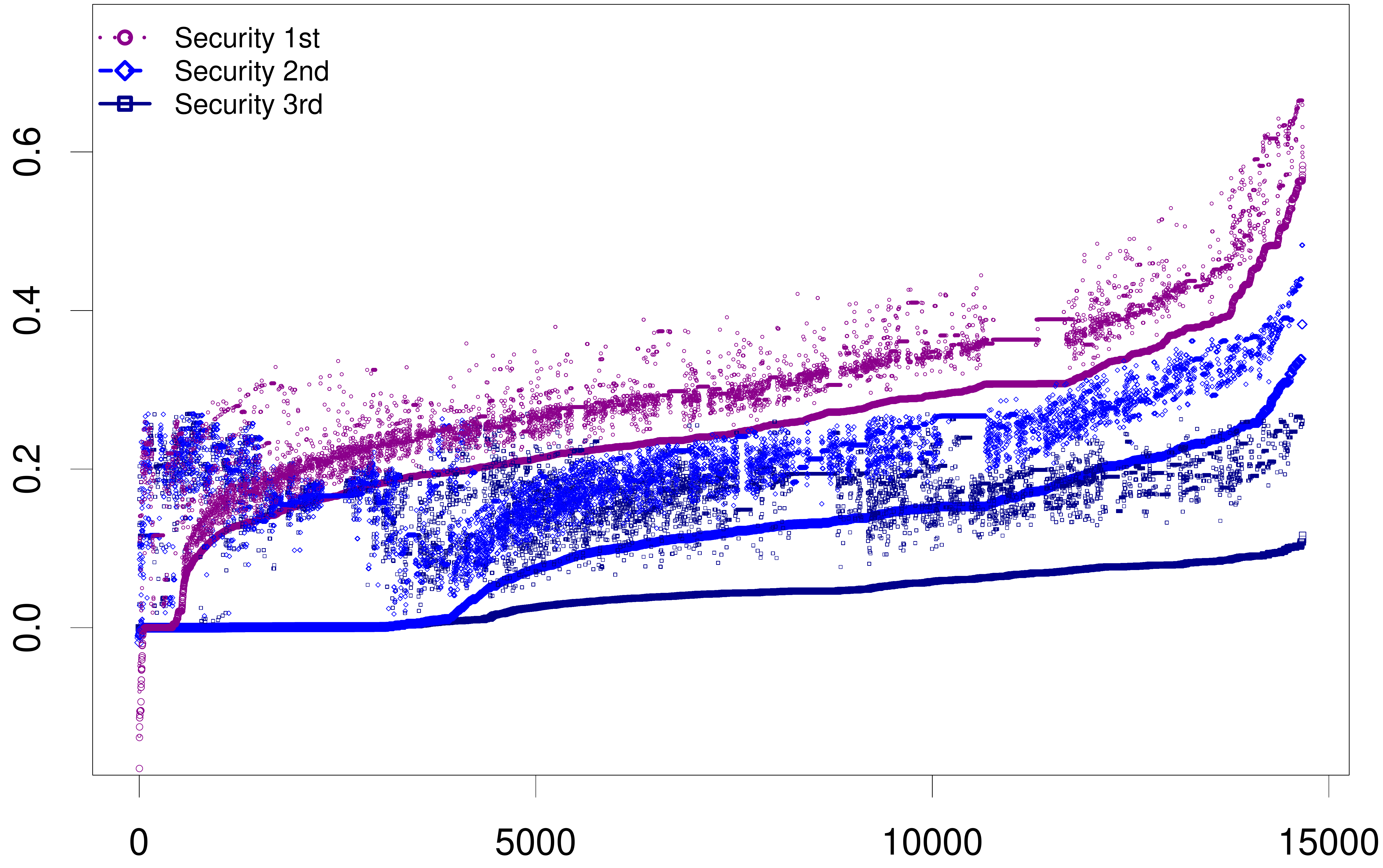}
 \vspace{-2mm}
 \caption{Non-decreasing sequence of $H_{M',d}(S)-H_{M',d}(\emptyset)$ $\forall$ $d \in S$. $S$ is all T2s, and their stubs.}\label{fig:seq:T2}
 \vspace{-3mm}
 \end{figure}

 \begin{figure}
    \begin{center}
      \includegraphics[width=2in]{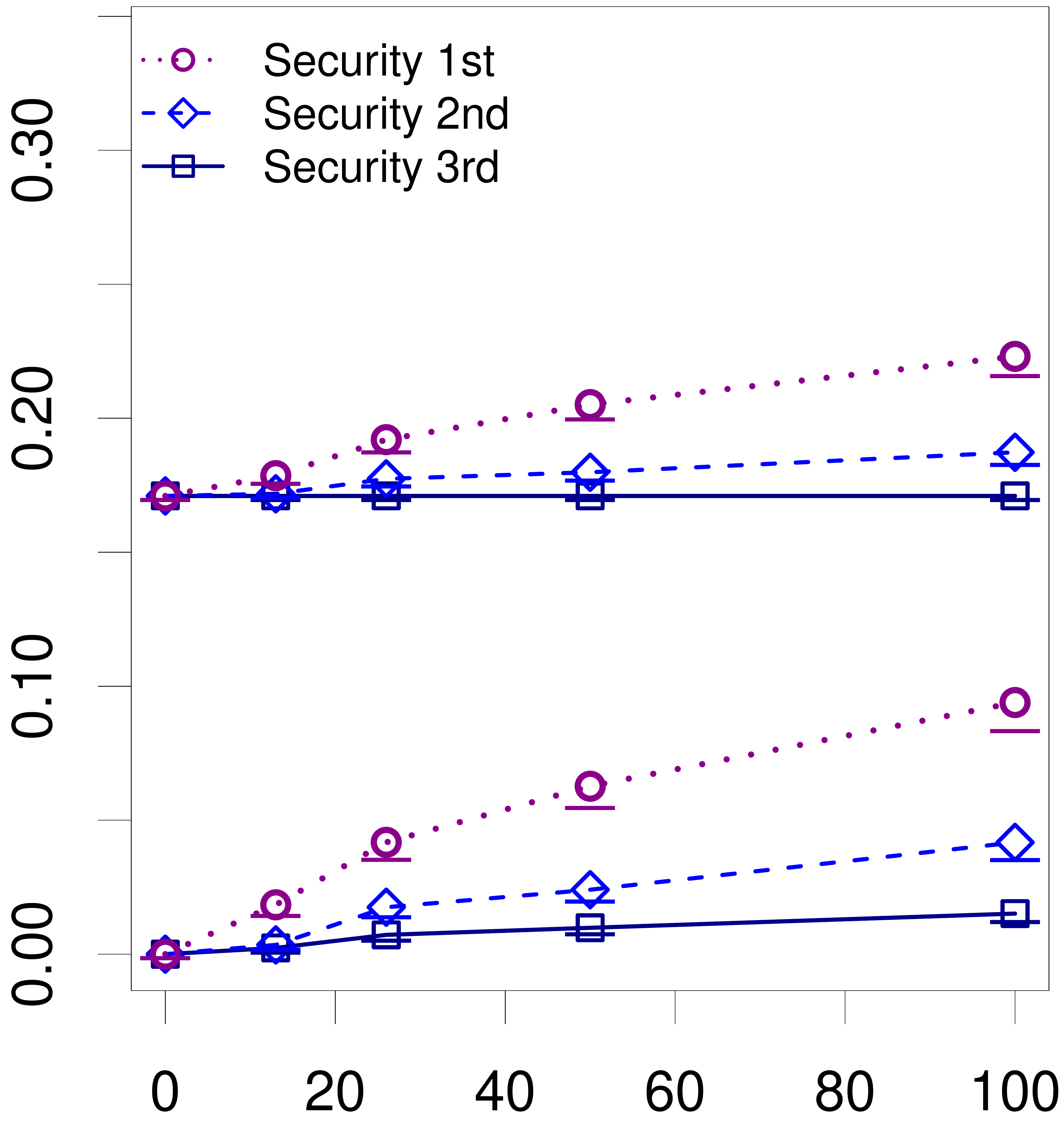} \vspace{-3mm}
      \caption{Tier 2 rollout: $H_{M',D}(S) - H_{M',D}(\emptyset)$ for each step in the T2 rollout. The $x$-axis is the number of non-stub, non-CP ASes in $S$.}
    \vspace{-5mm}\label{fig:metric:bigdeps:T2}
    \end{center}
\end{figure}

The results of the previous section motivate considering a deployment that excludes securing the Tier 1 ISPs.

\myparab{Secure just the Tier 2s?} We reproduce the analysis of Section~\ref{sec:bigdep:all} and Section~\ref{sec:uneven} with a rollout among only the Tier 2s and theirs stubs.  There are 100 Tier 2 ISPs in our AS graph (Table~\ref{tab:tiers}), and our Tier 2 rollout secures $Y$ Tier 2 ASes, and all of their stubs, where {\small $Y \in \{13, 26, 50, 100 \}$}; this amounts to securing about 18\%, 24\%, 30\%, and 38\% of ASes.

The results in \textbf{Figure~\ref{fig:metric:bigdeps:T2}} are similar to those in Figure~\ref{fig:metric:bigdeps}, except that the metric grows even more slowly, and  we see smaller improvements when security is \first.  This is consistent with our earlier observation (Section~\ref{sec:islands}) that the most dramatic improvements observed when security is \first are for Tier 1 destinations; the improvements for Tier 2 destinations and their stubs are much smaller when security is \first.   This causes the gap between the security \second and \first models to become smaller for the Tier 2 rollout  (relative to the Tier 1+2 rollout of Section~\ref{sec:uneven}); this can be observed from \textbf{Figure~\ref{fig:seq:T2}} which reproduces the results of Figure~\ref{fig:seq} for the last step of the Tier 2 rollout.  However, the gap between security \second and \first is smaller not only because Tier 2s see bigger improvements when security is \second model; this is also because they see worse improvements when security is \first.

 \begin{figure}
 \includegraphics[width=.45\textwidth]{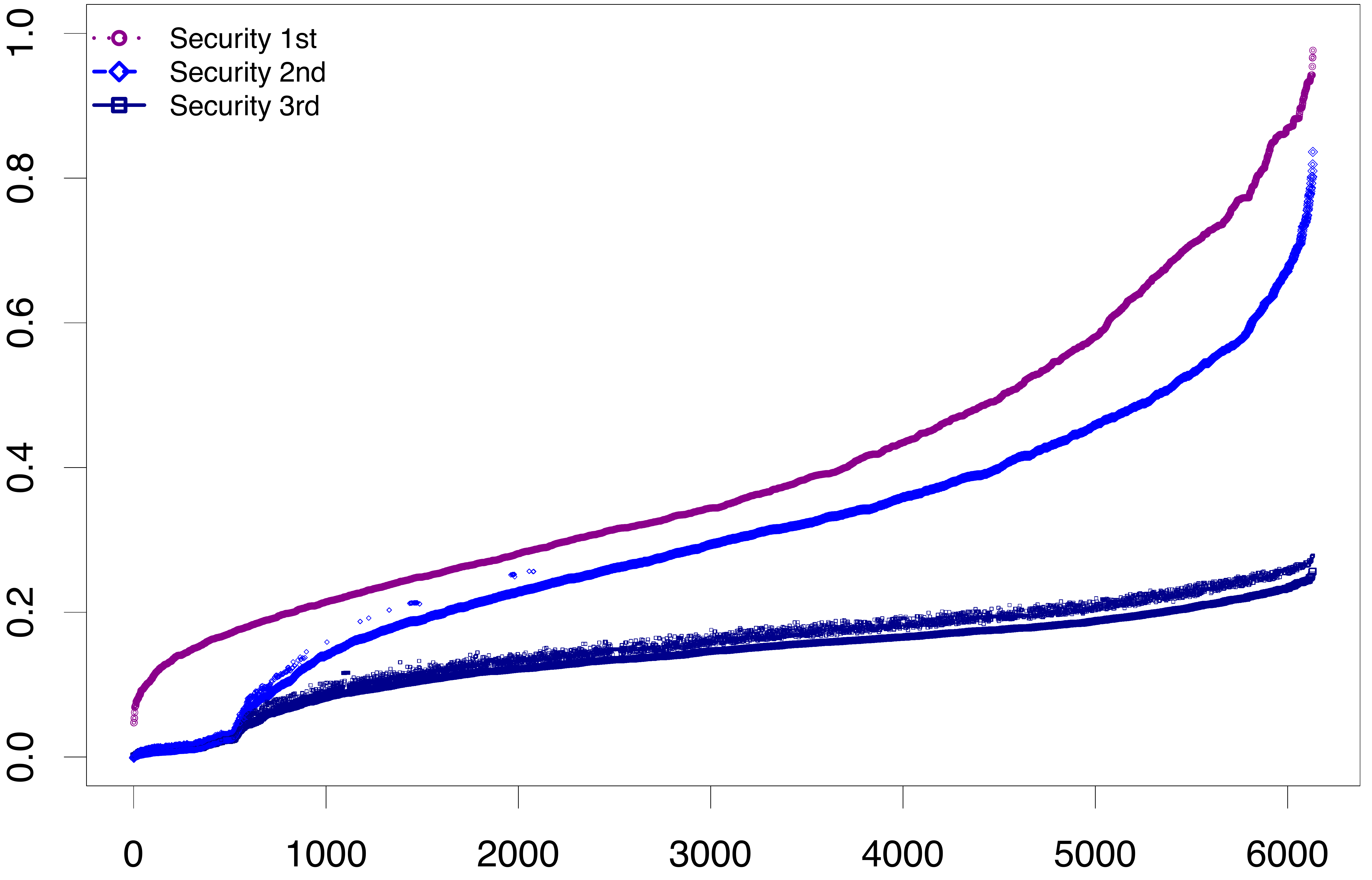}
 \vspace{-2mm}
 \caption{Non-decreasing sequence of $H_{M',d}(S)-H_{M',d}(\emptyset)$ $\forall$ $d \in S$. $S$ is all non stubs.}\label{fig:seq:nonstub}
 \vspace{-3mm}
 \end{figure}

\myparab{Secure just the nonstubs?} Finally, we consider securing \emph{only non-stub ASes} (\ie 6178 ASes or 15\% of the AS graph). We see a $6.2\%$, $4.7\%$ and $2.2\%$ worst-case improvement in the metric $H_{M',D}(S)$ when security is \first, \second, and \third respectively; this scenario therefore is similar to last step in our Tier 2 rollout, with exception that the gap between the security \second and \first model is even smaller. This is corroborated by \textbf{Figure~\ref{fig:seq:nonstub}}, which reproduces the results of Figure~\ref{fig:seq} for the scenario where only non-stub ASes are secure.  We see that the benefits available when security is \second almost reach those that are possible when security is \first.

\mypara{Summary.} Taken together, our suggest that in the security \first model, destinations that are Tier 1s or their stubs see the largest improvements in security. In such cases, the security \second model behaves much like the security \third model.  However, in cases where Tier 1s and their stubs are not secure, the gap between the security \second and \first model diminishes, in exchange for smaller gains when security is \first.

\fi

\subsection{Prescriptive deployment guidelines.}\label{sec:guidelines}

Section~\ref{sec:stability} suggested that ASes use consistent routing policies. We now suggest a few more deployment guidelines.

\subsubsection{On the choice of early adopters.}\label{sec:T1suck:metric}\label{sec:T1suck:why}

Previous work~\cite{adopt,JenYannis,adoptability} suggests that Tier 1s should be the earliest adopters of S*BGP. % (due to their centrality and the high volumes of traffic they transit).
However, the discussion in Sections~\ref{sec:T1suck}~and~\ref{sec:islands} suggests that securing Tier 1s might not lead to good security benefits at the early adoption stage, when ASes are most likely to rank security \second or \third. We now confirm this.

\myparab{All Tier 1s and their stubs.} Even in a deployment that includes \emph{all} 13 Tier 1 ASes and their stubs (\ie 7872 ASes or $\approx20\%$ of the AS graph), improvements in security were almost imperceptible.  With security \second or \third, the average change in $H_{M',d}(S)-H_{M',d}(\emptyset)$ over secure destinations $d\in S$ causes the metric to increase by $<0.2\%$.

\myparab{Tier 1s, their stubs, and content providers.}  Following~\cite{adopt,FCCcommit}, we consider securing the CPs, the Tier 1s and all of their stubs, and obtained similar results.

\ifnum\full=1

\myparab{Analysis. } Why is a deployment at more than 20\% of the ASes in AS graph, including the large and well-connected Tier 1s, provide so little improvement in security? Recall that in Section~\ref{sec:T1suck} and Figure~\ref{fig:partitions:dest:sl}, we showed that when Tier 1 destinations are attacked, the vast majority of source ASes are doomed and almost none are protectable. It follows that if a source retains a secure route to a Tier 1 destination during an attack, that source is likely to be immune. The same argument also applies to other secure destinations (\ie CPs of stub customers of T1s); this is because, in the deployment scenarios above, most secure routes traverse a Tier 1 as their first hop. Because almost every source AS that continued to use a secure route during an attack would have routed to the legitimate destination even if no AS was secure, we see little improvements in our security metric.

\textbf{Figure~\ref{fig:t1_cp_sec_r_breakdown}} confirms this. We show what happens to the secure routes to each CP destination when security is \third; similar observations hold when security is \second. The height of each bar is the fraction of routes to each CP destination that are secure under normal conditions.  The lower part of the bar shows secure routes that were lost to protocol downgrade attacks (averaged over all attacks by non-stubs in $M'$), and the middle part shows the fraction of secure routes from immune source ASes to the destination. We clearly see that (1) most secure routes are lost to protocol downgrade attacks, and (2) almost all the secure routes that remain during attacks from source ASes that are \emph{immune}.

\begin{figure}[b]
    \includegraphics[width=3.25in]{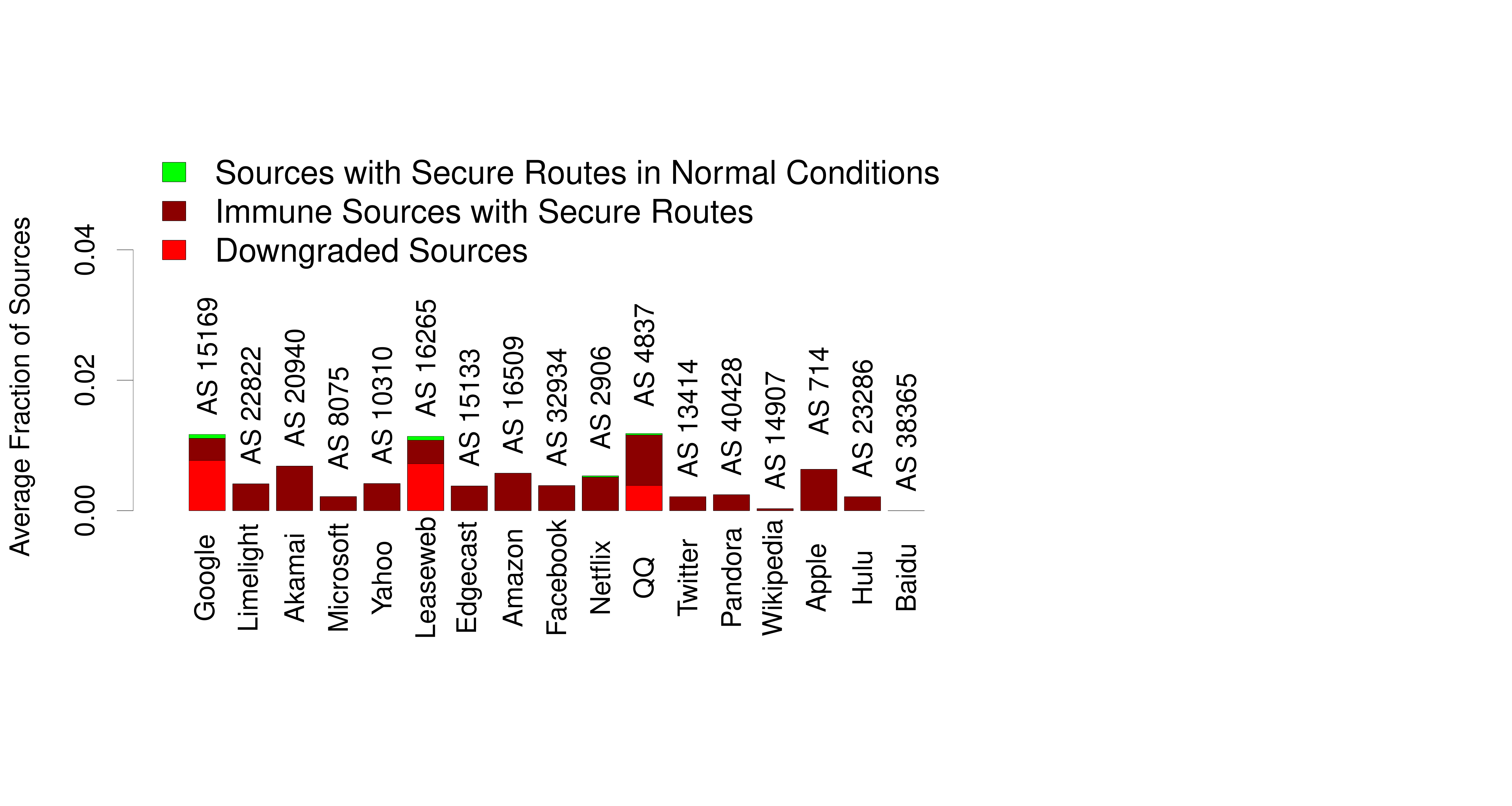}
    \vspace{-3mm}\caption{What happens to secure routes to each CP destination during attack. $S$ is the Tier 1s, the CPs, and all their stubs and security is~\third.}
    \vspace{-3mm}\label{fig:t1_cp_sec_r_breakdown}
\end{figure}	
\fi

\myparab{Choose Tier 2s as early adopters.}  We found that early deployments at the Tier 2 ISPs actually fare better than those at the larger, and better connected Tier 1s. For example, securing the 13 largest Tier 2s (in terms of customer degree) and all their stubs (a total of $6918$ ASes), the average change in $H_{M',d}(S)-H_{M',d}(\emptyset)$ over secure destinations $d\in S$ is $\approx 1\%$ when security is \second or \third.
\ifnum\full=1
This also agree with our observations in Section~\ref{sec:results:T2:nonstub}.
\fi

\subsubsection{Use simplex S*BGP at stubs.}\label{sec:simplex}

Next, we consider~\cite{BGPSEC,adopt}'s suggestion for reducing complexity by securing stubs with \emph{simplex S*BGP}.

\myparab{Simplex S*BGP.} Stub ASes have no customers of their own, and therefore (by \Ex) they will never send S*BGP announcements for routes through other ASes.  They will, however, announce routes to their own IP prefixes. For this reason \cite{adopt,BGPSEC} suggests either (1) allowing ISPs to send S*BGP messages on behalf of their stub customers or (2) allowing stubs to deploy S*BGP in a unidirectional manner, sending outgoing S*BGP messages but receiving legacy BGP messages.  Since a stub propagates only outgoing BGP announcements for a very small number of IP prefixes (namely, the prefixes owned by that stub), simplex mode can decrease computational load, and make S*BGP adoption less costly.

\smallskip\noindent
Given that 85\% of ASes are stubs, does this harm security?

\myparab{Figure~\ref{fig:metric:bigdeps}-\ref{fig:metric:Sonly}.}  The ``error bars'' in
Figure~\ref{fig:metric:bigdeps}-\ref{fig:metric:Sonly} show what happens when we suppose that all stubs run simplex S*BGP. There is little change in the metric.  To explain this, we note that (1) a stub's routing decision does not affect any other AS's routing decision, since by \Ex stubs do not propagate BGP routes from one neighbor to another, and (2) a stub's routing decisions are limited by the decisions made by its providers, so if its providers avoid attacks, so will the stub,  but (3) the stub acts like a secure destination, and therefore (nonstub) ASes establishing routes to the stub still benefit from S*BGP.
These results indicate that simplex S*BGP at stubs can lower the complexity of S*BGP deployment without impacting overall security.
\ifnum\full=1
Stub ASes that are concerned about their own security as \emph{sources} (rather than destinations) can, of course, always choose to deploy full S*BGP.
\fi

\ifnum\full=0 \newpage \fi
\section{Root Causes \& Non-Monotonicity}\label{sec:wrapup}

We now examine the reasons for the changes in our security metric as S*BGP is deployed.  We start by discussing two subtle phenomena: the collateral damages and collateral benefits incurred by insecure ASes from the deployment of S*BGP at \emph{other} ASes.  We then use these phenomena in a root-cause analysis of the results of Section~\ref{sec:results:metric}.

\subsection{Security is not monotonic!}\label{sec:examples}

The most obvious desiderata from S*BGP deployment is that the Internet should become only more secure as more ASes adopt S*BGP. Unfortunately, however, this is \emph{not} always the case. Security is not \emph{monotonic}, in the sense that securing more ASes can actually make other ASes unhappy.

To explain this, we use a running example taken from the UCLA AS graph, where the destination (victim) AS $d$ is Pandora's AS40426 (a content provider) and the attacker $m$ is an anonymized Tier 2 network. We consider the network \emph{before} and \emph{after} a partial deployment of S*BGP $S$ and see how the set of happy ASes changes; $S$ consists of all 100 Tier 2s, all 17 content providers, and all of their stubs.

\begin{table}
\begin{center}
\begin{small}
\begin{tabular}{|r|c|c|c|}
  \hline
  Security model & \first & \second & \third \\
  \hline\hline
  Protocol downgrade attacks & $X$ & $\checkmark$ & $\checkmark$ \\
  \hline
  Collateral benefits & $\checkmark$ & $\checkmark$ & $\checkmark$ \\
  \hline
  Collateral damages & $\checkmark$ & $\checkmark$ & $X$ \\
  \hline
\end{tabular}
\vspace{-2mm}
\caption{Phenomena in different security models}\label{tab:phenom}
\vspace{-7mm}
\end{small}
\end{center}
\end{table}

\subsubsection{Collateral Damages}\label{sec:Colldamages}

\myparab{Figure~\ref{fig:bensNdamages2nd}.} We show how AS~52142, a Polish ISP, suffers from collateral damage when security is \second. On the left, we show the network prior to S*BGP deployment.  AS~52142 is offered two paths, both insecure: a 3-hop path through his provider AS~5617 to the legitimate destination AS~40426, and a 5-hop bogus route to the attacker.
(The route to $m$ is really 4 hops long, but $m$ (falsely) claims a link to AS~40426 so AS~52142 thinks it is 5 hops long.)
AS~52142 will choose the legitimate route because it is shorter. On the right, we show the network after S*BGP deployment. AS~5617 has become secure and now prefers the secure route through its neighbor Cogent AS~174.  However, AS~5617's secure route is 5 hops long (right), significantly longer than the 2 hop route AS~5617 used prior to S*BGP deployment (left).  Thus, after S*BGP deployment AS~52142 learns a 6-hop legitimate route through AS~5617, and a 5-hop bogus route.  Since AS~52142 is insecure, it chooses the shorter route, and becomes unhappy as collateral damage.

\ifnum\full=1
\myparab{Collateral damages.} A source AS $s \notin S$ obtains collateral damages from an S*BGP deployment $S$ with respect to an attacker $m$ and destination $d$ if (a) $s$ was happy when the ASes in $T$ are secure, but (b) $s$ is unhappy when the ASes in $S$ are secure, and $S\supset T$.
\smallskip
\fi

\myparab{No collateral damages in the security \third model:} The collateral damage above occurs because AS 5617 prefers a \emph{longer} secure route over a shorter insecure route. This can also happen in the security \first model (but see also Appendix~\ref{apx:damages}), but not when security is \third.  See Table~\ref{tab:phenom}.
\begin{theorem}\label{thm:nodamage3rd}
In the security \third model, if an AS $s$ has a route to a destination $d$ that avoids an attacker $m$ when the set of secure ASes is $S$, then $s$ has a route to a destination $d$ that avoids attacker $m$ for every set of secure ASes in $T \supset S $.
\end{theorem}
\noindent
\ifnum\full=1
The proof is in Appendix~\ref{apx:mono}.
The security~\third model is our only \emph{monotone} model;  more secure ASes cannot result in fewer happy ASes, so the metric $H_{M,D}(S)$ grows monotonically in $S$.% In contrast, Figure~\ref{fig:bensNdamages2nd} proves that this does not hold in the other two models.
\fi

\ifnum\full=1
\myparab{Fewer immune ASes as security becomes more important?}  Collateral damages also explain why the fraction of immune ASes in the security \second model in Figure~\ref{fig:partitions} is smaller than the number of happy ASes in the baseline scenario (Section~\ref{sec:upperLower}). This is because in the security \second model, collateral damages mean that securing some ASes can actually make other ASes \emph{more} vulnerable to attack.
\fi

\subsubsection{Collateral Benefits}\label{sec:Collbenefits}

Insecure ASes can also become happy as a \emph{collateral benefit}, because \emph{other} ASes obtained secure routes:

\myparab{Figure~\ref{fig:bensNdamages2nd}.}  We show how AS~5166, with the Department of Defense Network Information Center, obtains collateral benefits when its provider AS~174, Cogent, deploys S*BGP.  On the left, we show the network prior to the deployment of S*BGP; focusing on Cogent AS~174, we see that it falls victim to the attack, choosing a bogus route through its customer AS~3491.  As a result, AS~5166 routes to the attacker as well. On the right, we show the network after S*BGP deployment.  Now, both AS~174 and AS~3491 are secure, and choose a longer secure customer route to the legitimate destination.  As a result, AS~5166, which remains insecure, becomes happy as a collateral benefit.

\ifnum\full=1
\myparab{Collateral benefits.} A source AS $s \notin S$ obtains collateral benefits from an S*BGP deployment $S$ with respect to an attacker $m$ and destination $d$ if (a) $s$ is unhappy when the ASes in $T$ are secure, but (b) $s$ is happy when the ASes in $S$ are secure, and $S\supset T$.
\fi

\smallskip
Collateral benefits are possible in all three routing policy models (Table~\ref{tab:phenom}).
\ifnum\full=0
Examples are in the full version.
\else
Here is an example when security is \third:

\myparab{Figure~\ref{fig:benefits3rd}.} We show how  AS34223, a Russian ISP, obtains collateral benefits in the security \third model.  The left subfigure shows how AS34223 and and its provider AS3267 react to the attack before S*BGP deployment; AS3267 learns two peer routes of equal length -- one bogus route to the attacker $m$ and one legitimate route to Pandora's AS 40426.  AS3267 then tiebreaks in favor of the attacker, so both AS3267 and his customer AS34223 become unhappy.  On the right, we show what happens after partial S*BGP deployment.  AS3267 has a \emph{secure} route to Pandora of equal length and type as the insecure route to $m$; so AS3267 chooses the secure route, and his insecure customer AS34223 becomes happy as a collateral benefit.

\fi

\begin{figure}
\begin{center}
 \includegraphics[width=1.4in]{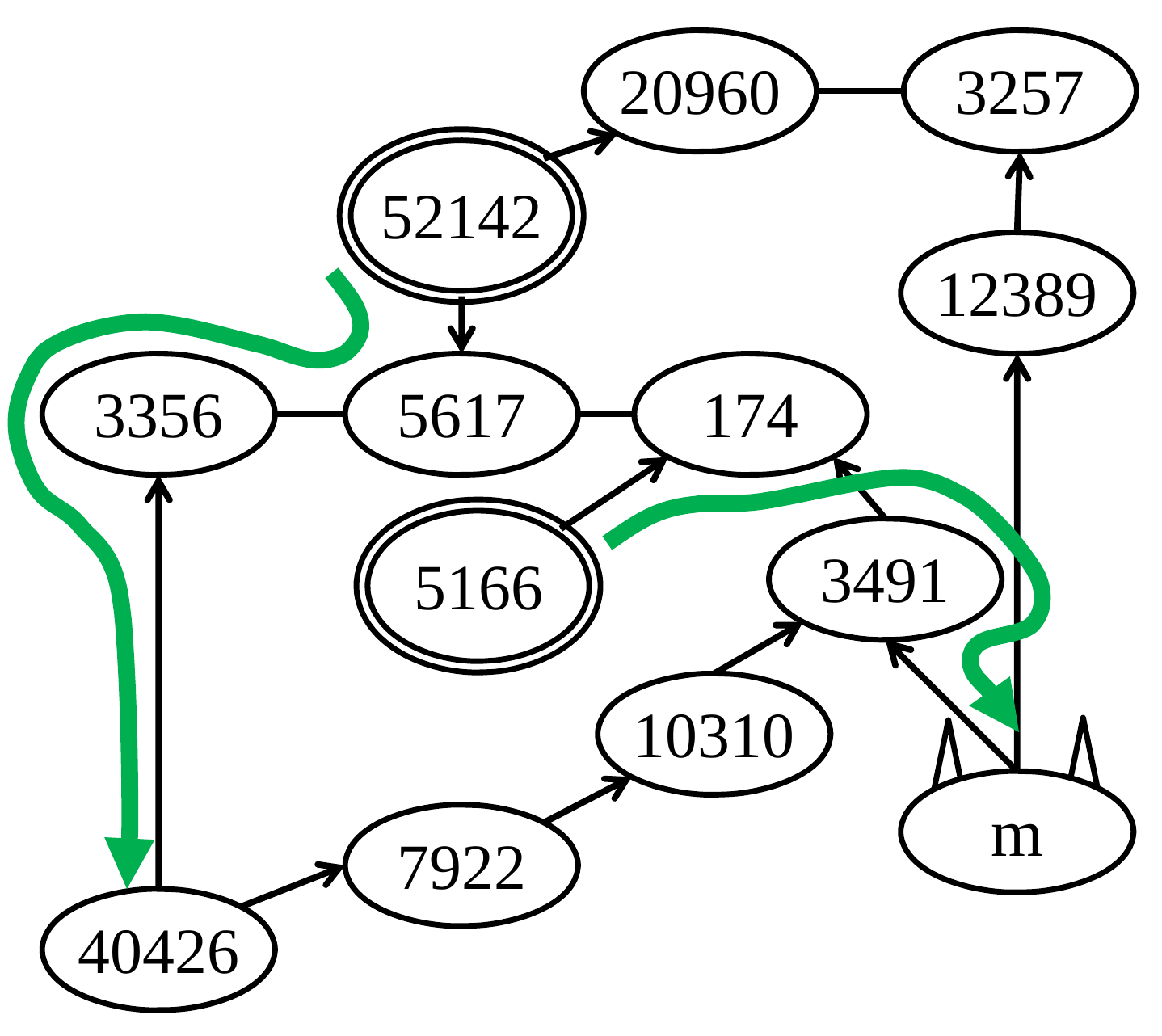}
 \includegraphics[width=1.4in]{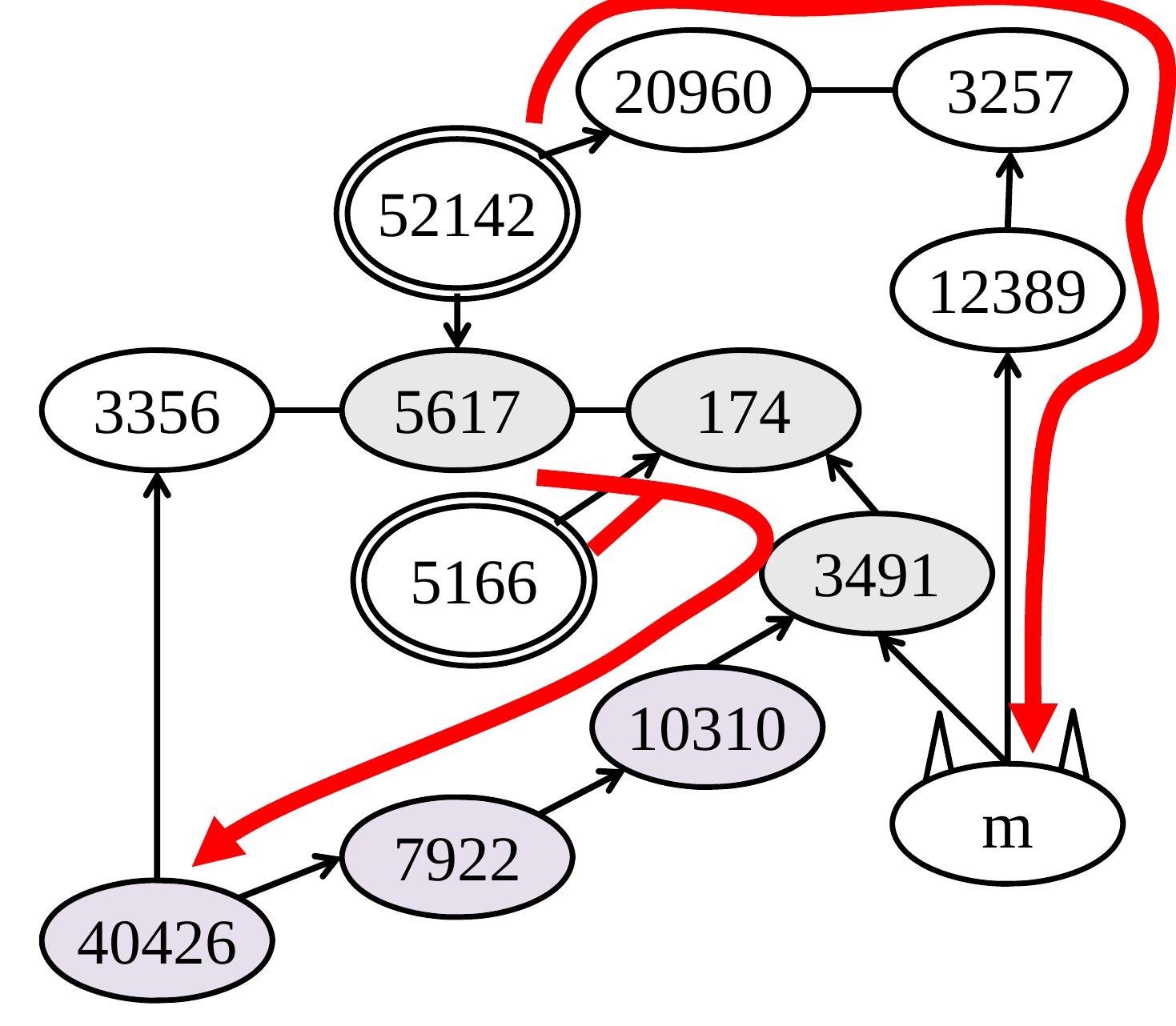}
  \vspace{-3mm}\caption{Collateral benefits \&  damages; sec~\second.}
  \vspace{-5mm}\label{fig:bensNdamages2nd}\end{center}
\end{figure}

\subsection{Root-cause analysis.}\label{sec:postmortem}

\ifnum\full=1

\begin{figure}
\begin{center}
 \includegraphics[width=.9in]{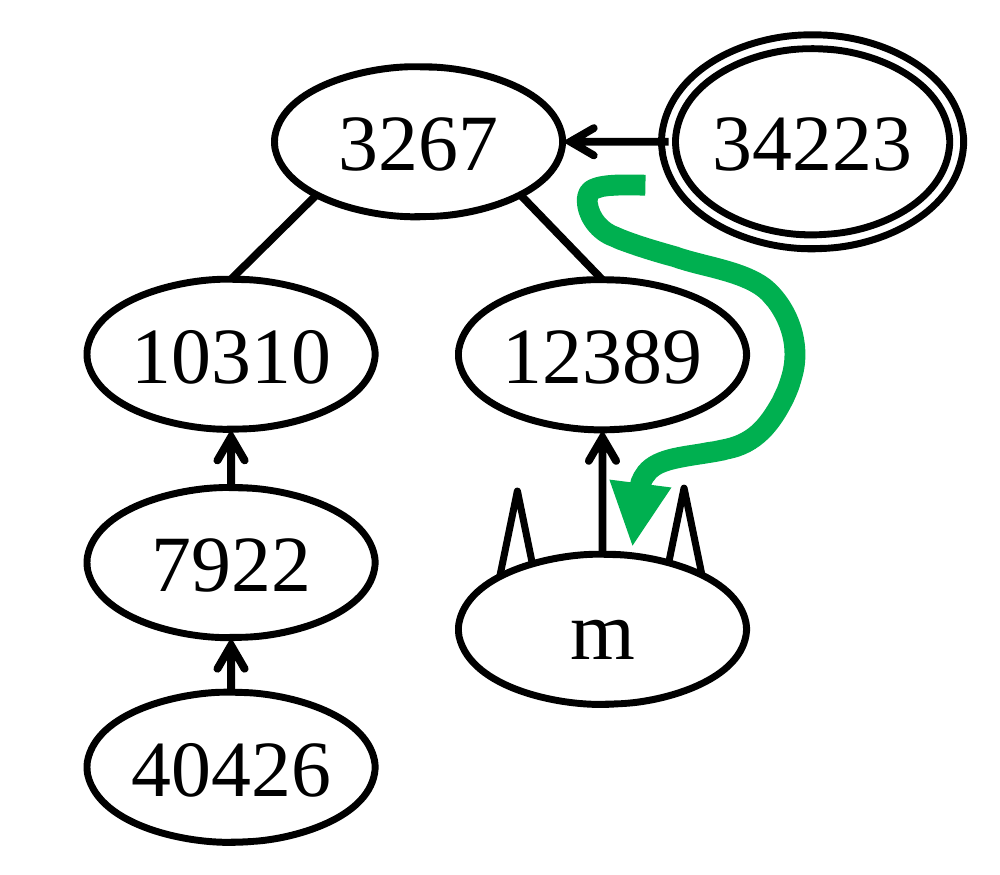}
  \includegraphics[width=.9in]{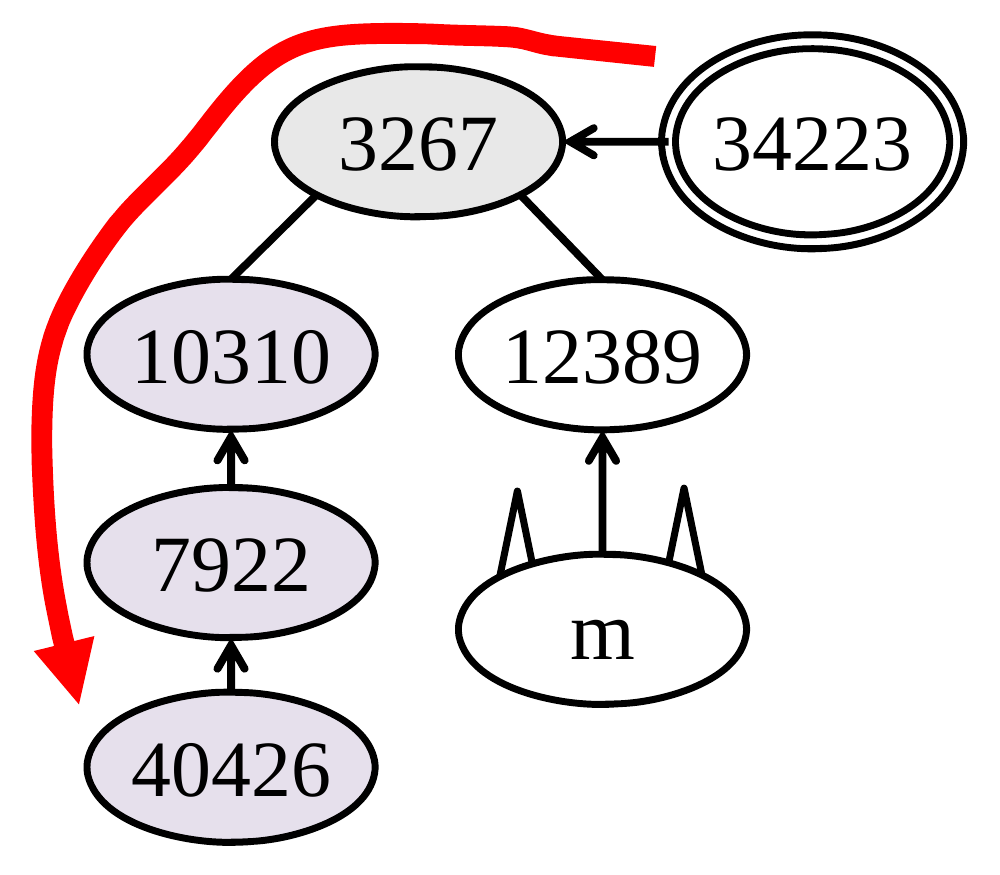}
  \vspace{-3mm}\caption{Collateral benefits; security \third.}
  \vspace{-5mm}\label{fig:benefits3rd}\end{center}
\end{figure}
\fi

Which of the phenomena in Table~\ref{tab:phenom} have the biggest impact on security?
We now check how these phenomena play out in the last step of the Tier 1 \& Tier 2 rollout of Section~\ref{sec:bigdep:all}. Recall that $S$ is all 13 Tier 1s, all 100 Tier 2s and all of their stubs, \ie roughly 50\% of the AS graph.
%
% We have shown that even such large-scale deployment of S*BGP leads to a very modest increase in security (only $4.2\%$ and $8.4\%$ in the security \second and the security \third models, respectively). We now investigate the reasons for these discouraging results.

\myparab{Figure~\ref{fig:rootcause} (left). }  We start with a root cause analysis for the security~\third model.  Recall that Theorem~\ref{thm:nodamage3rd} showed that collateral damages do not occur in the security \third model, and so we do not consider them here.

\mypara{Changes in secure routes. } 
\ifnum\full=1
We start with an analysis similar to that of Section~\ref{sec:T1suck:why};
\fi
The bottom three parts of the bar show the fraction of secure routes available in normal conditions, prior to any routing attacks. (Averaging is across all $V^2$ sources and destinations.) During routing attacks, these routes can be broken down into three types: (1) secure routes lost to protocol downgrade attacks (lowest part of the bar), (2) secure routes that are ``wasted'' on ASes that would have been happy \emph{even in the absence of S*BGP} (second lowest part), and (3) secure routes that protected ASes that were unhappy in the absence of S*BGP (third lowest part). (Averaging is, as usual, over $M'$ and $D=V$ and all $V$ source ASes.)  Importantly, improvements in our security metric can only result from the small fraction of secure routes in class (3); the remaining secure routes either (1) disappear due to protocol downgrades, or (2) are ``wasted'' on ASes that would have avoided the attack even without S*BGP.

\mypara{Changes in the metric.} The top two parts of the bar show how (the lower bound on) the metric $H_{M',V}(S)$ grows relative to the baseline scenario $S=\emptyset$ due to: (a) secure routes in class (3), and (b) (the lower bound on) the fraction of insecure ASes that obtained collateral benefits.  Figure~\ref{fig:rootcause}(left) thus illustrates the importance of collateral benefits.

\myparab{Figure~\ref{fig:rootcause} (right). }  We perform the same analysis for the security \first model.  By Theorem~\ref{thm:noPDA1st}, protocol downgrade attacks occur only rarely in this model, so these are not visible in the figure. However, we now have to account for collateral damages (Section~\ref{sec:Colldamages}), which we depict with the smaller sliver on right of the  figure.  We obtain the change in the metric by subtracting the collateral damages from the gains resulting from (a) offering secure routes to unhappy ASes and (b) collateral benefits.  Fortunately, we find collateral damages to be a relatively rare phenomenon.%, compared to the others we study.

\ifnum\full=1
\myparab{Fitting it all together?} This analysis reveals that changes in the metric can be computed as follows: (Secure routes created under normal conditions) $+$ (collateral benefits)  $-$ (protocol downgrades)  $-$ (secure routes ``wasted'' on ASes that are already happy) $-$ (collateral damages).  We find that  all of these phenomena (with the exception of collateral damage) have significant impact on the security metric. These observations also drive home the point that the number of routes learned via S*BGP in \emph{normal conditions} is a poor proxy for the ``security'' of the network; more sophisticated metrics like the ones we use here are required.
\fi

\smallskip
Results where security \second look very similar to results when security is \third, with the addition of a small amount of collateral damage.  The bottom line is, when security is \second or \third, (1) protocol downgrade attacks cause many secure routes that were available under normal conditions to disappear, and (2) those ASes that retain their secure routes during the attack would have been happy even if S*BGP had not been deployed; the result is meagre increases in the security metric.  Meanwhile, when security is \first, few downgrades occur, and the security metric is greatly improved.

\begin{figure}
\begin{center}
    \includegraphics[height=1.4in]{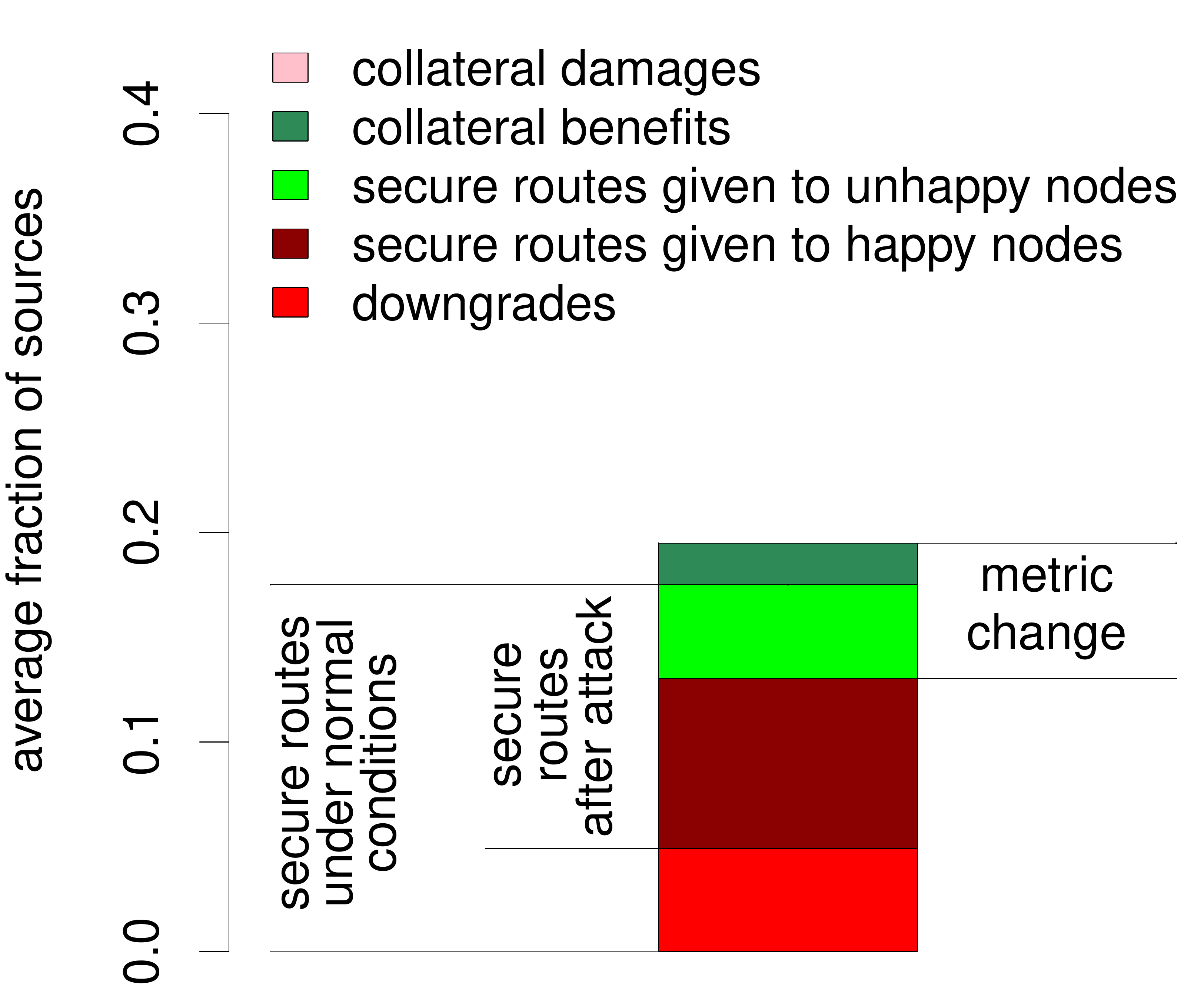}
    \includegraphics[height=1.4in]{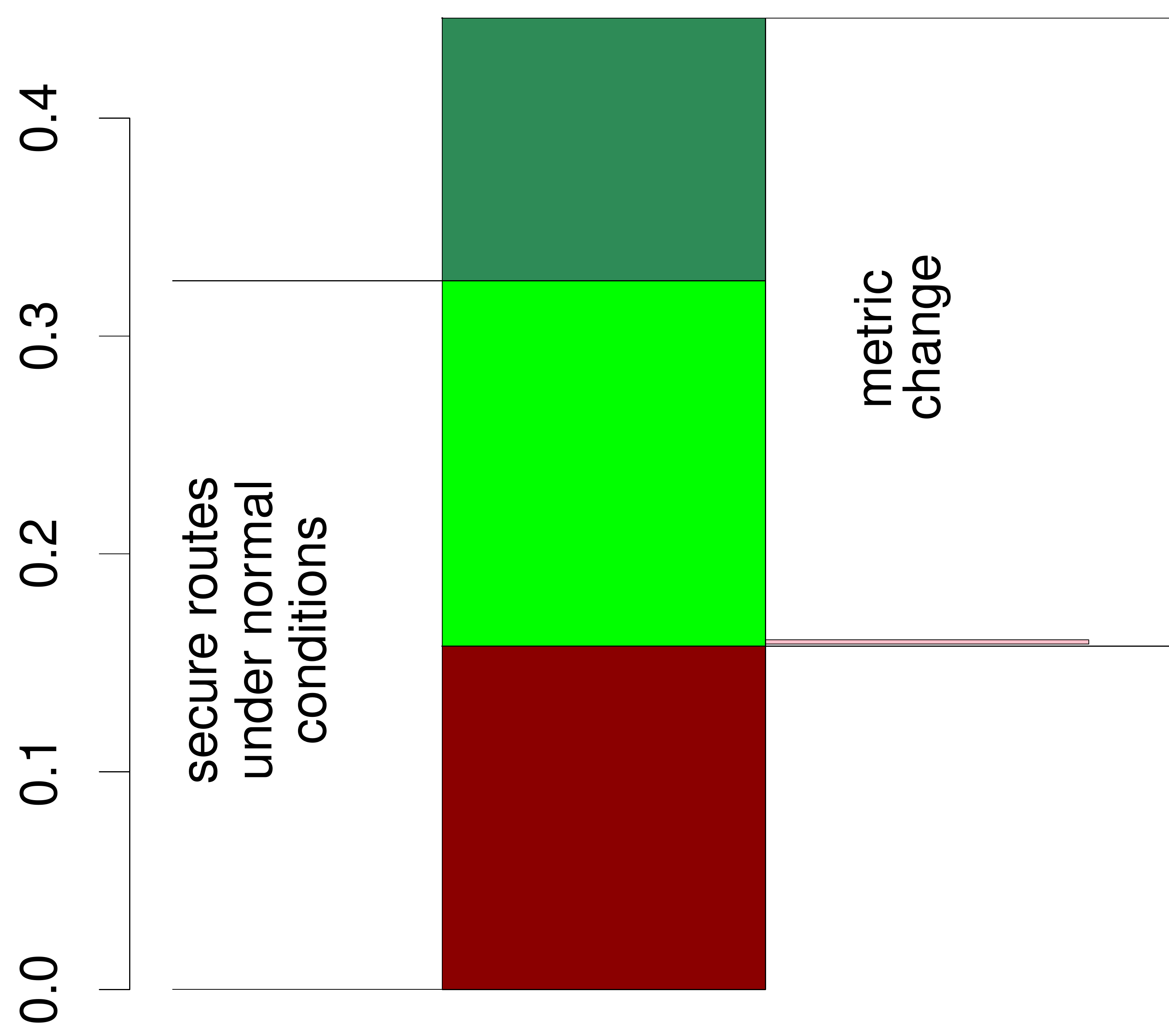}
    \vspace{-6mm}\caption{Changes in the metric explained. Sec~\third (left) and Sec~\first (right).}
    \vspace{-5mm}\label{fig:rootcause}\end{center}
\end{figure}

\iffalse
we consider what happened to the secure routes during routing attacks.   The long vertical bar shows the fraction of secure routes available in normal conditions (\ie in the absence of attacks).  The lowest part of the bar in the figure shows the average fraction of these routes lost to protocol downgrade attacks, and the second-lowest part of the bar shows the average fraction of the routes that remain during the attack that are ``wasted'' on ASes that were \emph{already happy} in the baseline scenario.  (Averaging is, as usual over $M'$ and $D=V$.) Most of secure routes available under normal conditions are either (1) lost to downgrades, or (2) wasted on immune ASes. As mentioned above, only the remaining secure routes (at ASes that were \emph{unhappy} in the baseline scenario) improve the metric.

\fi
\section{Related Work}\label{sec:related}

Over the past decades several security extensions to BGP have been proposed; see~\cite{BGPsurvey} for a survey. However, proposals of new security extensions to BGP, and their subsequent security analyses typically assume that secure ASes will never accept insecure routes~\cite{JenYannis,adoptability}, which is reasonable in the full deployment scenario where every AS has already deployed S*BGP~\cite{BGPattack,BGPsurvey,PaulF}.  There have also been studies on incentives for S*BGP adoption~\cite{adopt,adoptability};
these works suggest that ``S*BGP and BGP will coexist in the long term''~\cite{adopt}, which motivated our study of S*BGP in partial deployment.
\ifnum\full=1
 The partial deployment scenarios we considered have been suggested in practice~\cite{FCCcommit} and in this literature~\cite{JenYannis,adopt,adoptability}.
\fi

Our work is most closely related to~\cite{BGPattack}, which also measures ``security'' as the fraction of source ASes that avoid having their traffic intercepted by the attacking AS. However, \cite{BGPattack} always assumes that the S*BGP variant is \emph{fully deployed}. Thus, as discussed in Section~\ref{sec:baseline}, \cite{BGPattack} also finds that fully-deployed origin authentication provides good security against attack we studied here (\ie announcing ``$m,d$'' using insecure BGP, see Section~\ref{sec:attack:dets}), but rightly assumes this attack fails against fully-deployed S*BGP.   Moreover, \cite{BGPattack} does not analyze interactions between S*BGP and BGP that arise during partial deployment (\eg Table~\ref{tab:phenom}).

Finally, \cite{robertCCS} includes cryptographic analysis of S*BGP in partial deployment, and an Internet draft~\cite{BGPSECthreat} mentions protocol downgrade attacks. However, neither explores how attacks on partially-deployed S*BGP can impact routing, or considers the number / type of ASes harmed by an attack.

%; this observation, coupled with our experience with the long, slow deployments of new network protocols like IPv6 and DNSSEC, highlight the importance of studying security in the partial deployment scenario.

\section{Conclusion}\label{sec:conclusion}

On one hand, our results give rise to guidelines for partially-deployed S*BGP: (1) Deploying lightweight simplex S*BGP at stub ASes, instead of full-fledged S*BGP; this reduces deployment complexity at the majority of ASes without compromising overall security.  (2) Incorporating S*BGP into routing policies in a similar fashion at all ASes, to avoid introducing routing anomalies like BGP Wedgies.  (3) Deploying S*BGP at Tier 2 ISPs, since deployments of S*BGP at Tier 1s can do little to improve security.
On the other hand, we find that partially-deployed S*BGP %introduces new instabilities and vulnerabilities, and
provides, on average, limited security benefits over route origin authentication when ASes do not prioritize security \first.

We hope that our work will call attention to the challenges that arise during partial deployment, and drive the development of solutions that can help surmount them.  One idea is to find ways to limit protocol downgrade attacks, as these cause many of our negative results. For example, one could add ``hysteresis'' to S*BGP, so that an AS does not immediately drop a secure route when ``better'' insecure route appears.  Alternatively, one could find deployment scenarios that create ``islands'' of secure ASes that agree to prioritize security \first for routes between ASes in the island; the challenge is to do this without disrupting existing traffic engineering or business arrangements.
Other security solutions could also be explored. For example, origin authentication with anomaly detection and prefix filtering could be easier to deploy (they can be based on the RPKI), and may be as effective as partially-deployed S*BGP.

%  Given the enormous effort required to deploy S*BGP, and in light of all these complications, our results cast doubt about the value of deploying S*BGP. Route origin authentication, on its own, may just be good enough.

\ifnum\anon=0
\section*{Acknowledgments}
We are grateful to BU and XSEDE for computing resources, and Kadin Tseng, Doug Sondak, Roberto Gomez and David O'Neal for helping us get our code running on various platforms. We thank Walter Willinger and Mario Sanchez for providing the list of ASes in each IXP that we used to generate our IXP-augmented AS graph, Phillipa Gill for useful discussions and sharing the results of~\cite{surveyEmail} with us, and Leonid Reyzin, Gonca Gursun, Adam Udi, our shepherd Tim Griffin and the anonymous SIGCOMM reviewers for comments on drafts of this paper.
This work was supported by NSF Grants S-1017907, CNS-1111723, ISF grant 420/12, Israel Ministry of Science Grant 3-9772, Marie Curie Career Integration Grant, IRG Grant 48106, the Israeli Center for Research Excellence in Algorithms, and a gift from Cisco.
\fi

\begin{small}
\bibliographystyle{abbrv}%ieeetr}	
\bibliography{partialSec_paper.bib}
\end{small}

\appendix
%%%%%%%%%%%%%%%% CANDIDATE FOR LAST MINUTE CHOPPING %%%
\section{More Collateral damage}\label{apx:damages}

\begin{figure}[b]
\begin{center}
   \includegraphics[width=1.5in]{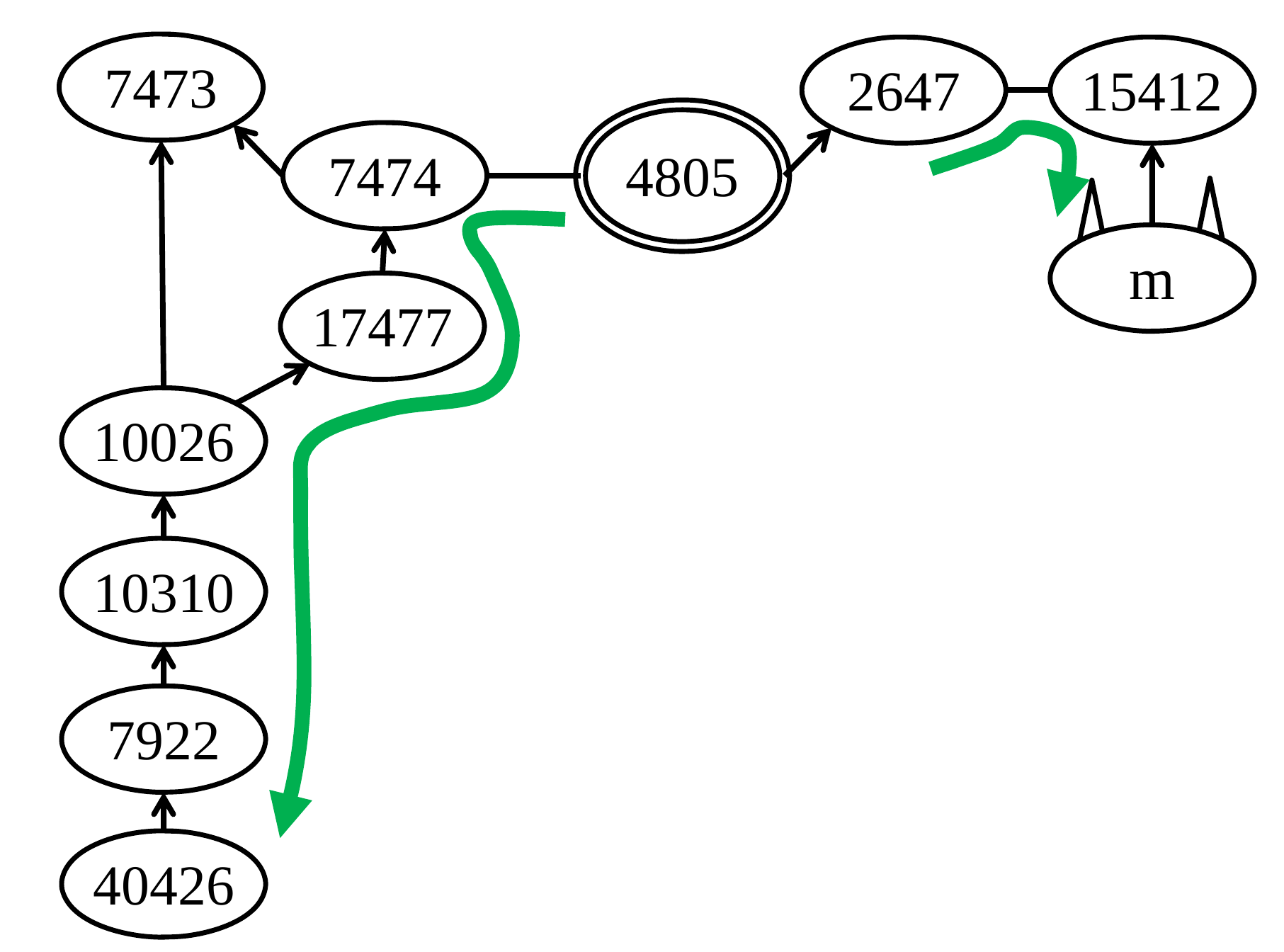}
  \includegraphics[width=1.5in]{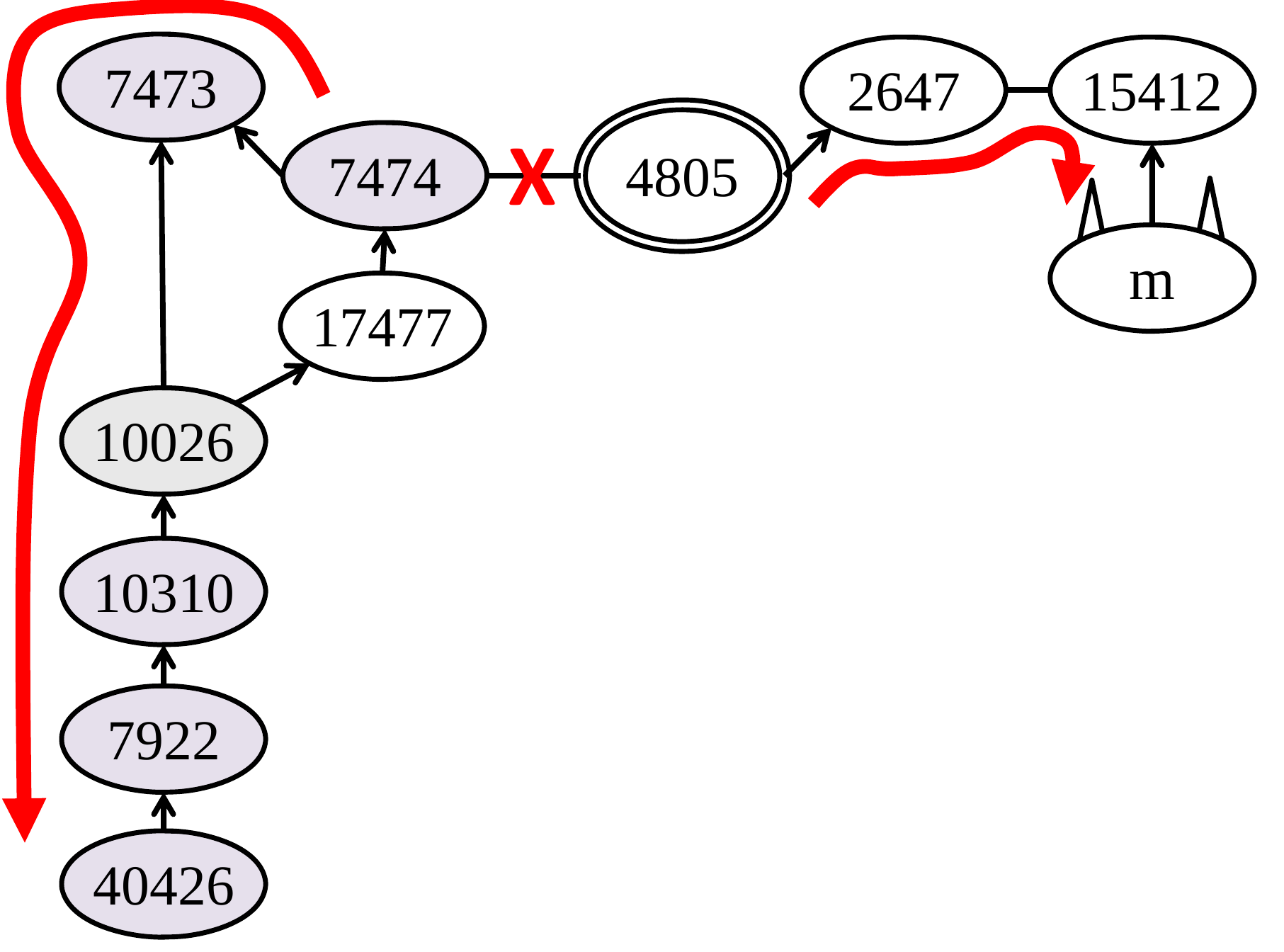}
  \vspace{-3mm}\caption{Collateral damages; security \first.}
  \vspace{-5mm}\label{fig:damages1st}\end{center}
\end{figure}

Figure~\ref{fig:bensNdamages2nd} revealed that collateral damages can be caused by secure ASes that choose \emph{long} secure paths.  When security is~\first, collateral damages can also be caused by secure ASes that choose \emph{expensive} secure paths:

\myparab{Figure~\ref{fig:damages1st}.}  We show how AS~4805, Orange Business in Oceania, suffers from collateral damage when security is \first. On the left, we show the network prior to S*BGP deployment. Orange Business AS4805 learns two routes: a legitimate route through its peer Optus Communications AS~7474, and a bogus route through its provider AS~2647. Since AS~4805 prefers peer routes over provider routes per our \LP rule, it will choose the legitimate route and avoid the attack. On the right, we show what happens after S*BGP deployment. Now, Optus Communications AS~7474 has started using a secure route. However, this secure route is through its provider AS~7473. Observe that AS~7474 is no longer willing to announce a route to its peer AS~4805 as this would violate the export policy \Ex.  AS~4805 is now left with the bogus provider route through AS~2647, and becomes unhappy as collateral damage.

\ifnum\full=1
%\section{Simulation Algorithms} \label{apx:algorithms_and_properties}

\section{Computing Routing Outcomes}\label{apx:algos}

Below we present algorithms for computing S*BGP routing outcomes in the presence of an attacker (per Section~\ref{sec:attack:dets}),  in each of our three S*BGP routing models. These algorithms receive as input an attacker-destination pair $(m,d)$ and the set of secure ASes $S$ and output the S*BGP routing outcome (in each of our three S*BGP routing models). We point out that our algorithms can also be used to compute routes during normal conditions (when there is no attacker $m=\emptyset$), and when no AS is secure $S=\emptyset$. In these algorithms, which extend the algorithmic approach used in~\cite{BGPattack,adopt,quicksand} to handle partial S*BGP deployment in the presence the adversary described in Section~\ref{sec:attack:dets}, we carefully construct a partial two-rooted routing tree by performing multi-stage breadth-first-search (BFS) computations with $d$ and $m$ as the two roots. We prove the correctness of our algorithms (that is, that they indeed compute the desired S*BGP routing outcomes) in Appendix~\ref{apx:alg:corr}.  In subsequent sections, we show how to use these algorithms to partition ASes into doomed/immune/protectable nodes, to determine which ASes are happy, or experience protocol downgrade attacks for a given $(m,d)$-pair and deployment $S$.

\subsection{Notation and preliminaries.}

Since BGP (and S*BGP) sets up routes to each destination independently, we focus on routing to a unique destination $d$. We say that  a route is \legit if it does not contain the attacker $m$ (either because there is no attacker $m=\emptyset$ or because the attacker is not on the route). We say that a route is \attacked otherwise. Observe that in the presence of an attacker $m$ launching the attack of Section~\ref{sec:threat}, all \attacked routes have $m$ as the first hop following $d$. We use the following definition of ``perceivable routes'' from~\cite{LGSTR}.

\begin{definition}[Perceivable routes] A simple (loop-free) route $R = \{v_{i-1}, \ldots, v_1,d\}$ is \emph{perceivable} at AS $v_i$ if one of the two following conditions holds:
\begin{enumerate}
\item $R$ is \legit (so $v_1 \neq m$), and for every $0 < j < i$ it follows that $v_j$ announcing the route $(v_j,\ldots,d)$ to $v_{j+1}$ does not violate \Ex.

\item $R$ is \attacked (so $v_1 = m$), and for every $1 < j < i$ it follows that $v_j$ announcing the route $(v_j,\ldots,d)$ to $v_{j+1}$ does not violate \Ex.
\end{enumerate}
\end{definition}

Intuitively, an AS's set of perceivable routes captures all the routes this AS could potentially learn during the S*BGP convergence process. All non-perceivable routes from an AS can safely be removed from consideration as the \Ex condition ensures that they will not propagate from the destination/attacker to that AS.

We say that a route $(v_{i-1}, \ldots, v_1, d)$ is a \emph{customer route} if $v_{i-1}$ is a customer of $v_i$.  We define \emph{peer routes} and \emph{provider routes} analogously. We say that a route $R = \{v_i, v_{i-1}, ... , v_1, d\}$ contains AS $x$, if at least one AS in $\{v_i, v_{i-1}, ... , v_1, d\}$ is $x$.

\myparab{$\PR$ and $\BPR$ sets.}
Let $\PR(v_i, m, d)$ be the set of perceivable routes from $v_i$ for the attacker-victim pair $(m,d)$ when attacker $m$ attacks destination $d$ using the attack described in Section~\ref{sec:attack:dets}. (We set $m=\emptyset$ when there is no attacker.) Given a set of secure ASes $S$, for every AS $v_i$ we define the $\BPR(v_i, S, m, d)$ to be the set of all perceivable routes in $\PR(v_i, m, d)$ that are preferred by $v_i$ over all other perceivable routes, before the arbitrary tiebreak step \TB, according to the routing policy model (\ie security \first, \second, or \third) under consideration.  (Again, we set $m=\emptyset$ when there is no attacker and $S=\emptyset$ when no ASes are secure).  We define $\Nxt(v_i, S, m, d)$ to be the set of all neighbors of $v_i$ that are next hops of all routes in $\BPR(v_i,  S, m, d)$.  We will just use $\Nxt(v_i)$ when it is clear what $S$, $m$ and $d$ are.

Observe that in each of our models, all routes in $\BPR(v_i, S, m, d)$ must (1) belong to the same type---customer routes, peer routes, or provider routes, (2) be of the same length, and (3) either all be (entirely) secure or insecure.

\subsection{Algorithm for security \third.}\label{apx:alg:3}

We now present our algorithm for computing the S*BGP routing outcome in the security \third model in the presence of a set of secure ASes $S$ and an attacker $m$. We note that this algorithm also serves to compute the routing outcome when no ASes are secure, \ie $S=\emptyset$.
As in~\cite{LGSTR} (which studies a somewhat different BGP routing model and does not consider S*BGP) we exhibit an iterative algorithm \textbf{Fix-Routes} (FR) that, informally, at each iteration fixes a single AS's route and adds that AS to a set $\mathcal{I} \subseteq V$. This goes on until all ASes are in $\mathcal{I}$ (that is, all ASes' routes are fixed). We will later prove that FR indeed outputs the BGP routing outcome.

FR consists of three subroutines: \textbf{Fix Customer Routes} (FCR), \textbf{Fix Peer Routes} (FPeeR), and \textbf{Fix Provider Routes} (FPrvR), that FR executes in that order. Note that at the very beginning of this algorithm, $\mathcal{I}$ contains only the legitimate destination $d$ and the attacker $m$ (if there is an attacker). We now describe FR and its subroutines.

\myparab{Step I: The FCR subroutine.} FR starts with FCR; at this point $\mathcal{I}$ contains only the legitimate destination $d$ and the attacker $m$. Intuitively, FCR constructs a partial two-rooted tree (rooted at $d$ and $m$ on the graph, using a BFS computation in which only customer-to-provider edges are traversed.% (\ie only ASes who are (indirect) providers of those ASes that have already been added to the partial two-rooted tree are explored).
Initially, $d$ has path length $0$ and $m$ has path length $1$ (to capture the fact that $m$ announces that it is directly connected to $d$ in the attack of Section~\ref{sec:attack:dets}).

We set $\PR^{0}(v_i)=\PR(v_i,m,d)$ and $\BPR^0(v_i)=\BPR(v_i, S, m, d)$ for every AS $v_i$. We let $r$ be the FR iteration and initialize it to $r:=0$.

While there is an AS $s\notin \mathcal{I}$ such that $\PR^{r-1}(s)$ contains at least one customer route, we ``fix'' the route of (at least) one AS by executing the following steps:

\begin{enumerate}
	\item{r++;}
	\item{Select the AS $v_i \notin \mathcal{I}$ that has the shortest \textbf{customer} route in its set $\BPR^{r-1}(v_i)$ (if there are multiple such ASes, choose one arbitrarily);}
 	\item{Add $v_i$ to $\mathcal{I}$; set $ \Nxt(v_i)$ to be $v_i$'s next-hop on the route in $\BPR^{r-1}(v_i)$ selected according to its tie-breaking rule \TB;}
 	\item{Remove, for every AS $v_j$, all routes in $\PR^{r-1}(v_i)$ that contain $v_i$ but whose suffix at $v_i$ is not in  $\BPR^{r-1}(v_i)$ to obtain the new set $\PR^r(v_i)$; set $\BPR^r(v_j)$ to be $v_j$'s most preferred routes in $\PR^r(v_i)$}
 	\item{Add all ASes $v_j$ such that $\PR^r(v_i)=\emptyset$ to $\mathcal{I}$.}
\end{enumerate}

\myparab{Step II: the FPeeR subroutine.} This step starts with $\mathcal{I}$ and the configuration of the routing system and the $\PR$ and $\BPR$ sets the way it is after execution of FCR (all the ASes discovered the FCR step have their route selections locked), \ie $\mathcal{I}$ contains only $d$, $m$, and ASes with either empty or customer routes. We now use only single peer-to-peer edges to connect new yet-unexplored ASes to the ASes that were locked in the partial routing tree in the \first stage of the algorithm.

While there is an AS $s\notin \mathcal{I}$ such that $\PR^{r-1}(s)$ contains at least one peer route, the following steps are executed:
\begin{enumerate}
\item r++
	\item{select an AS $v_i \notin \mathcal{I}$;}
 	\item{add $v_i$ to $\mathcal{I}$; set $ \Nxt(v_i)$ to be $v_i$'s next-hop on the route in $\BPR^{r-1}(v_i)$ selected according to its tie-breaking rule \TB;}
 	\item{remove, for every AS $v_j$, all routes in $\PR^{r-1}(v_i)$ that contain $v_i$ but whose suffix at $v_i$ is not in  $\BPR^{r-1}(v_i)$ to obtain the new set $\PR^r(v_i)$; set $\BPR^r(v_j)$ to be $v_j$'s most preferred routes in $\PR^r(v_i)$}
 	\item{add all ASes $v_j$ such that $\PR^r(v_i)=\emptyset$ to $\mathcal{I}$.}
\end{enumerate}

\myparab{Step III: The FPrvR subroutine.} We now run a BFS computation in which only provider-to-customer edges are traversed, that is, only ASes who are direct customer of those ASes that have already been added to the partial two-rooted tree are explored. This step starts with $\mathcal{I}$ and the configuration of the routing system and the $\PR$ and $\BPR$ sets the way it is after the consecutive execution of FCR and FPeeR.

While there is an AS $s\notin \mathcal{I}$ such that $\PR^{r-1}(s)$ contains at least one provider route, we execute the identical steps as in FCR, with the exception that we look for the $v_i$ that has the shortest \textbf{provider} route in its set $\BPR^{r-1}(v_i)$.

\subsection{Algorithm for security \second.}
\label{apx:alg:2}

Our algorithm for the security \second model is a refinement of the iterative algorithm \textbf{Fix Routes} (FR) presented above for the security \third model. This new algorithm is also a 3-stage BFS in which customer routes are fixed before peer routes, which are fixed before provider routes. In each stage we are careful to prioritize ASes with \emph{secure} routes over ASes with insecure routes.

We present the following two new subroutines. (1) \textbf{Fix Secure Customer Routes} (FSCR): FSCR is identical to FCR, with the sole exception that for the AS chosen at each iteration $r$ has a $\BPR^{r-1}$ that contains a \emph{secure} customer route; (2) \textbf{Fix Secure Provider Routes} (FSPrvR): FSPrvR is identical to FPrvR, with the sole exception that for the AS chosen at each iteration $r$ has a $\BPR^{r-1}$ that contains a \emph{secure} provider route. The variant of FR for the security \third model executes the subroutines the following order:

\begin{enumerate}
\item FSCR
\item FCR
\item FPeeR
\item FSPrvR
\item FPrvR
\end{enumerate}

\subsection{Algorithm for security \first.} \label{apx:alg:1}

Once again, we present a variant of the Fix Routes (FR) algorithm. This multi-stage BFS computation first discovers all ASes that can reach the destination $d$ via secure routes and only then discovers all other ASes (as in our algorithm for the security \third model).

We present the following new subroutine. \textbf{Fix Secure Peer Routes} (FSPeeR): FSPeeR is identical to FSPeeR, except that the AS chosen at each iteration $r$ has a \emph{secure} peer route in its $\BPR^{r-1}$ set. This variant of FR executes the subroutines in the following order:

\begin{enumerate}
\item FSCR
\item FSPeeR
\item FSPrV
\item FCR
\item FPeeR
\item FPrvR
\end{enumerate}

%%%%%%%%%%%%%%%%%%%%%

\subsection{Correctness of Algorithms}\label{apx:alg:corr}

We now prove that that our algorithms for computing the S*BGP routing outcomes indeed output the desired outcome.

\subsubsection{Correctness of algorithm for security \third.}

The proof that our algorithm for the security \third model outputs the S*BGP routing outcome in this model follows from the combination of the lemmas below. Recall that each of our algorithms computes, for every AS $v_i$, a next-hop AS $ \Nxt(v_i)$. Let $R_{v_i}$ be the route from $v_i$ induced by these computed next-hops.

\begin{lemma} \label{lem_get_stable_provider} Under S*BGP routing, the route of every AS added to $\mathcal{I}$ in FCR is guaranteed to stabilize to the route $R_{v_i}$. \end{lemma}
\begin{proof}
We prove the lemma by induction on the FCR iteration. Consider the first iteration. Observe that the AS chosen at this iteration of FCR must be a direct provider of $d$ (that is, have a customer route of length 1). Hence, in the security \third model, once this AS learns of $d$'s existence it will select the direct route to $d$ and never choose a different route thereafter (as this is its most preferred route). Now, let us assume that for every AS chosen in iterations $1,\ldots,r$ the statement of the lemma holds. Let $v_i$ be the AS chosen at iteration $r+1$ of FCR. Consider $v_i$'s $\BPR$ set at that time. By definition, every route in the $\BPR$ set is perceivable and so must comply with \Ex at each and every ``hop'' along the route. Notice, that this, combined with the fact that all routes in $v_i$'s $\BPR$ set are customer routes, implies that the suffix of every such route is also a perceivable customer route. Consider an AS $v_j$ that is $v_i$'s next-hop on some route in $v_i$'s $\BPR$ set. Notice that $v_j$'s route is fixed at some iteration in $\{1,\ldots,r\}$ (as $v_j$ has a shorter perceivable customer route than $v_i$). Hence, by the induction hypothesis, at some point in the S*BGP convergence process, $v_j$'s route converges to $R_{v_j}$ for every such AS $v_j$. Observe that, from that point in time onward, $v_i$'s best available routes are precisely those capture by $\BPR$ in the $r+1$'th iteration of FCR. Hence, from that moment onwards $v_i$ will repeatedly select the route $R_{v_i}$ according to the tiebreak step $\TB$ and never select a different route thereafter.
\end{proof}

\begin{lemma} \label{lem_get_stable_peer} Under S*BGP routing, the route of every AS added to $\mathcal{I}$ in FPeeR is guaranteed to stabilize to the route $R_{v_i}$. \end{lemma}
\begin{proof}
Consider an AS $v_i$ chosen at some iteration of FPeeR. Observe if $(v_i,v_{i-1},\ldots,d)$ is a perceivable peer route then $(v_{i-1},\ldots,d)$ must be a perceivable customer route (to satisfy the \Ex condition). Hence, for every such route in $v_i$'s $\BPR$ set it must be the case that the route of $v_i$'s next-hop on this route $v_j$  was fixed in FCR. By Lemma~\ref{lem_get_stable_provider}, at some point in the S*BGP convergence process, $v_j$'s route converges to $R_{v_j}$ for every such AS $v_j$. Observe that, from that point in time onward, $v_i$'s best available routes are precisely those capture by its $\BPR$ set at the iteration of FPeeR in which $v_i$ is chosen. Hence, $v_i$ will select the route $R_{v_i}$ according to the tiebreak step $\TB$ and never select a different route thereafter.
\end{proof}

\begin{lemma} \label{lem_get_stable_customer} Under S*BGP routing, the route of every AS added to $\mathcal{I}$ in FPrvR is guaranteed to stabilize to the route $R_{v_i}$. \end{lemma}
\begin{proof}
We prove the lemma by induction on the number of FPrvR iterations. Consider the first iteration. Let $v_i$ be the AS chosen at this iteration, let $v_j$ be a next-hop of $v_i$ on some route $R$ in $v_i$'s $\BPR$ set, and let $Q$ be the suffix of $R$ at $v_j$. Observe that $Q$ cannot possibly be a provider route, for otherwise $v_j$ would have been chosen in FPrvR before $v_i$. Hence, $Q$ must be either a customer route or a peer route, and so $v_j$'s route must have been fixed in either FCR or FPeeR. Hence, by the previous lemmas, under S*BGP convergence, every such $v_j$'s route will eventually converge to $R_{v_j}$. Observe that once all such ASes' routes have converged and onwards $v_i$'s best available routes are precisely those captured by $\BPR$ in the $r+1$'th iteration of FPrvR. Hence, $v_i$ will select the route $R_{v_i}$ according to the tiebreak step $\TB$ and never select a different route thereafter.

Now, let us assume that for every AS chosen in iterations $1,...,r$ the statement of the lemma holds. Let $v_i$ be the AS chosen at iteration $r+1$ of FPrvR and consider $v_i$'s $\BPR$ set at this time. Let $v_j$ again be a next-hop of of $v_i$ on some route $R$ in $v_i$'s $\BPR$ set, and let $Q$ be the suffix of $R$ at $v_j$. Observe that if $Q$ is a provider route then $v_j$'s route must have been fixed in FPrvR at some point in iterations $\{1,\ldots,r\}$. If, however, $Q$ is either a customer route or a peer route $v_j$'s route must have been fixed in either FCR or FPeeR. Hence, by the previous lemmas and the induction hypothesis, under S*BGP convergence, every such $v_j$'s route will eventually converge to $R_{v_j}$. From that moment onwards $v_i$'s best available routes are precisely those captured by $\BPR$ in the $r+1$'th iteration of FprvR. Hence, $v_i$ will select the route $R_{v_i}$ according to the tiebreak step $\TB$ and never select a different route thereafter.
\end{proof}

\subsubsection{Correctness of algorithm for security \second.}

The proof that our algorithm for the security \second model outputs the S*BGP routing outcome in this model follows from the combination of the lemmas below. Let $R_{v_i}$ be the route from $v_i$ induced by the algorithm's computed next-hops.

\begin{lemma}\label{lem_get_secure_stable_provider}  Under S*BGP routing, the route of every AS added to $\mathcal{I}$ in FSCR is guaranteed to stabilize to the route $R_{v_i}$. \end{lemma}
\begin{proof}
The proof is essentially the proof of Lemma~\ref{lem_get_stable_provider} (where now all routes must be secure).
\end{proof}

\begin{lemma}  Under S*BGP routing, the route of every AS added to $\mathcal{I}$ in FCR is guaranteed to stabilize to the route $R_{v_i}$. \end{lemma}
\begin{proof}
As in proof of Lemma~\ref{lem_get_stable_provider}, this lemma is proved via induction on the FCR iteration. Consider the first iteration. Let $v_i$ be the AS chosen at this iteration and let $v_j$ be a next-hop on a route in $v_i$'s $\BPR$ set. Observe that it must be that either $v_j=d$ or $v_j$'s route was fixed in FSCR (for otherwise, $v_j$ would have been selected in FCR before $v_i$). Hence, Lemma~\ref{lem_get_secure_stable_provider} (and the fact that $d$'s route is trivially fixed) implies that under S*BGP convergence each such $v_j$'s route will stabilize at some point and from that point onwards $v_i$ will repeatedly select $R_{v_i}$ (see similar argument in Lemma~\ref{lem_get_stable_provider}). Now, let us assume that for every AS chosen in iterations $1,\ldots,r$ the statement of the lemma holds. Let $v_i$ be the AS chosen at iteration $r+1$ of FCR. Consider $v_i$'s $\BPR$ set at that time and consider again an AS $v_j$ that is $v_i$'s next-hop on some route in $v_i$'s $\BPR$ set. Notice that $v_j$'s route must either have been fixed in FCR at some iteration in $\{1,\ldots,r\}$ (if $v_j$ has a shorter perceivable customer route than $v_i$) or in FSCR (if $v_j$ has a secure customer route to $d$). Hence, by the induction hypothesis, at some point in the S*BGP convergence process, $v_j$'s route converges to $R_{v_j}$ for every such AS $v_j$. As before, from that point in time onward $v_i$ will repeatedly select $R_{v_i}$.
\end{proof}

\begin{lemma} Under S*BGP routing, the route of every AS added to $\mathcal{I}$ in FPeeR is guaranteed to stabilize to the route $R_{v_i}$. \end{lemma}
\begin{proof}
The proof is identical to that of Lemma~\ref{lem_get_stable_peer}.
\end{proof}

\begin{lemma} \label{lem_get_secure_stable_customer} Under S*BGP routing, the route of every AS added to $\mathcal{I}$ in FSPrvR is guaranteed to stabilize to the route $R_{v_i}$. \end{lemma}
\begin{proof}
The proof is essentially the proof of Lemma~\ref{lem_get_stable_customer} (where now all routes must be secure).
\end{proof}

\begin{lemma}  Under S*BGP routing, the route of every AS added to $\mathcal{I}$ in FPrvR is guaranteed to stabilize to the route $R_{v_i}$. \end{lemma}
\begin{proof}
As in the proof of Lemma~\ref{lem_get_stable_provider}, we prove this lemma by induction on the number of FPrvR iterations. Consider the first iteration. Let $v_i$ be the AS chosen at this iteration, let $v_j$ be a next-hop of $v_i$ on some route $R$ in $v_i$'s $\BPR$ set, and let $Q$ be the suffix of $R$ at $v_j$. Observe that either $Q$ is a customer/peer route, in which case $v_j$'s route was fixed before $FSPrvR$ or $Q$ is a secure provider route, in which case $v_j$'s route was fixed in $FSPrvR$. We can now use Lemma~\ref{lem_get_secure_stable_customer} and an argument similar to that in the proof of Lemma~\ref{lem_get_stable_provider} to conclude that $v_i$'s route will indeed converge to $R_{v_i}$ at some point in the S*BGP routing process.

Now, let us assume that for every AS chosen in iterations $1,...,r$ the statement of the lemma holds. Let $v_i$ be the AS chosen at iteration $r+1$ of FPrvR and consider $v_i$'s $\BPR$ set at this time. Let $v_j$ again be a next-hop of of $v_i$ on some route $R$ in $v_i$'s $\BPR$ set, and let $Q$ be the suffix of $R$ at $v_j$. Observe that if $Q$ is a provider route then $v_j$'s route must have been fixed in either FSPrvR or in FPrvR at some point in iterations $\{1,\ldots,r\}$. If, however, $Q$ is either a customer route or a peer route $v_j$'s route must have been fixed in either FCR or FPeeR. Hence, by the previous lemmas and the induction hypothesis, under S*BGP convergence, every such $v_j$'s route will eventually converge to $R_{v_j}$. Again, we can conclude that $v_i$'s route too will converge to $R_{v_i}$.
\end{proof}

\subsubsection{Correctness of algorithm for security \first.}

The proof that our algorithm for the security \first model outputs the S*BGP routing outcome in this model follows from the combination of the lemmas below (whose proofs is almost identical to the proof for the other two models and is therefore omitted). Again, let $R_{v_i}$ be the route from $v_i$ induced by the algorithm's computed next-hops.

\begin{lemma} Under S*BGP routing, the route of every AS added to $\mathcal{I}$ in FSCR is guaranteed to stabilize to the route $R_{v_i}$. \end{lemma}

\begin{lemma} Under S*BGP routing, the route of every AS added to $\mathcal{I}$ in FSPeeR is guaranteed to stabilize to the route $R_{v_i}$. \end{lemma}

\begin{lemma} Under S*BGP routing, the route of every AS added to $\mathcal{I}$ in FSPrvR is guaranteed to stabilize to the route $R_{v_i}$. \end{lemma}

\begin{lemma} Under S*BGP routing, the route of every AS added to $\mathcal{I}$ in FCR is guaranteed to stabilize to the route $R_{v_i}$. \end{lemma}

\begin{lemma} Under S*BGP routing, the route of every AS added to $\mathcal{I}$ in FPeeR is guaranteed to stabilize to the route $R_{v_i}$. \end{lemma}

\begin{lemma} \label{lem_get_stable_customer_1st} Under S*BGP routing, the route of every AS added to $\mathcal{I}$ in FPrvR is guaranteed to stabilize to the route $R_{v_i}$. \end{lemma}

\section{Bounds on happy ASes.}

We use the three algorithms in Appendix~\ref{apx:alg:3}-\ref{apx:alg:1} to compute upper and lower bounds on the set of happy ASes (as discussed in Section~\ref{sec:metric}), for a given attacker-destination pair $(m,d)$, set of secure ASes $S$ and routing model.  To do this, each algorithm records, for every AS discovered in the BFS computation, whether (1) all routes in its $\BPR$ at that iteration lead to the destination, or (2) all these routes lead to the attacker or (3) some of these routes lead to the destination and others to the attacker. The number of ASes in the \first category is then set to be a lower bound on the number of happy ASes. The total number of ASes in the \first and \third category is set to be an upper bound on the number of happy ASes. % The total number of ASes in the \second category are all the unhappy sources.

The correctness of this approach follows from the correctness of our algorithms (Appendix~\ref{apx:alg:corr}), and the fact that all the routes in the $\BPR^r(v_i)$ of a node $v_i$ at iteration $r$ have the same length, type, and are either all secure or insecure, so the \TB criteria completely determines which of these routes are chosen. As such, ASes in the \first category choose legitimate routes (and are happy) regardless of the \TB criteria, ASes in the \second category choose attacked routes (and are unhappy) regardless of the \TB criteria, and whether ASes in the \third category are happy completely depends on the \TB criteria.

\section{BGP convergence.}\label{apx:converge}

Taken together, Lemmas~\ref{lem_get_stable_provider}-\ref{lem_get_stable_customer_1st} proven in Appendix~\ref{apx:alg:corr} above imply Theorem~\ref{thm:convergence}; that is, when all ASes prioritize secure routes the same way, convergence to a single stable routing state is guaranteed, regardless of which ASes adopt S*BGP, even in presence of attacks discussed in Section~\ref{sec:threat}.

\iffalse

\subsection{Optimizing Simulations}

With respect to every destination $d$ our simulations include the following computations:

\begin{enumerate}
\item The BGP routing outcome with respect to every possible pair $(m,d)$ (to compute the partitions);
\item The S*BGP routing outcome for every possible $(m,d)$ in each of our $3$ S*BGP routing models and for every deployment set $S$ considered in the paper (to compute the happy ASes and as part of the computation for quantifying the success of downgrade attacks);
\item The S*BGP routing outcome for every possible $(m,d)$ when $S=\emptyset$ in each of our $3$ S*BGP routing models (as part of the computation for quantifying the success of downgrade attacks).
\end{enumerate}

Hence, the overall complexity of our simulations is $O(|M||D|(|V|+|E|)$. We optimize the running time of our simulations in two ways:

\begin{itemize}
\item {\bf Re-using information.} Instead of running multiple computations ``from scratch'' our simulations often re-use information and pass it on from one computation to the next (e.g., an AS that is doomed with respect to a specific attacker-destination pair $(m,d)$ will not route to $d$ regardless of the deployment scenario $S$).
\item {\bf Parallelization.} We run these computations in parallel across all destinations $d$. Our code was written in C++ and parallelization was achieved with MPI on a BlueGene and Blacklight supercomputers.
\end{itemize}
%\end{comment}

\fi

%\input{correctness_partitions_happy_downgrades_new}
\section{Partitions.}\label{apx:partitions}

Recall from Section~\ref{sec:doomedImm} that a source AS $s$ is \emph{protectable} if S*BGP can affect whether or not it routes to the legitimate destination $d$ or the attacker $m$;  the source AS $s$ is \emph{doomed} (resp. \emph{immune}) if it always routes to the attacker $m$ (resp. routes to the legitimate destination $d$), regardless of how S*BGP is deployed in the network.  In the security \first model all ASes are assumed to be protectable (we do this to avoid the complications discussed in Appendix~\ref{apx:sec1parts}).  In this section we describe how we compute the sets of immune, doomed, and protectable ASes with respect to an attacker-destination pair $(m,d)$ in the security \second and \third models. To do this, we set $S=\emptyset$ and compute the BGP routing outcome for that $(m,d)$ pair using the algorithm in Section~\ref{apx:alg:3}.

\subsection{Computing partitions: security \third}

To determine the partitions for the security \third model, this algorithm records, for every AS discovered in the BFS computation whether (1) all routes in its $\BPR$ set at that iteration lead to the destination, or (2) all these routes lead to the attacker or (3) some of these routes lead to the destination and others to the attacker.  We classify ASes in the \first category as immune, ASes in the \second category as doomed, and ASes in the \third category as protectable.  We show below that this indeed coincides with our definitions of immune, doomed, and protectable ASes in Section~\ref{sec:doomedImm} for the security \third model.

The following allows us to prove the correctness of our algorithm for computing partitions:
\begin{cor} \label{cor:3rdInsensitive_1}
In the security \third routing model, for any destination $d$, attacker $m$, source $s$ and deployment $S \subseteq V$, $s$ will stabilize to a route of the same type and length as any route in $\BR(s, \emptyset, m, d)$.
\end{cor}

\begin{proof} This follows from the correctness of our algorithm for computing routes in the security \third model (Appendix~\ref{apx:alg:3}).  Note that because in the security \third model route security is prioritized below path length, all routes in $\BPR^r(s)$ must be contained in $\BPR(s,$ $\emptyset, m, d)$, where $\BPR^r(s)$ is the set of best perceivable routes of $s$ during iteration $r$ of the subroutine FCR, FPeeR or FPrvR of our algorithm, when $\BPR(s, S, m, d)$ contains customer, peer or provider routes respectively.  Recognize that by the correctness of our algorithm, $s$ must  stabilize to a route in $\BPR^r(s)$ for some iteration $r$ of exactly one of these subroutines.

Therefore, any $s$ that has customer routes in $\BPR(s,$ $\emptyset, m, d)$ will be ``fixed'' to a route in the FCR subroutine for any choice of $S$.  Similarly, if $s$ has peer (\resp provider) routes in $\BPR(s, \emptyset, m, d)$, it will be ``fixed'' to a route in the FPeeR (\resp FPrvR) subroutine for any choice of $S$.  Therefore, the type of the route will be fixed to the same type as that of the $\BPR(s,$ $\emptyset, m, d)$ for all $S$. Moreover, when we choose to ``fix'' the route of $s$ in the appropriate subroutine, we do so by selecting $s$ with a shortest routes out of all the sources that have not been ``fixed'', and regardless of $S$, and it follows the the length of the route will be the same for all $S$.
\end{proof}

Corollary \ref{cor:3rdInsensitive_1} tells us that for determining whether $s$ is immune, doomed or protectable in security \third model, it is sufficient to keep track of all the routes of the best type and shortest length of $s$ (i.e. all the routes in $\BPR(s, \emptyset, m, d)$), because $s$ is guaranteed to  stabilize to one of these routes.  Therefore, if all such routes are legitimate (\resp attacked), then $s$ will always  stabilize to a legitimate (\resp attacked) route under any S*BGP deployment $S$, so $s$ must be immune (\resp doomed).  However, if some of these routes are legitimate and some are attacked, then whether $s$  stabilizes to a route to $m$ or $d$ depends on deployment $S$, so $s$ must be protectable.

\subsection{Computing partitions: security \second}

The algorithm for determining partitions for the security \second model is slightly different from that used when security is third.  We still use the algorithm from Appendix~\ref{apx:alg:3}, except that now, for every AS discovered in the BFS computation we need to keep track of all perceivable routes in its $\PR$ set that are of the same \emph{type} as the routes in its $\BPR$ set.  %This must be done throughout the duration of each subroutine of the algorithm for computing routes in this model.  This is important because route length in this model is less important then security.
We keep track of whether (1) all such routes lead to the destination, or (2) all such routes lead to the attacker or (3) some of these routes lead to the destination and others to the attacker.  We classify ASes in the \first category as immune, ASes in the \second category as doomed, and ASes in the \third category as protectable. 
 
The following allows us to prove the correctness of this algorithm:
\begin{cor} \label{cor:2ndInsensitive}
In the security \second routing model, for any destination $d$, attacker $m$ source $s$ and deployment $S \subseteq V$, $s$ will  stabilize to a route of the same type as any route in $\BPR(s, \emptyset, m, d)$.
\end{cor}
\begin{proof} This follows from the correctness of our algorithm for computing routes in the security \second model (Appendix~\ref{apx:alg:2}).
Because in the security \second model security is prioritized above route length, but below route type, all the routes in $\BPR(s)^r$ must be contained in the set of routes in $\PR(s, m, d)$ that are of the same type as routes in $\BPR(s, \emptyset, m, d)$.  Recall that $\BPR^r(s)$ is the set of best perceivable routes of $s$ during iteration $r$ of the appropriate subroutines FSCR and FCR, FPeeR, or FSPrvR and FPrvR of our algorithm, if $\BPR(s, S, m, d)$ contains customer, peer or provider routes respectively.  Also, note that by the correctness of our algorithm, $s$ must  stabilize to a route in $\BPR^r(s)$ for some iteration $r$ of exactly one of these subroutines.

Therefore, if $s$ has customer routes in $\BPR(s, \emptyset, m, d)$, it  will be ``fixed'' to a route during either the FSCR or FCR subroutines of this algorithm for any choice of $S$.  If $s$ has peer routes in $\BPR(s, \emptyset, m, d)$, it will be ``fixed'' to a route in the FPeeR  subroutine for any choice of $S$.  Finally, if $s$ has provider  routes in $\BPR(s, \emptyset, m, d)$, it will be ``fixed'' to a route in either FSPrvR or FPrvR subroutines for any choice of $S$.
\end{proof}

Corollary \ref{cor:2ndInsensitive} tells us that to determine if $s$ is immune, doomed or protectable in security \second model, it is sufficient to keep track of all the routes of the best type of $s$ (\ie. all $s$'s perceivable routes of the same type as routes in $\BPR(s, \emptyset, m, d)$), because $s$ is guaranteed to  stabilize to one of these routes. Therefore, if all such perceivable routes are legitimate (\resp attacked) , then $s$ must  stabilize to a legitimate (\resp attacked) route under any S*BGP deployment $S$, so $s$ must be immune (\resp doomed).  However, if some of these routes are legitimate and some are attacked, then whether $s$  stabilizes to a route to $m$ or $d$ depends on deployment $S$, so $s$ must be protectable.

\subsection{Computing partitions: security \first}\label{apx:sec1parts}

In this paper we assume that all source ASes are protectable in security \first model (see \eg Figure~\ref{fig:partitions}).  Technically, however, there can be doomed and immune ASes in the security \first model, in a few exceptional cases; here we argue the the number of such ASes is negligible.

\myparab{Doomed ASes.} We can characterize doomed ASes as follows.
\begin{obs} \label{obs:doomed1st}
In the security \first model, for a particular destination-attacker pair $(d,  m)$, a source AS $v_i$ is doomed if and only if every one of its perceivable routes $\PR(v_i,m,d)$ contains $m$.
\end{obs}
If every perceivable route from $v_i$ to $d$ contains $m$, then there is no S*BGP deployment scenario that could result in $v_i$ being happy.  On the other hand, if $v_i$ is not doomed, then there must be at least one S*BGP deployment scenario that results in $v_i$ being happy, in which case $v_i$ must select a route to $d$ that does not contain $m$.

ASes that single-homed to the attacking AS $m$ are certainly doomed, per
Observation~\ref{obs:doomed1st}. There are $11,953$ and $11,585$ single-homed stub ASes (without peers) for the regular and the IXP-augmented graphs respectively.  As an upper bound, we consider only the former number. Recall from Section~\ref{sec:metric} that our security metric is an average of happy sources, where the average is taken over all sources and all appropriate destination-attacker pairs. It follows that that for any one destination, there can be at most  $11,953$  doomed single-homed ASes when summed over all attackers and all sources.  Therefore, the fraction of doomed sources does not exceed $.001\%$ and $.01\%$ when considering all and only non-stub attackers respectively.
While Observation~\ref{obs:doomed1st} suggests there could be other doomed nodes (other than the just the single-homed stub ASes), however, the Internet graph is sufficiently well-connected to ensure that the number of such ASes is small.

\myparab{Immune ASes.} A similar characterization is possible for immune ASes.
\begin{obs} \label{obs:doomed1st}
In the security \first model, for a particular destination-attacker pair $(d,  m)$, a source AS $v_i$ is immune if every one of its perceivable routes $\PR(v_i,\emptyset,m)$ contains $d$.
\end{obs}
As we discussed above, immune ASes tend to be single-homed stub ASes.

\section{Protocol downgrade attacks.}\label{apx:downgrades}

In Section~\ref{sec:pda} we discussed how protocol downgrades can occur in the security \second and \third model.  We now prove Theorem~\ref{thm:noPDA1st}, that shows that protocol downgrade attacks are avoided in the security \first model; that is, every AS $s$ that uses a secure route that does not contain the attacker $m$ under normal conditions, will continue to use that secure route when $m$ launches its attack.

\begin{proof}[of Theorem~\ref{thm:noPDA1st}]  The theorem follows from the correctness of the algorithm in Appendix~\ref{apx:alg:1} for computing routes when security is \first.  Suppose the set of secure routes is $S$. Consider an AS $s$ who has its secure route $R_s$ fixed during the FSCR, FSPeerR, FSPrR subroutine of the algorithm in Appendix~\ref{apx:alg:1} when the set of secure ASes is $S$ and the attacker is $m=\emptyset$ (\ie during normal conditions, when there is no attack).
If $R_s$ does not contain $m$, then $s$ will have its route fixed to exactly the same secure route $R_s$ during the FSCR, FSPeeR, FSPrvR subroutine of the algorithm in Appendix~\ref{apx:alg:1} when the set of secure ASes is $S$ and $m$ attacks. This follows because all routes that contain $m$ must be fixed \emph{after} the FSCR, FSPeeR, FSPrvR portions of the algorithm (since, by definition, all routes containing $m$ must be insecure during $m$'s attack).  An inductive argument shows that all ASes on route $R_s$ will therefore be fixed to the same route that they used in normal conditions, and the theorem follows.
\end{proof}

\subsection{Computing protocol downgrades.}

To quantify the success of protocol downgrade attacks with respect to an attacker-destination pair $(m,d)$ and a set of secure ASes $S$, we need to first establish which ASes have a secure route to the destination under normal conditions, that is, when there is no attack. To do this, we compute the S*BGP routing outcome when there is no attacker (setting $m=\emptyset$ for the set $S$) for the specific model under consideration. The algorithm records for every AS discovered in this BFS computation whether (1) all routes in its $\BPR$ set at that iteration is secure or (2) all these routes are insecure.  We then compute the S*BGP routing outcome for the pair $(m, d)$ for the set $S$ (for the specific model under consideration)). Again, the algorithm records for every AS discovered in this BFS computation whether (1) all routes in its $\BPR$ set at that iteration are secure or (2) all these routes are insecure. We conclude that a protocol-downgrade attack against an AS is successful if that AS falls in the \first category in the first of these computations and in the \second category in the second computation.  The correctness of this approach follows from the correctness of our algorithms in Appendix~\ref{apx:algos}.

\section{Monotonicity}\label{apx:mono}

In Section~\ref{sec:examples} and Appendix~\ref{apx:damages} we showed that collateral damage is possible in the security \second and \first models.  We now prove Theorem~\ref{thm:nodamage3rd} that shows that collateral damage does not occur in the security \third model; that is, for any destination $d$, attacker $m$, source $s$ and S*BGP deployments $T$ and $S \subseteq T$, if $s$  stabilizes to a legitimate route in deployment $S$, then $s$ stabilizes to a legitimate route in deployment $T$.

\begin{proof}[of Theorem~\ref{thm:nodamage3rd}] The theorem follows from the correctness of our algorithm for computing routing outcomes when security is \third (Appendix~\ref{apx:alg:3}). First,  an inductive argument shows that every AS $s$ that the algorithm ``fixes'' to a secure route in deployment $S$ is also ``fixed'' to a secure route in $T$; it follows that all such ASes stabilize to a legitimate route in both $S$ and $T$.  Next we argue that every AS $s$ that the algorithm ``fixes'' to an insecure legitimate route in $S$ is also fixed to a legitimate route in $T$. There are two cases: (a) if $s$ is fixed to a secure route in $T$, it uses a legitimate route,  (b) otherwise, an inductive argument shows that the algorithm computes the same next hop $\Nxt(s)$ for $s$ in both deployments $T$ and $S$, and since the route was legitimate in $S$, it will be legitimate in $T$ as well.
\end{proof}

\section{Simulations}\label{apx:sim}

Our simulations compute the following for each destination $d$:
\begin{enumerate}
\item The S*BGP routing outcome for each of our $3$ S*BGP routing models and for every deployment set $S$ considered in the paper (to enable computations that quantify protocol downgrade attacks );
\item The BGP routing outcome with respect to every possible pair $(m,d)$ and with $S=\emptyset$ (to compute partitions into doomed/immune/protectable ASes, and to determine which ASes where happy in the baseline scenario where $S=\emptyset$);
\item The S*BGP routing outcome for every possible $(m,d)$ in each of our $3$ S*BGP routing models and for every deployment set $S$ considered in the paper (to compute the happy ASes, to detect phenomena like collateral benefits and damages, and as part of computations that quantify protocol downgrade attacks );
\end{enumerate}
To do this, we use the algorithms in Appendix~\ref{apx:alg:3}-\ref{apx:alg:1}, where the we execute the FCR, FSCR, FPeeR, FSPeeR, FPrvR, and FSPrvR subroutines using breath-first searches.  The overall complexity of our simulations is therefore $O(|M||D|(|V|+|E|)$ for each deployment $S$. We optimize the running time of our simulations in two ways:

\myparab{Re-using information.} Instead of running multiple computations ``from scratch'' our simulations often re-use information and pass it on from one computation to the next (\eg an AS that is doomed with respect to a specific attacker-destination pair $(m,d)$ will not route to $d$ regardless of the deployment scenario $S$, \etc).

\myparab{Parallelization.} We run these computations in parallel across all destinations $d$. Our code was written in C++ and parallelization was achieved with MPI on a BlueGene and Blacklight supercomputers.
\section{Hardness Results}\label{apx:hardness_results}

\begin{figure}
\begin{center}
  \includegraphics[width=2in]{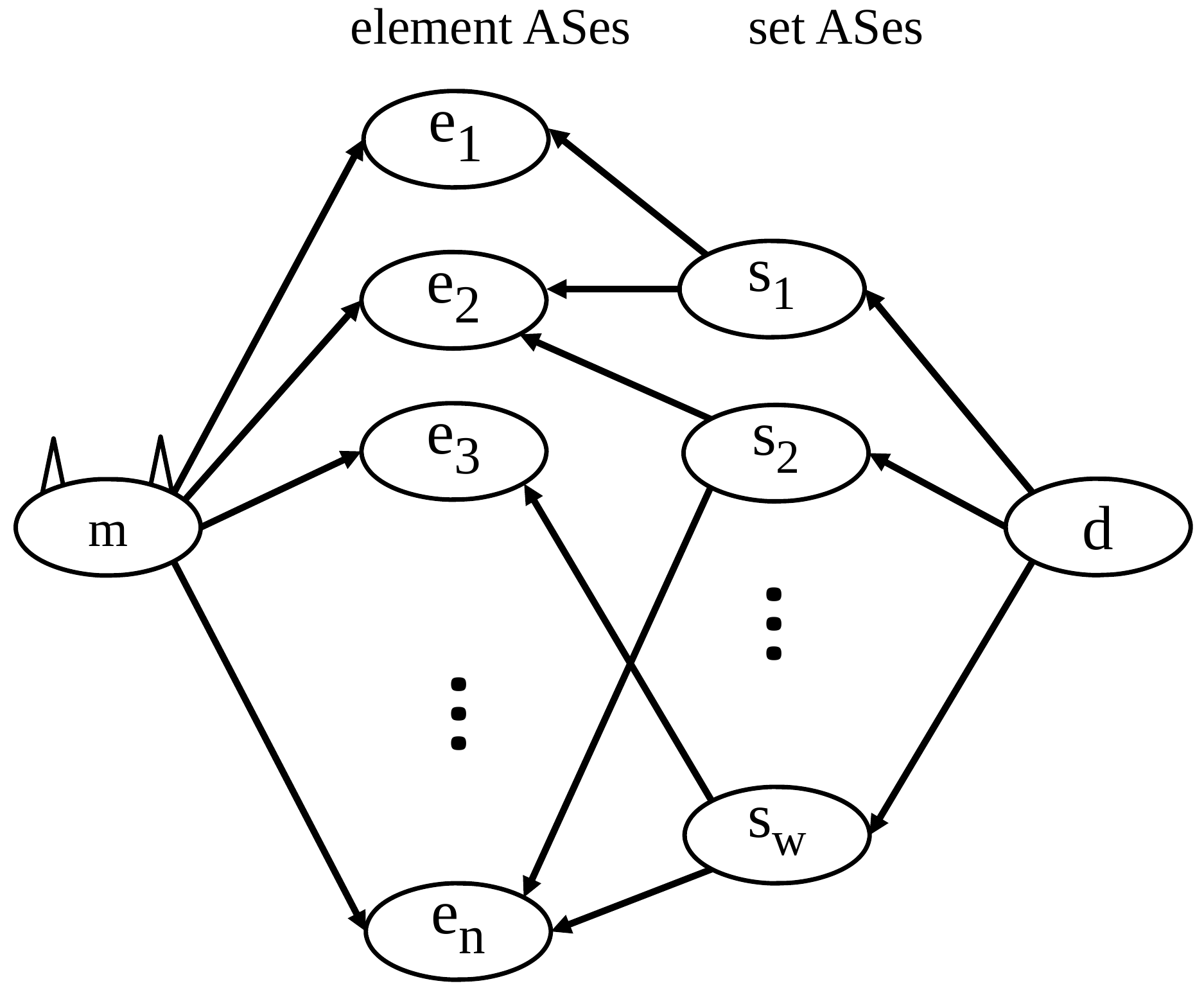}
  \caption{Reduction}\label{fig:hard}
\end{center}
\end{figure}

We prove Theorem~\ref{thm:hard}, that shows that the ``Max-k-Security'' problem is NP-hard in each of our three routing models.  Recall from Section~\ref{sec:hard}, in the ``Max-k-Security'' problem, we are given an AS graph,  $G=(V,E)$, a specific attacker-destination pair $(m,d)$, and a parameter $k>0$, find a set of ASes $S$ of size $k$ that maximizes the total number of happy ASes.

To prove Theorem~\ref{thm:hard}, we consider a slightly different problem that we will call the  ``Decisional-k-$\ell$-Security'' problem (D$k\ell$SP): Given an AS graph, a specific attacker-destination pair $(m, d)$, and parameters $k>0$ and $1 \leq \ell\leq |V|$, determine if there is a set of secure ASes $S$ of size $k$ that results in at least $\ell$ happy ASes.  Notice that this problem is in NP (since we can check the number of happy ASes in polynomial time given the algorithms discussed in Appendix~\ref{apx:algos}) and is certainly poly-time reducible to ``Max-k-Security''.  Therefore, the following theorem implies Theorem~\ref{thm:hard}:

\begin{theorem}
D$k\ell$SP is NP-Complete in each of our three routing policy models.
\end{theorem}

\begin{proof}
We present a poly-time reduction from the Set Cover Decisional Problem (SCDP).  In SCDP, we are given a set $N$ with $n$ elements, a family $F$ of $w$ subsets of $N$ and an integer $\gamma\leq w$, and we must decide if there exist $\gamma$ subsets in the family $F$ that can cover all the elements in $N$.

Our reduction is shown in Figure~\ref{fig:hard}. For each element $e_i \in N$ in the SCDP instance, we create an AS $e_i$ in our D$k\ell$SP instance and connect it to the attacker via a provider-to-customer edge.
For each subset $s_j \in F$, we create an AS $s_j$ in our our D$k\ell$SP instance and connect it to the destination $d$ via a provider-to-customer edge.
We connect AS $e_i$ to AS $s_j$ via a provider-to-customer edge if $e_i \in s_j$ in the SCDP problem.  Moreover, we require that every $e_i$'s has a tiebreak criteria \TB that prefers the route through $m$ over any route through any $s_j$.
%
%Finally, we add $x=\ell-n-w$ additional ``padding ASes'' $p_1,...,p_x$ that we connect directly to the the destination $d$ via customer-provider edges. Observe that both the padding ASes $p_1,...,p_x$ and the set ASes $s_1,...,s_w$  are immune; they will be happy regardless of the set of secure ASes.
%
Notice that the perceivable routes at every $e_i$ are of the same length and type; namely, two-hop customer routes.
Finally, we let $\ell=n+w+1$, and let $k=n+\gamma+1$.

%To do this, we observe that in the D$k\ell$SP problem, there will be a set of immune ASes $immune_{m}$ of ASes whose most preferred routes do not contain $m$, a set of doomed ASes $doomed_m$ of ASes whose most preferred routes contain $m$, and a set protectable sources $protectable_{m}$.  Recall that after $k$ ASes deploy S*BGP, only the protectable sources in $protectable_{m}$  could be persuaded to select routes that do not contain $m$.  In this scenario, we say that a protectable AS is protected if it becomes happy after gaining a secure route to $d$.

Suppose that our SCDP instance has a $\gamma$-cover.  We argue that this implies that our corresponding D$k\ell$SP should be able to choose a set $S$ of $k$ secure ASes that ensure that at least $\ell$ ASes are happy.  The following set $S$ of secure ASes suffice: $S=\{d, e_1,...,e_n\} \cup \{s_j | s_j \text{ is in the $\gamma$ cover}\}$.  Notice that $S$ is of size $k=n+\gamma+1$, and results in exactly $\ell=n+w+1$ happy ASes. (This follows because $d$ is happy by definition, all the set ASes $s_1,...,s_w$ are happy regardless of the choice of $S$, and all the element ASes $e_1,...,e_n$ choose legitimate routes to the destination because they have secure routes to $d$ by construction.)

On the other hand, suppose we are able to secure exactly $k$ ASes while ensuring that $\ell$ ASes are happy. First, note that all the set ASes $s_1,...,s_w$ and the destination AS and are immune; they are happy regardless of which ASes are secure. Next, note that if any of the $n$ element ASes $e_1, ... ,e_n$  are insecure, then by construction it will choose a route to the attacker and be unhappy, and we will have less than $\ell$ happy nodes.  Similarly, if the destination $d$ is insecure, by construction all of the element ASes will choose an insecure route to the attacker.  Thus, if we secure all the element ASes and the destination, we have $k-1-n=\gamma$ remaining ASes to secure; by construction, these must be distributed amongst the set ASes, and thus we will have a $\gamma$-cover by construction.

Finally, note that this result holds in all three secure routing models; the reduction is agonistic to how ASes rank security in their route preference decisions, since the perceivable routes at every element AS $e_i$ have the same length and type.
\end{proof}

To extend this result to multiple destinations $D$ and attackers $M$, we can show the hardness of the following variant of the ``Max-k-Security'' problem: given $G(V,E)$, sets $M,D \subseteq V$ and an integer $k$, the objective is to maximize the average number of happy ASes across all $(m,d)$ pairs in $M\times D$.  The argument is the same as the above, except that now we create multiple copies of the $m$ and $d$ nodes (and their adjacent edges) in Figure~\ref{fig:hard}, and let $M$ be the copies of the $m$ nodes and $D$ be the copies of the $d$ nodes.

 %hardness \input{apx_hard_new}
\newpage
\section{The IXP-Augmented Graph}\label{apx:ixp}

We repeated our experiments on the IXP-augmented graph described in Section~\ref{sec:policies} to obtain the following results.

\subsection{Plots for Section~\ref{sec:partitions}.}

\begin{figure}[h]
\begin{center}
    \subfigure[]{
    \includegraphics[width=2in]{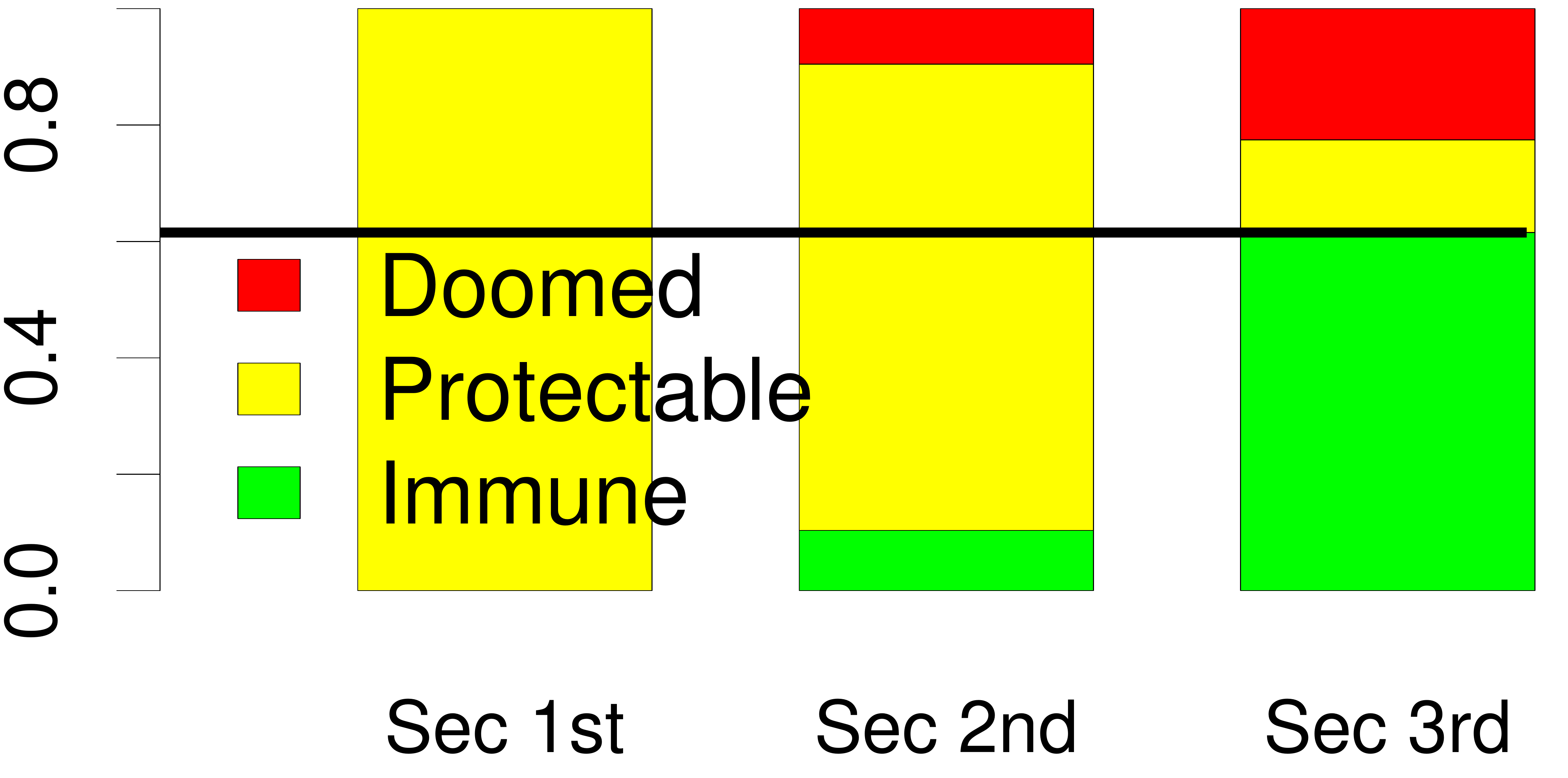}
    \label{fig:partitions:ixp}
    }
    \subfigure[]{
    \includegraphics[width=2.7in]{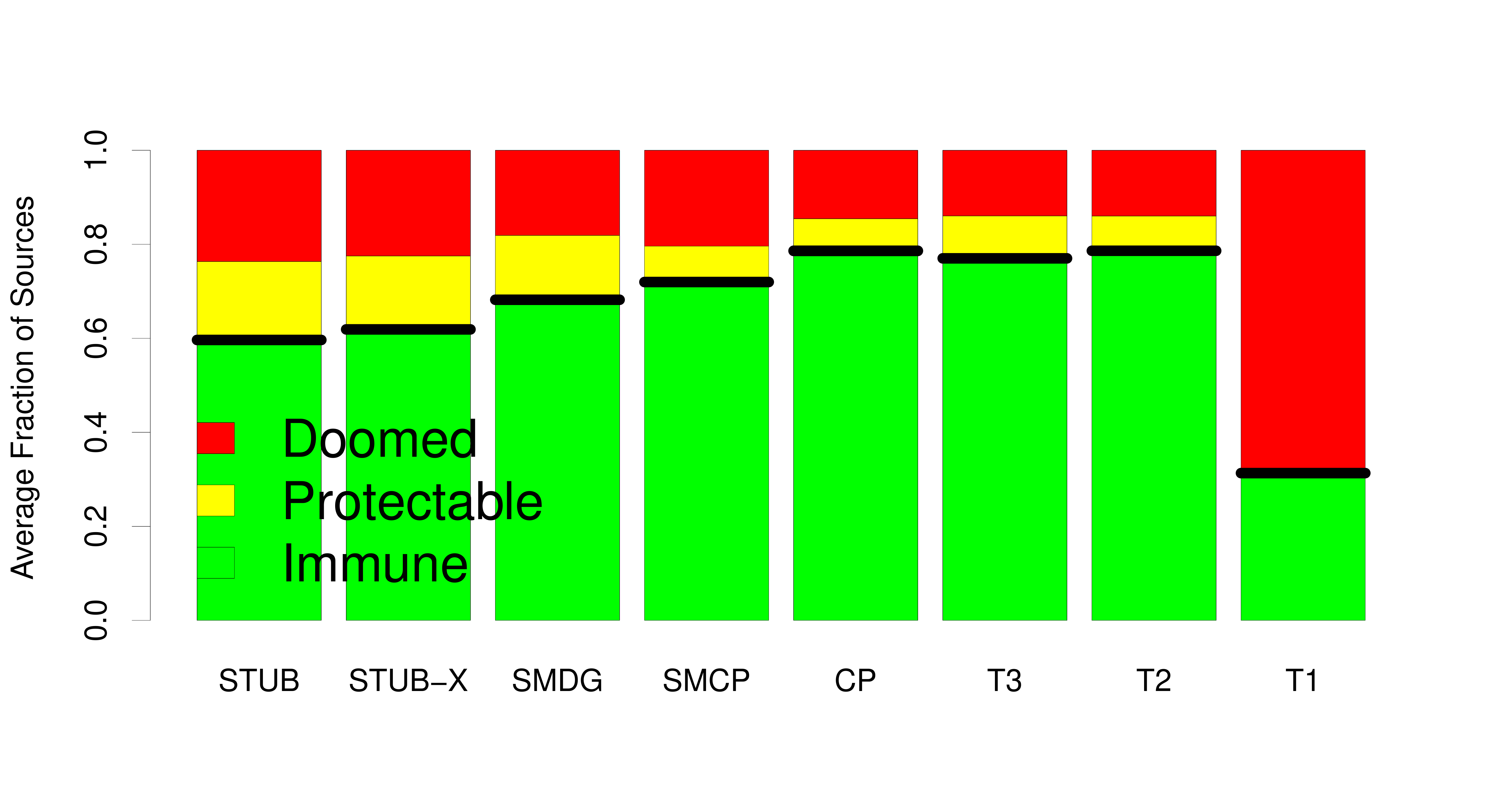}
    \label{fig:partitions:dest:sl:ixp}
    }
    \subfigure[]{
    \includegraphics[width=2.7in]{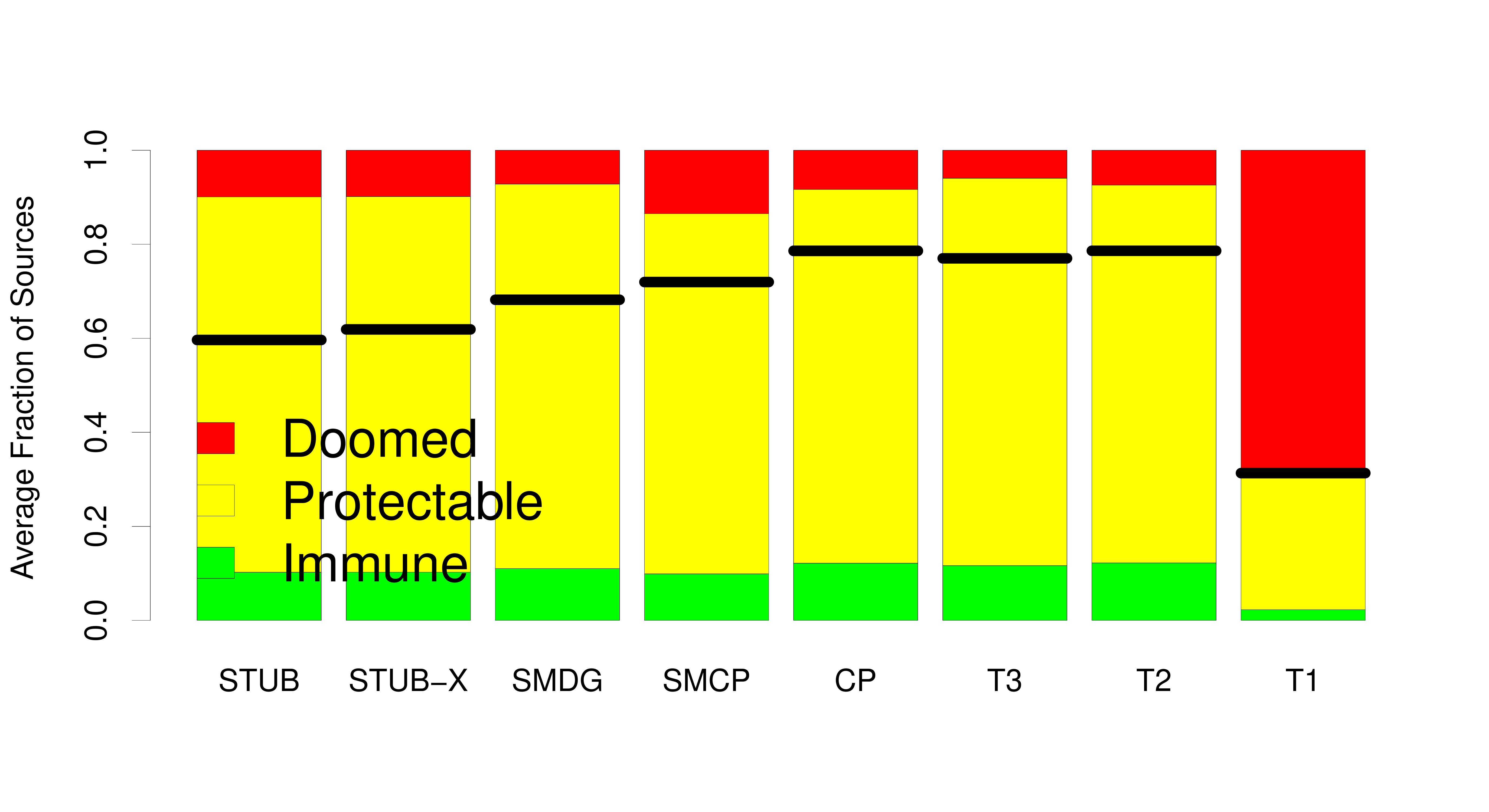}
    \label{fig:partitions:dest:ss:ixp}
    }
    \subfigure[]{
    \includegraphics[width=2.7in]{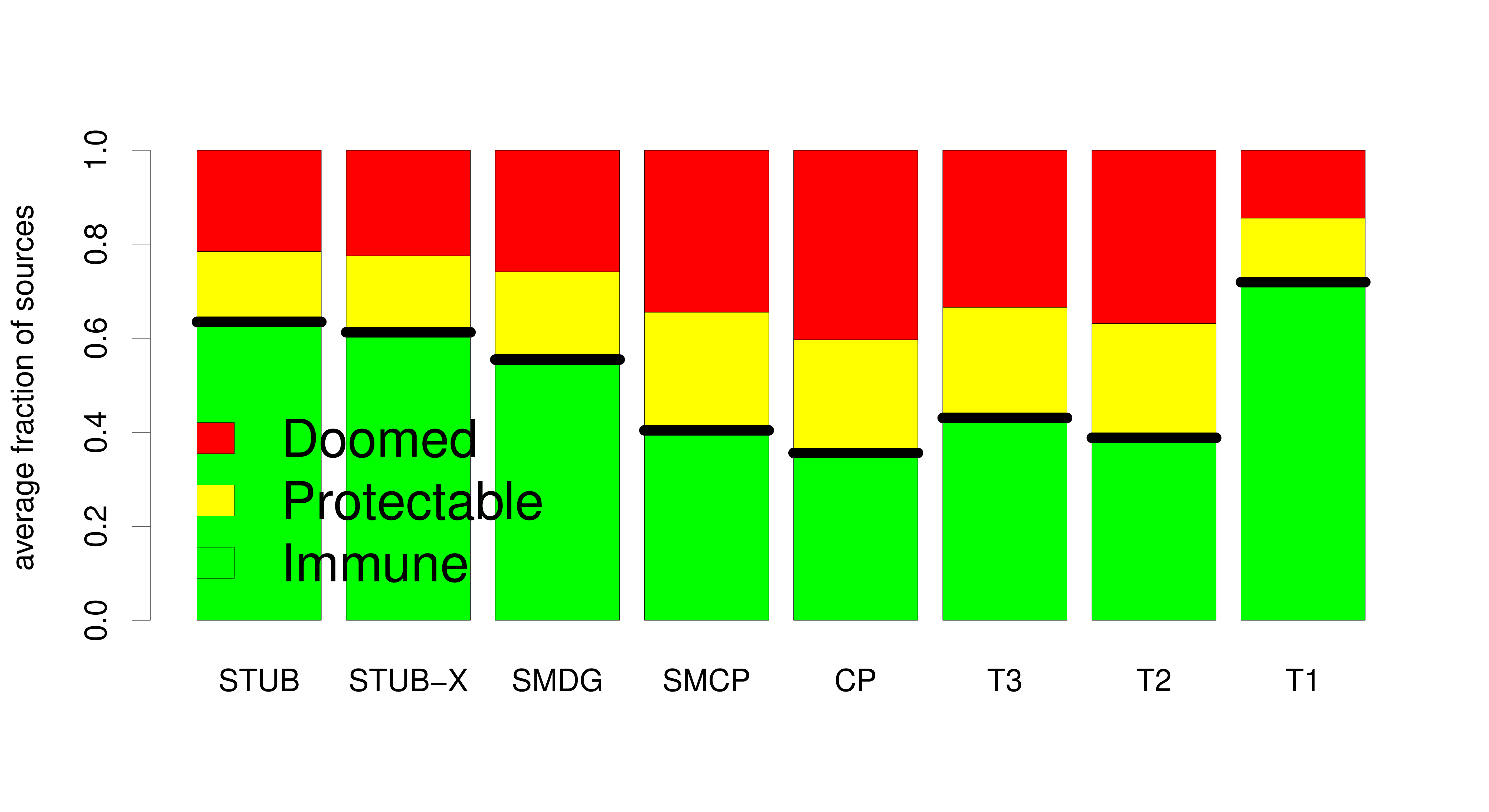}
    \label{fig:partitions:attacker:sl:ixp}
    }
    \caption{Plots for Section~\ref{sec:partitions}, IXP-augmented graph. \subref{fig:partitions:ixp} Partitions.
    \subref{fig:partitions:dest:sl:ixp} Partitions by destination tier. Sec~\third.
    \subref{fig:partitions:dest:ss:ixp} Partitions by destination tier. Sec~\second.
    \subref{fig:partitions:attacker:sl:ixp} Partitions by attacker tier. Sec~\third.
    }
    \end{center}
\end{figure}

\newpage
\subsection{Plots for Section~\ref{sec:results:metric}.}

\begin{figure}[h]
\begin{center}
    \subfigure[]{
    \includegraphics[width=1.2in]{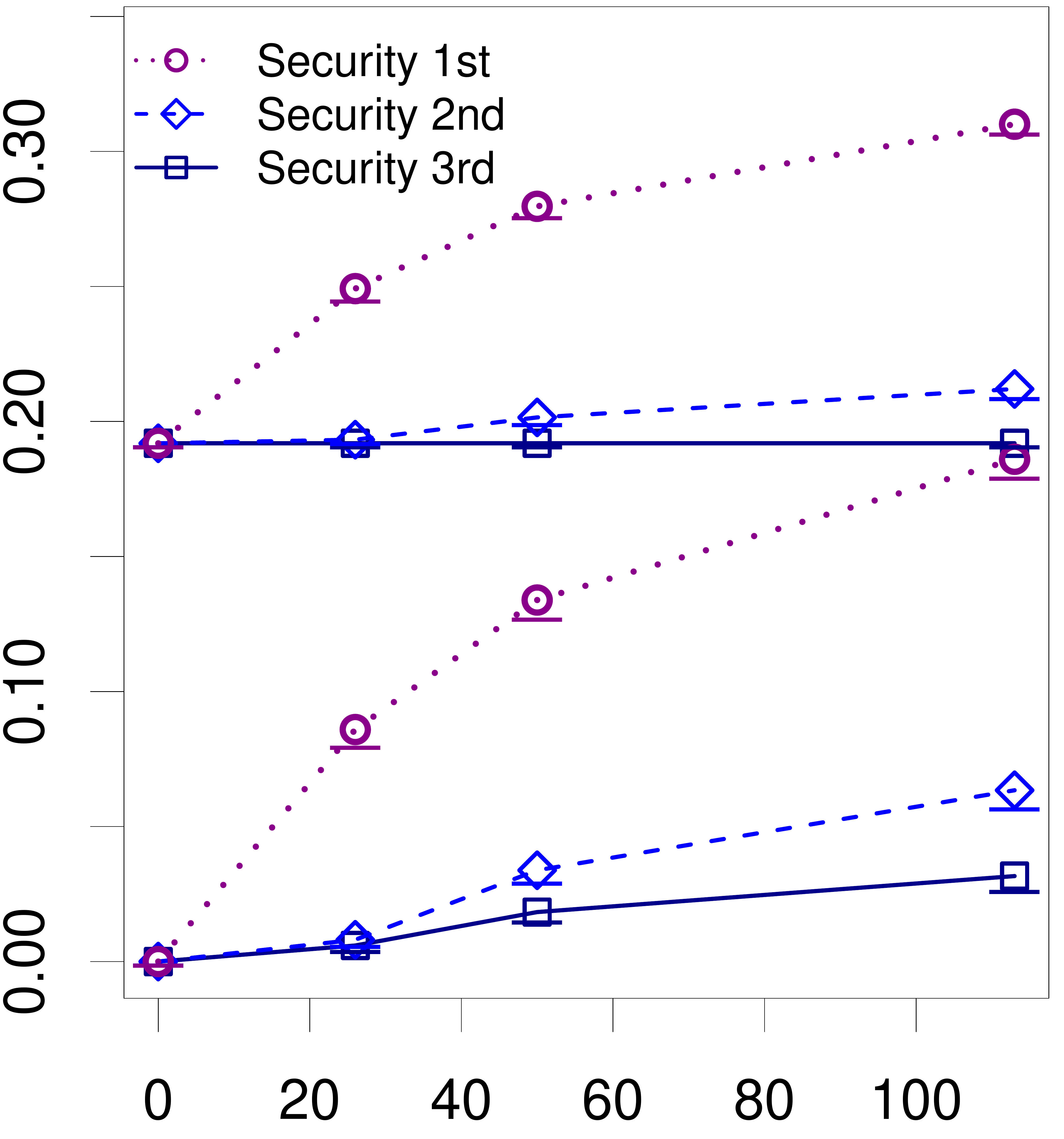}
    \label{fig:metric:large:IXP}
    }
    \subfigure[]{
      \includegraphics[width=1.2in]{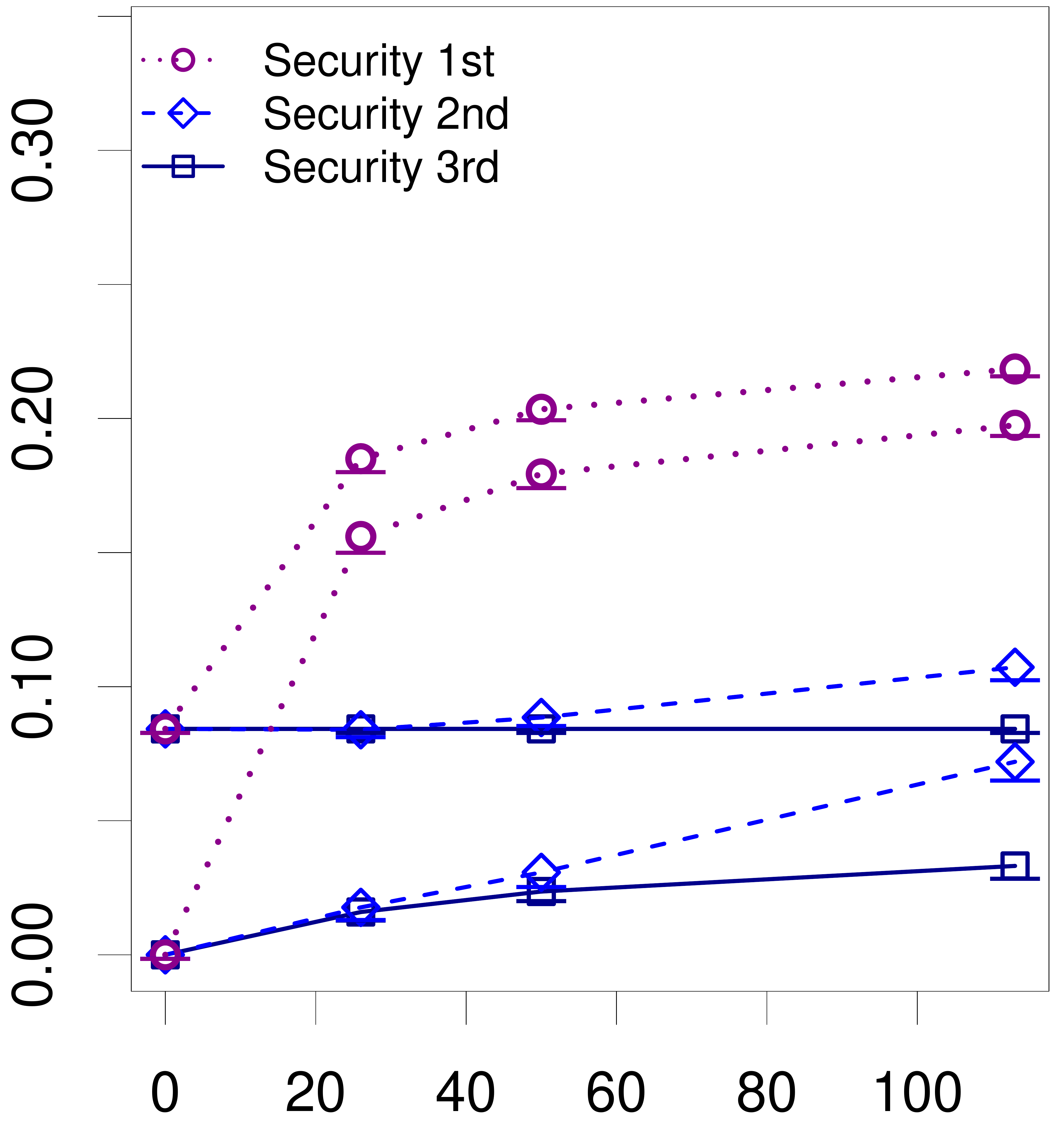} \label{fig:metric:bigdeps:CP:ixp}
    }
     \subfigure[]{
      \includegraphics[width=1.2in]{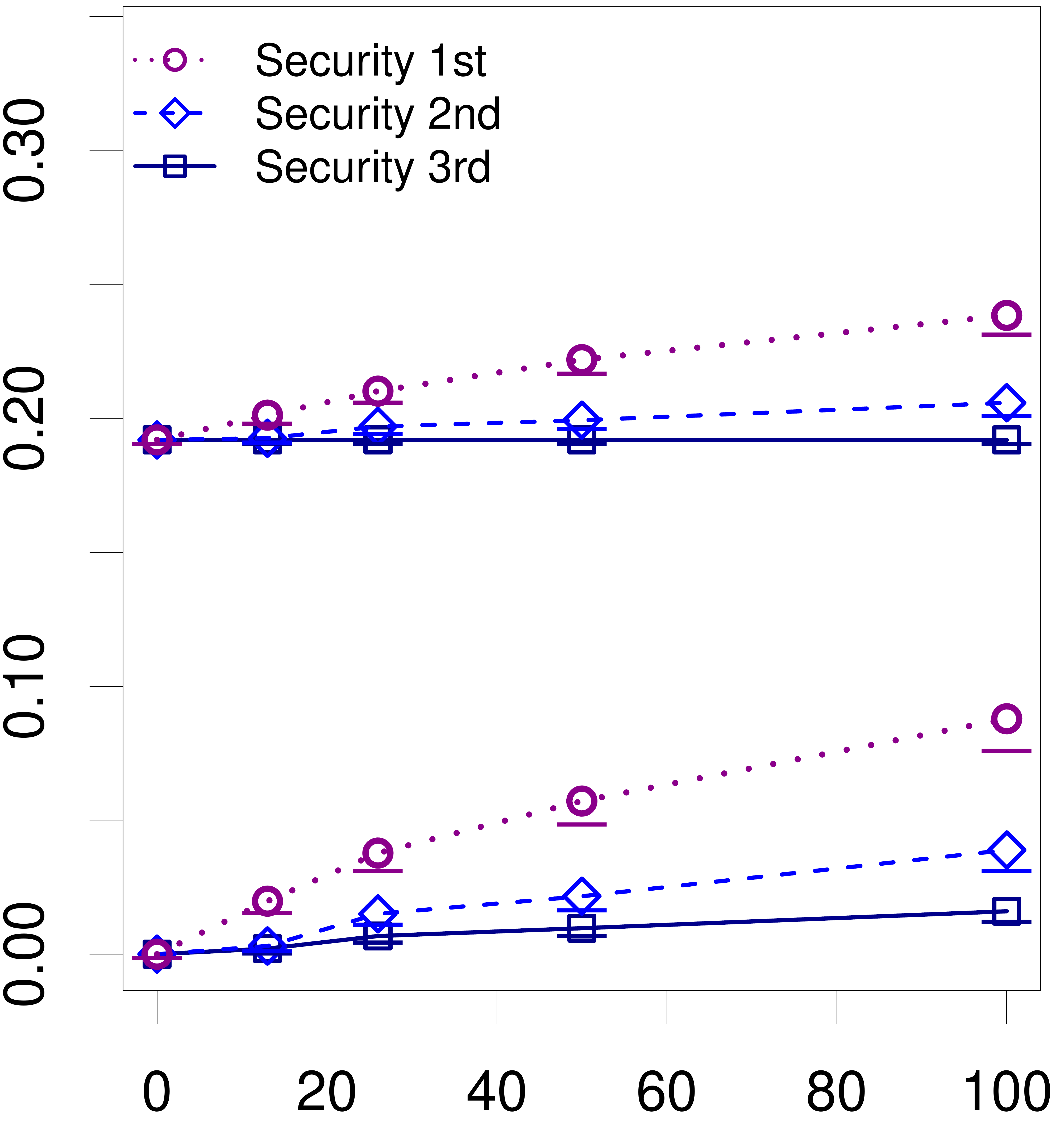}  \label{fig:metric:bigdeps:T2:ixp}    
    }
      \vspace{-3mm}
    \caption{Plots for Section~\ref{sec:results:metric}, IXP-augmented graph. For each plot, the $x$-axis is the number of non-stub, non-CP ASes in $S$ and the ``error bars'' are explained in Section~\ref{sec:simplex}.
    \subref{fig:metric:large:IXP} Tier 1+2 rollout: For each step $S$ in rollout, upper and lower bounds on $H_{M',V}(S) - H_{M',V}(\emptyset)$. 
    \subref{fig:metric:bigdeps:CP:ixp} Tier 1+2+CP rollout: $H_{M',CP}(S) - H_{M',C}(\emptyset)$ for each step in the rollout. 
    \subref{fig:metric:bigdeps:T2:ixp} Tier 2 rollout: $H_{M',D}(S) - H_{M',D}(\emptyset)$ for each step in the T2 rollout.
    }
    \vspace{-5mm}
    \end{center}
\end{figure}

\begin{figure}[h]
    \includegraphics[width=2.5in]{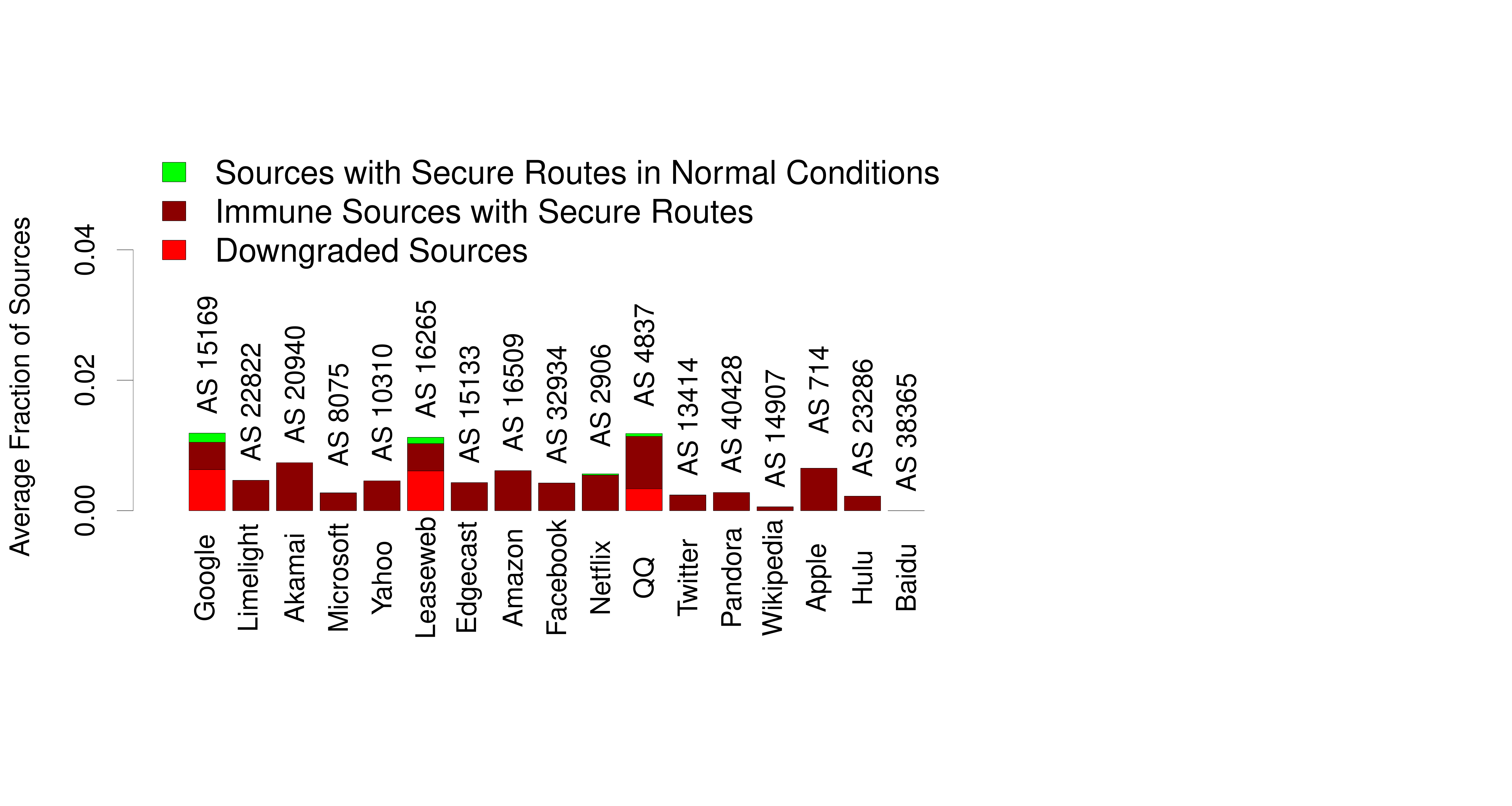}
    \vspace{-3mm}\caption{Plot for Section~\ref{sec:results:metric}, IXP-augmented graph. What happens to secure routes to each CP destination during attack. $S$ is the Tier 1s, the CPs, and all their stubs and security is~\third.}
    \vspace{-3mm}\label{fig:t1_cp_sec_r_breakdown:ixp}
\end{figure}	

\newpage
\begin{figure}[h]
\begin{center}
    \subfigure[]{
    \includegraphics[width=.35\textwidth]{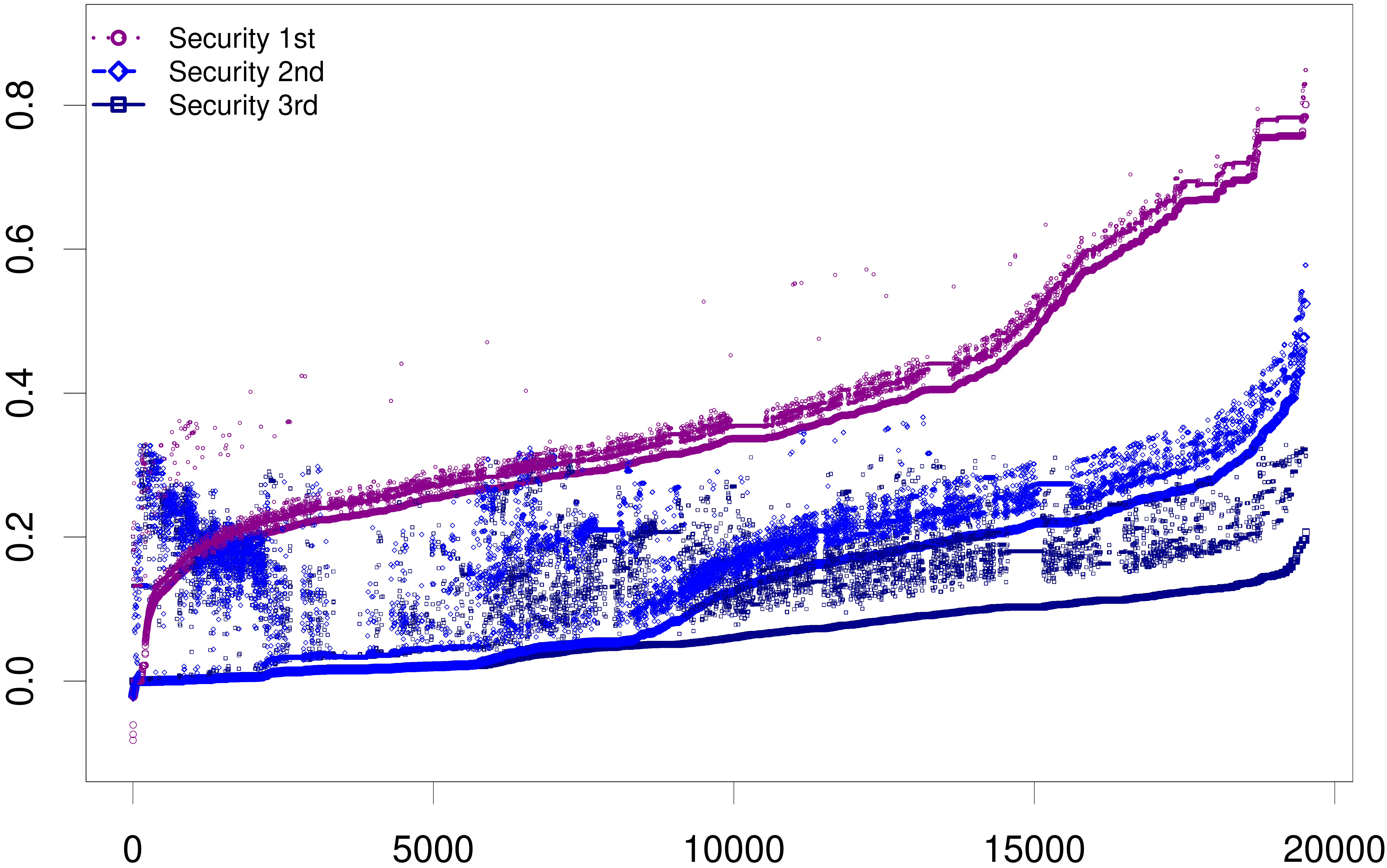}
    \label{fig:seq:ixp}
    }
    \subfigure[]{
    \includegraphics[width=.35\textwidth]{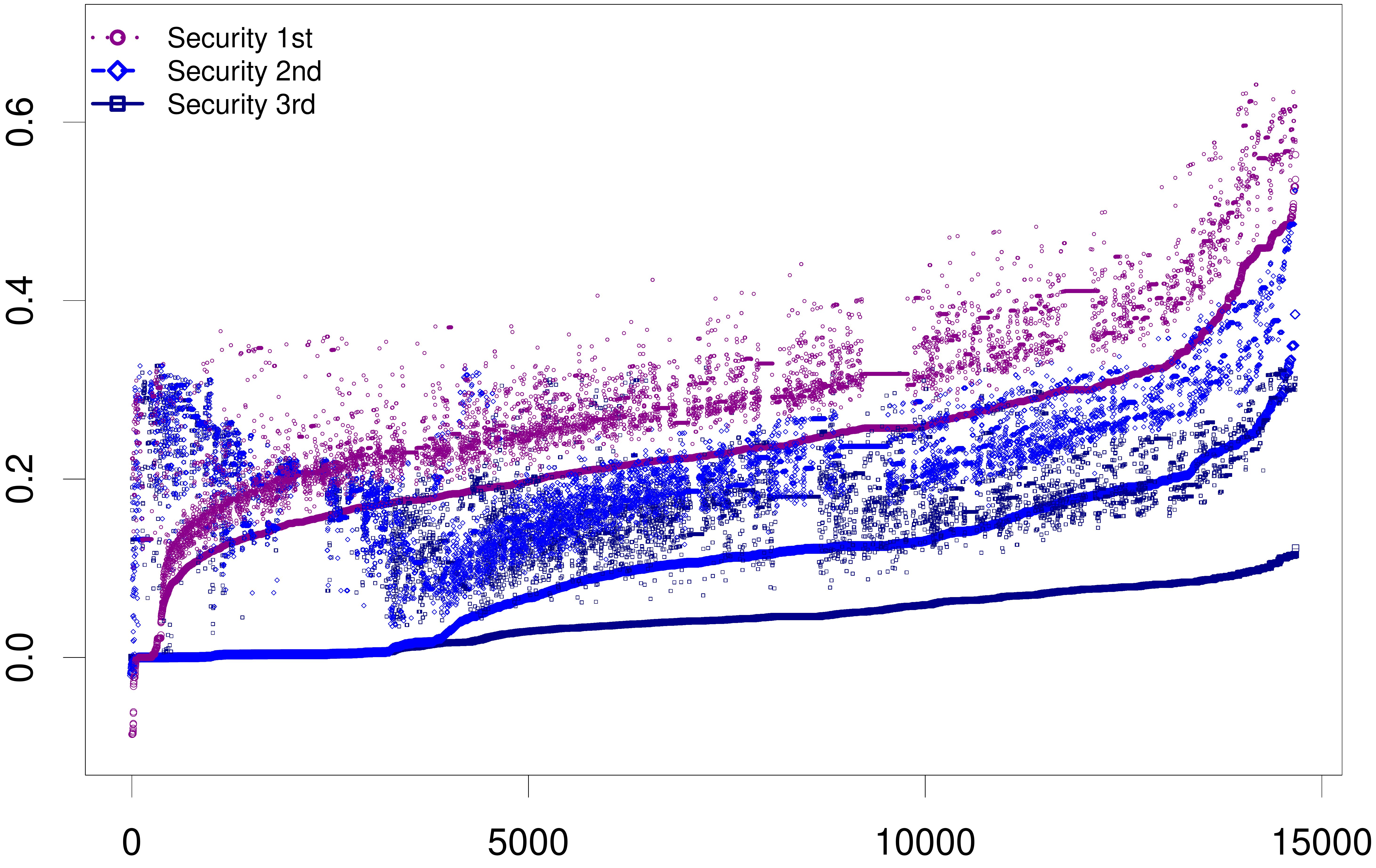}
    \label{fig:seq:T2:ixp}
    }
    \subfigure[]{
    \includegraphics[width=.35\textwidth]{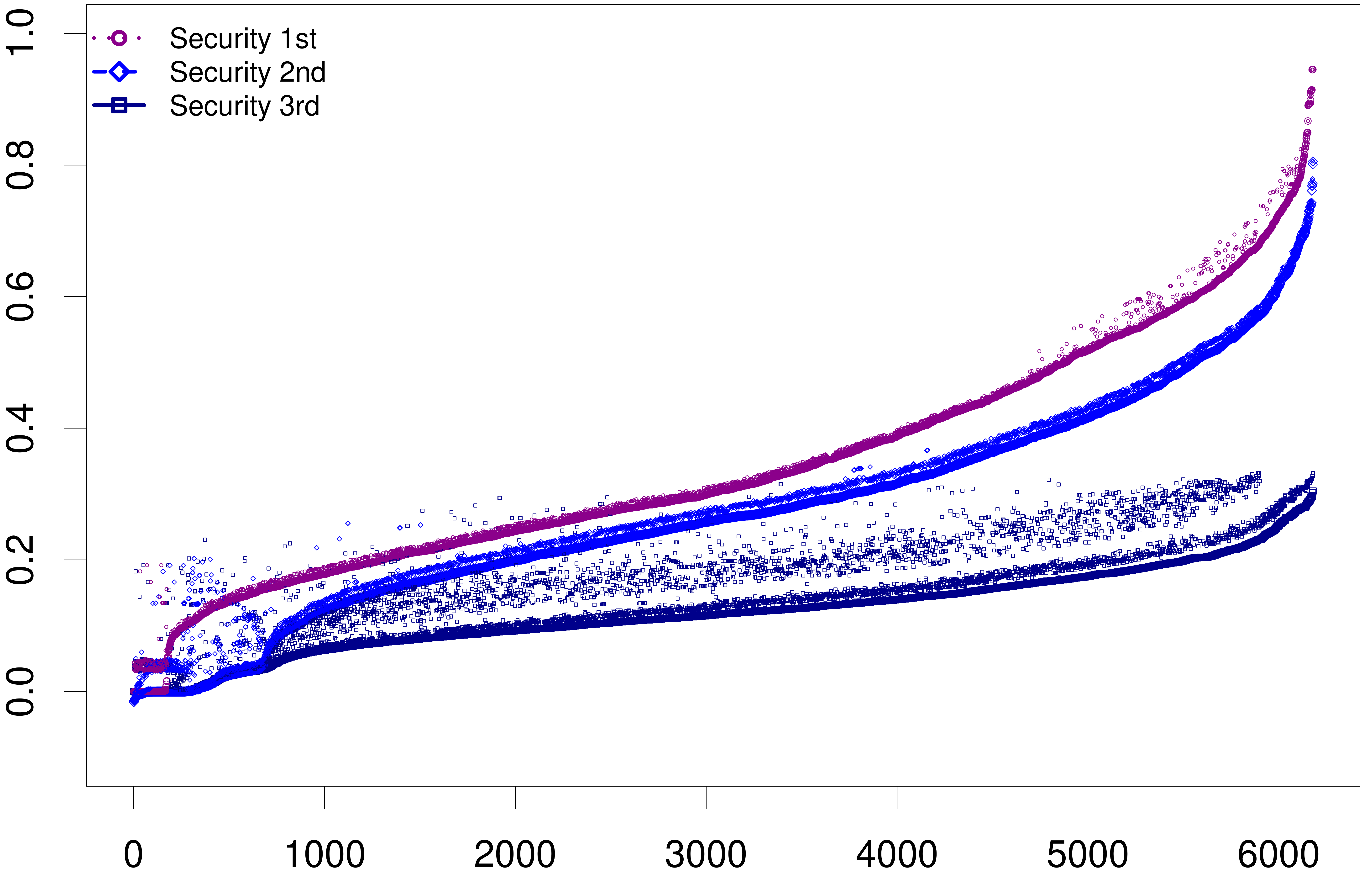}
    \label{fig:seq:nonstub:ixp}
    }
 \vspace{-2mm}
\caption{Plot2 for Section~\ref{sec:results:metric}, IXP-augmented graph. Non-decreasing sequence of $H_{M',d}(S)-H_{M',d}(\emptyset)$ $\forall$ $d \in S$.
 \subref{fig:seq:ixp}  $S$ is all T1s, T2s, and their stubs.
 \subref{fig:seq:T2:ixp}  $S$ is all T2s and their stubs.
  \subref{fig:seq:nonstub:ixp} $S$ is all non stubs.
}
 \vspace{-3mm}
\end{center}
\end{figure}

\newpage
\subsection{Plots for Section~\ref{sec:wrapup}.}

\begin{figure}[h]
\begin{center}
    \includegraphics[width=2in]{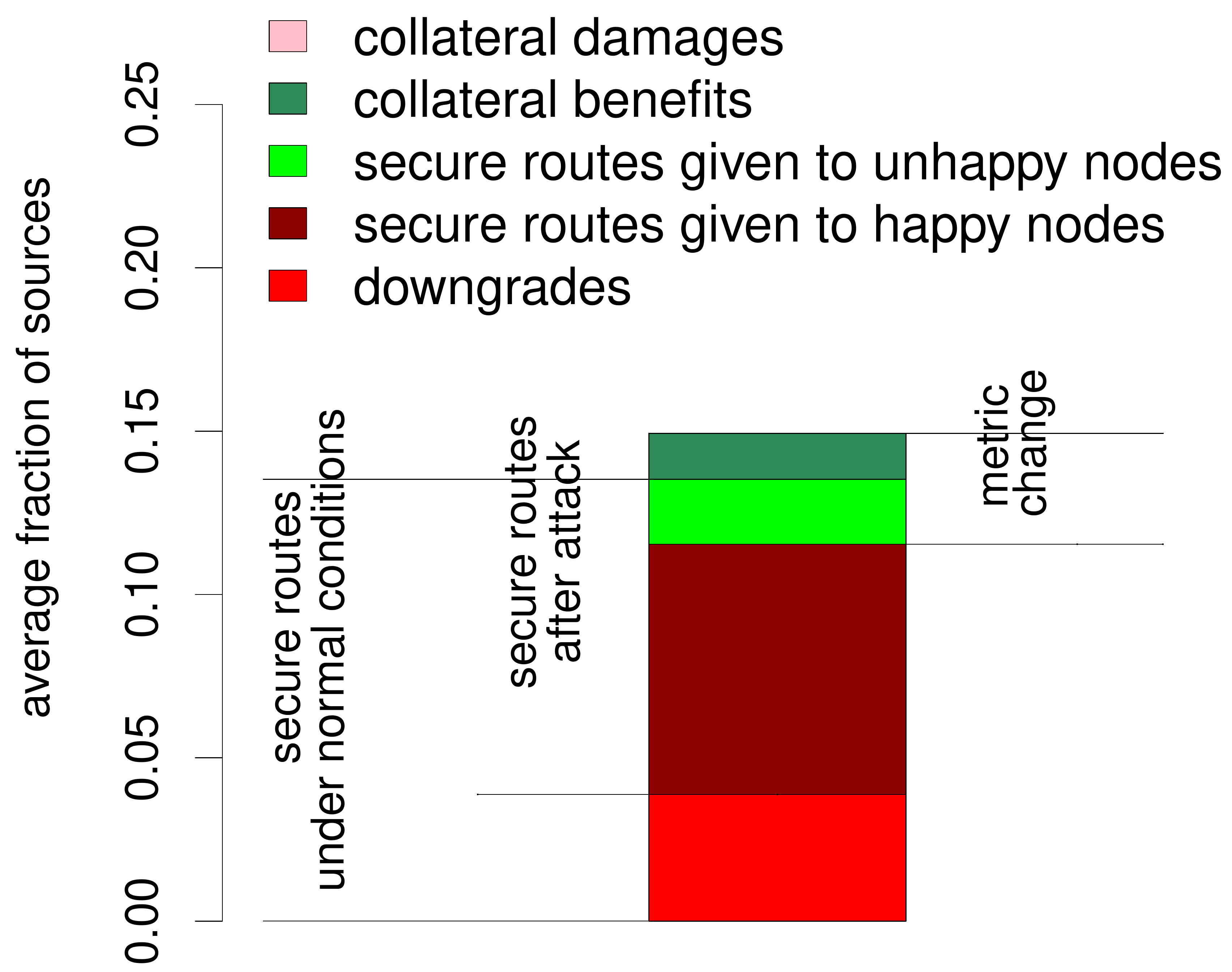}\\
    \includegraphics[width=2in]{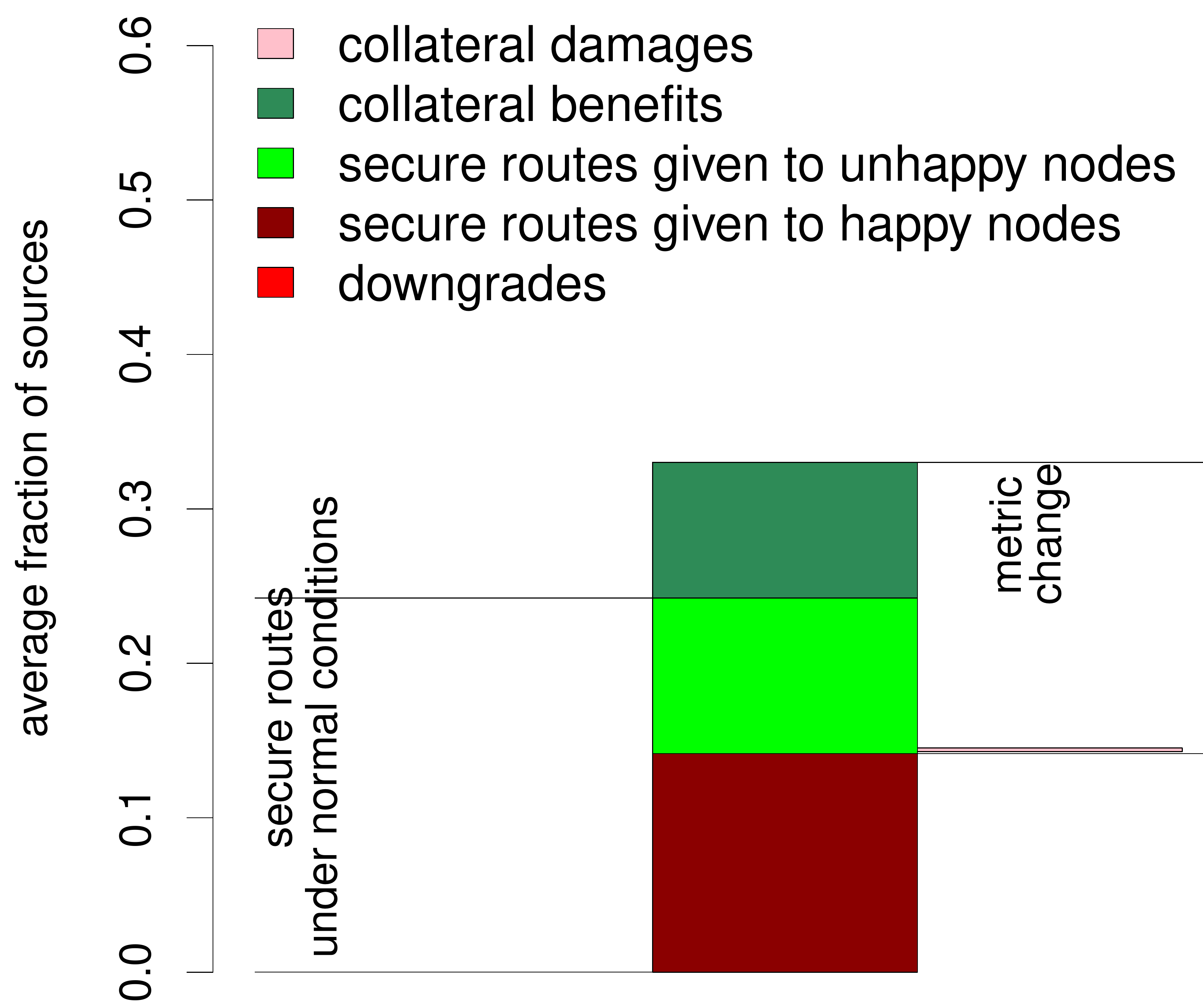}
    \vspace{-3mm}\caption{Plots for Section~\ref{sec:wrapup}, IXP-augmented graph. Changes in the metric explained. Sec~\third (top) and Sec~\first (bottom).}
    \vspace{-5mm}\label{fig:rootcause:ixp}\end{center}
\end{figure}

\newpage\newpage
\section{Sensitivity to Routing policy }\label{apx:robust:policies}\label{apx:stubborn_partitions}

Thus far, all our analysis has worked within the model of local preference (\LP) presented in Section~\ref{sec:insec_policies}. While the survey of \cite{surveyEmail} found that 80\% of network operators do prefer customer routes over peer and provider routes, there are some exceptions to this rule.  Therefore, in this Appendix we investigate alternate models of local preference, and consider how they impact the results we presented in Section~\ref{sec:partitions}; we are currently in the process of extending this sensitivity analysis to the results in Section~\ref{sec:results:metric}-\ref{sec:wrapup}.

\subsection{An alternate model of local preference.}

All our results thus far have used the following model of local preference:

\myparab{Local pref (\LP):} Prefer customer routes over peer routes. Prefer peer routes over provider routes.

\smallskip\noindent
However, \cite{surveyEmail} also found some instances where ASes, especially content providers, prefer shorter \emph{peer} routes over longer customer routes.  For this reason, we now investigate the following model of local preference:

\myparab{Local pref (\LP{k}):} Paths are ranked as follows:
\begin{itemize*}
\item Customer routes of length 1.
\item Peer routes of length 1.
\item ...
\item Customer routes of length $k$.
\item Peer routes of length $k$.
\item Customer paths of length $>k$.
\item Peer paths of length $>k$.
\item Provider paths.
\end{itemize*}
Following the \LP{k} step, we have the \SP and \TB steps as in Section~\ref{sec:insec_policies}. As before, the security \first model ranks \SecP above \LP{k}, the security \second model ranks \SecP between \LP{k} and \SP, and the security \third model ranks \SecP between \SP and \TB.

\myparab{Remark. } We will study this policy variant for various values of $k$; note that letting $k\rightarrow\infty$ is equivalent a routing policy where ASes equally prefer customer and provider routes, as follows:
\begin{itemize*}
\item Prefer peer and customer routes over provider routes.
\item Prefer shorter routes over longer routes.
\item Break ties in favor of customer routes.
\item Use intradomain criteria (\eg geographic location, device ID) to break ties among remaining routes.
\end{itemize*}

\subsection{Results with \LP{2} policy variant.}

We start with an analysis of the \LP{2} policy variant; we are in process of extending these results to other \LP{k} variants.  Here, a peer route of length less than or equal to 2 hops is preferred over a longer customer route.

\myparab{Partitions. } In Figure~\ref{fig:2stubborn:parititions} we show the partitions for the \LP{2} policy variants, for the UCLA graph and for the IXP augmented graph (\cf Figure~\ref{fig:partitions} and Section~\ref{sec:upperLower}).  The thick solid horizontal line shows the fraction of happy source ASes in the baseline scenario (where no AS is secure). As in Section~\ref{sec:upperLower}, we find that with security \third only limited improved improvements in the metric $H_{V,V}(S)$ are possible, relative to the baseline scenario $H_{V,V}(\emptyset)$; $82-71=11\%$ for the UCLA AS graph, and $88-72=13\%$ for the IXP augmented graph, both of which are slightly less than what we saw for our original \LP model.  In the security \second model, we again see better improvements than security \third, but not quite as much as we saw with our original \LP model;  $92-71=21\%$ for the UCLA AS graph, and  $94-72=22\%$ for the IXP augmented graph.  Interestingly, however, we do see one difference between the UCLA AS graph and the IXP augmented graph in this model; namely, we see more immune ASes when security is \second for the IXP augmented graph (41\% \vs 55\%). We discuss the observation in more detail shortly.

\begin{figure}[t]
\begin{center}
\subfigure[]{
\includegraphics [width = 2.7in]{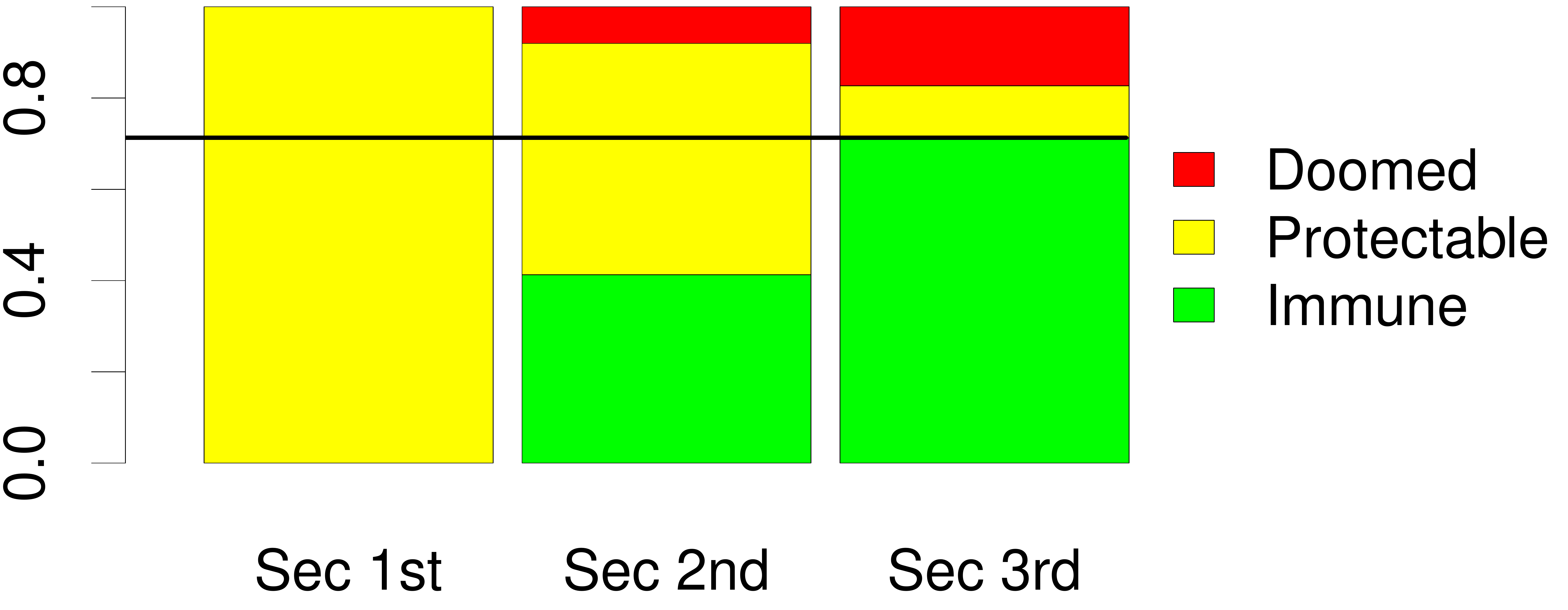} \label{fig:no_ixp_get_parts_stubborn_2}
}
\subfigure[]{
\includegraphics [width = 2.7in]{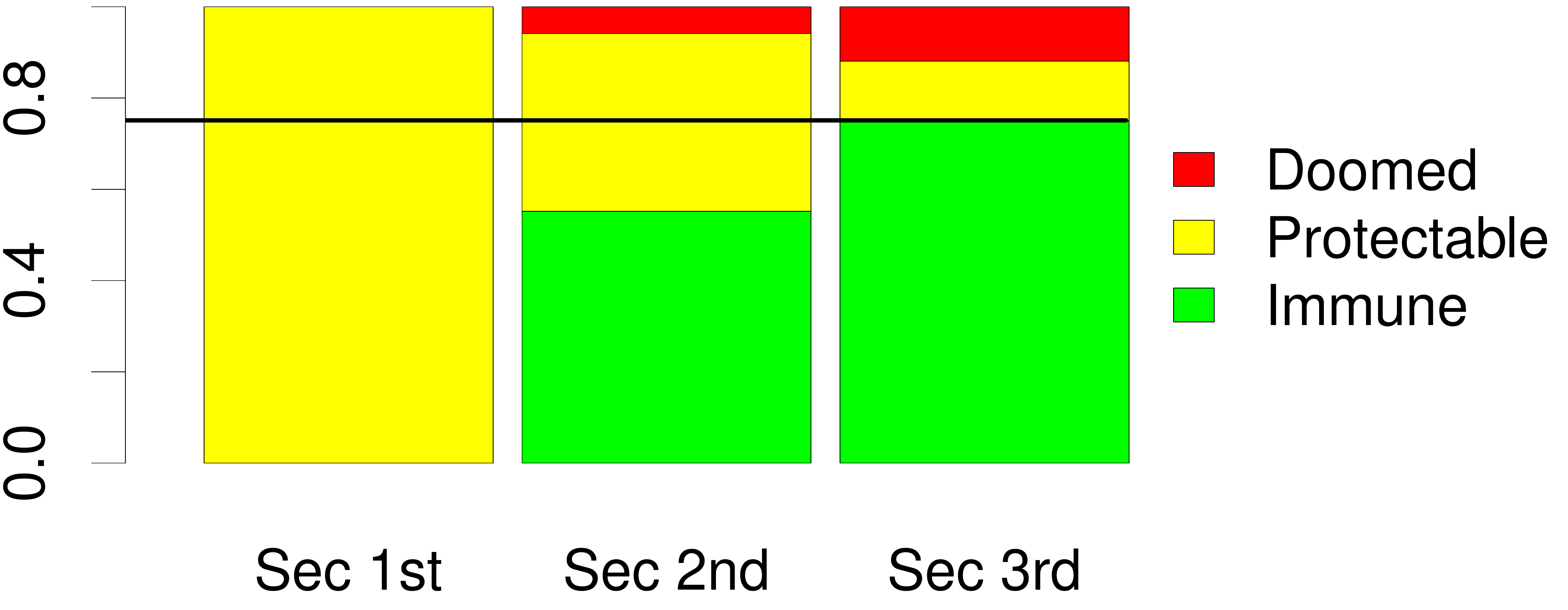} \label{fig:ixp_get_parts_stubborn_2}
}
\caption{Partitions for the \LP{2} policy variant, \subref{fig:no_ixp_get_parts_stubborn_2} UCLA graph \subref{fig:ixp_get_parts_stubborn_2}  IXP-augmented graph.}\label{fig:2stubborn:parititions}
\vspace{-5mm}
\end{center}
\end{figure}

\myparab{Partitions by destination tier.} In Figure~\ref{fig:2stubborn:parititions:dest} we show the partitions broken down by destination tier (see Table~\ref{tab:tiers}) when security is \second and \third for the \LP{2} policy variants, for the UCLA graph and for the IXP augmented graph (\cf Figure~\ref{fig:partitions:dest:sl}, Figure~\ref{fig:partitions:dest:ss} and Section~\ref{sec:partitions:robust}). The thick solid horizontal line shows the fraction of happy source ASes in the baseline scenario (where no AS is secure) for each destination tier. While in Section~\ref{sec:partitions:robust} we found that most destination tiers have roughly the same number of protectable ASes here we see slightly different trends.

\mypara{1. } Most of the protectable nodes are at stub and SMDG (low-degree non-stub ASes) destinations. The higher-degree AS destinations, \ie Tier 2s, Tier 2s, and CPs, have very few protectable ASes but many more immune ASes as compared to the results we obtained for our original \LP model in Figure~\ref{fig:partitions:dest:sl}.  This is even more apparent for the IXP augmented graph in the \LP{2} model.

\begin{figure}[h]
\begin{center}
\subfigure[]{
\includegraphics [scale = .06]{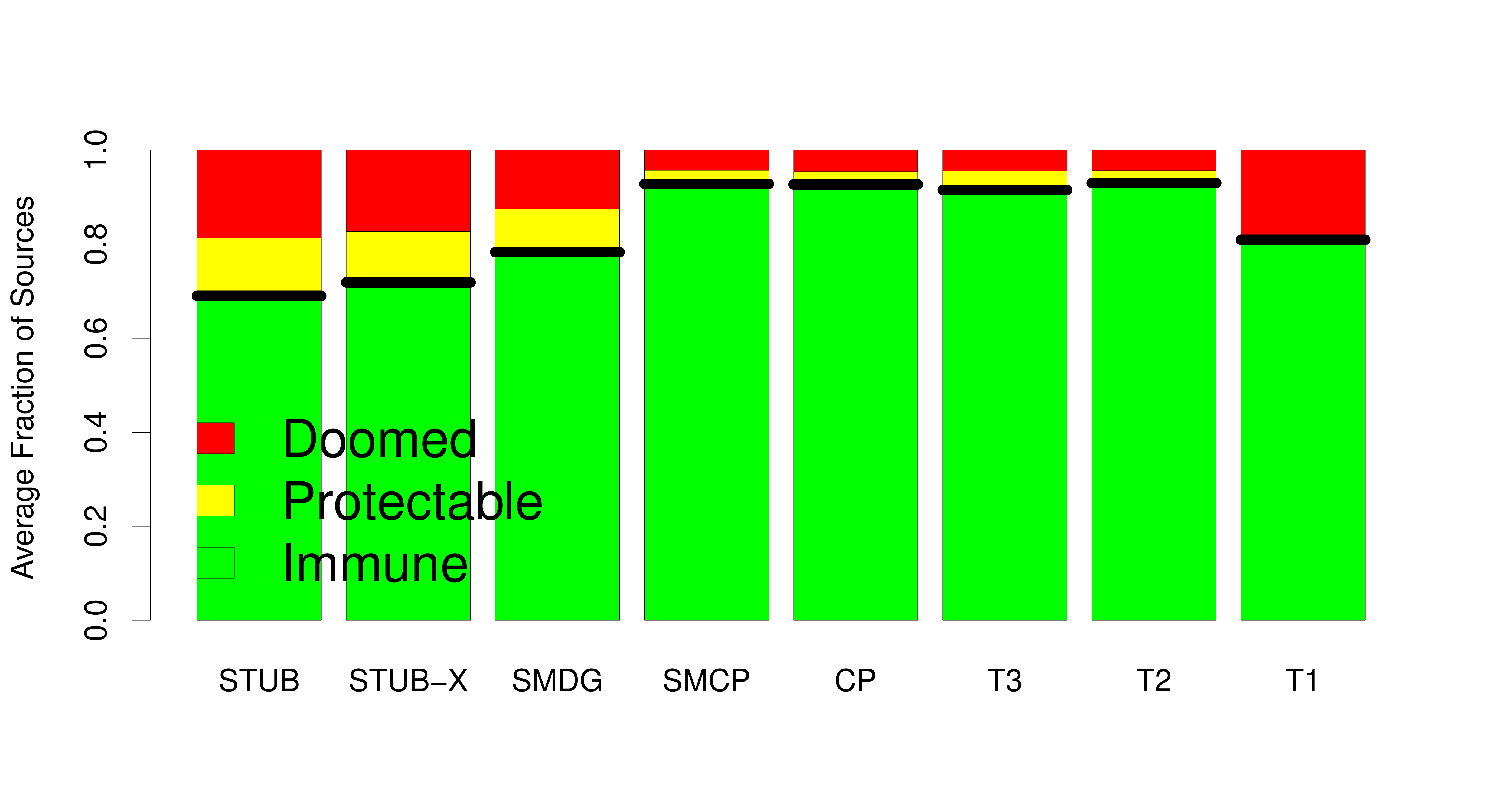} \label{fig:no_ixp_per_dest_sl_stubborn_2}
}
\subfigure[]{
\includegraphics [scale = .06]{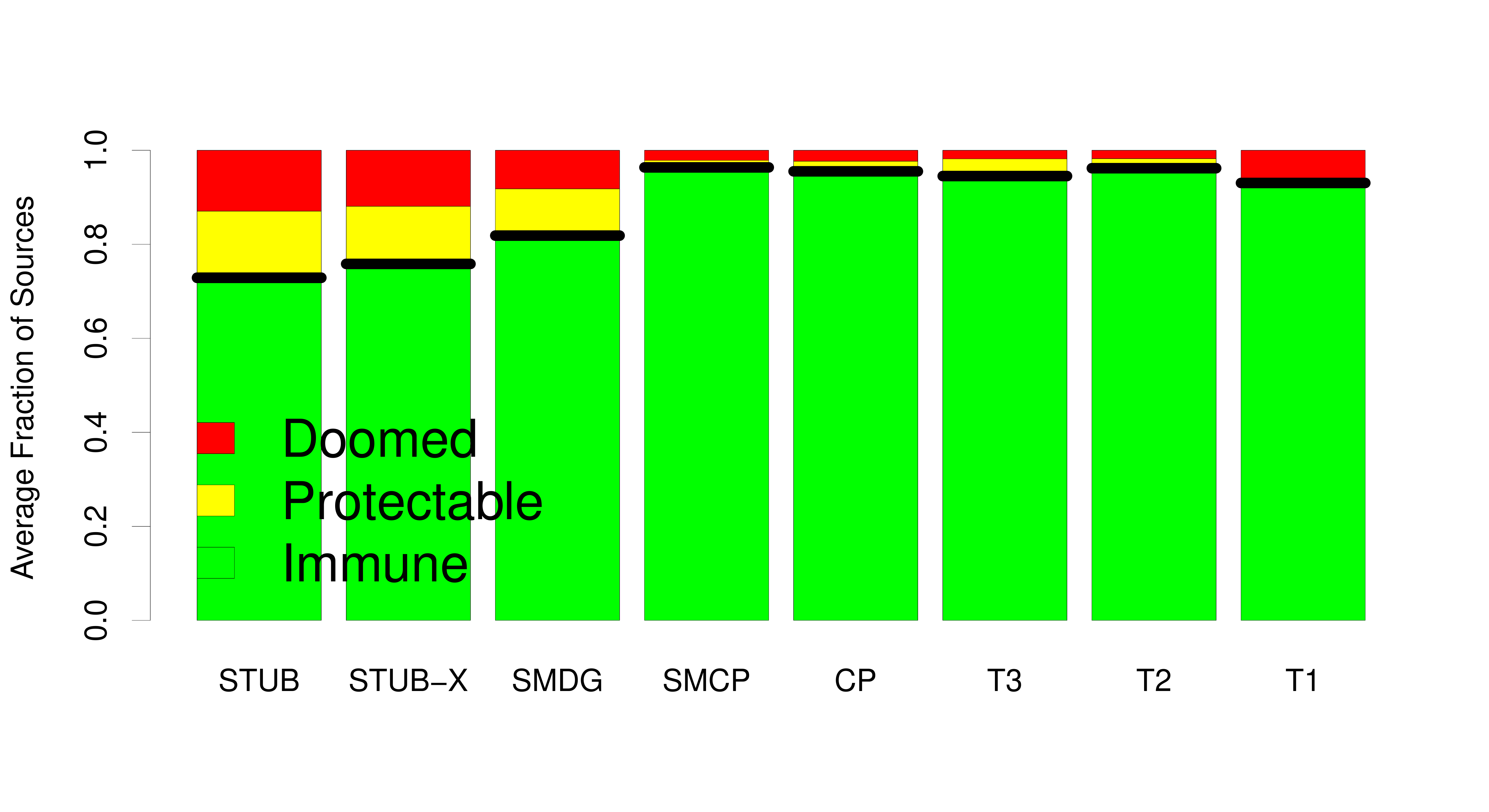} \label{fig:ixp_per_dest_sl_stubborn_2}
}
\subfigure[]{
\includegraphics [scale = .06]{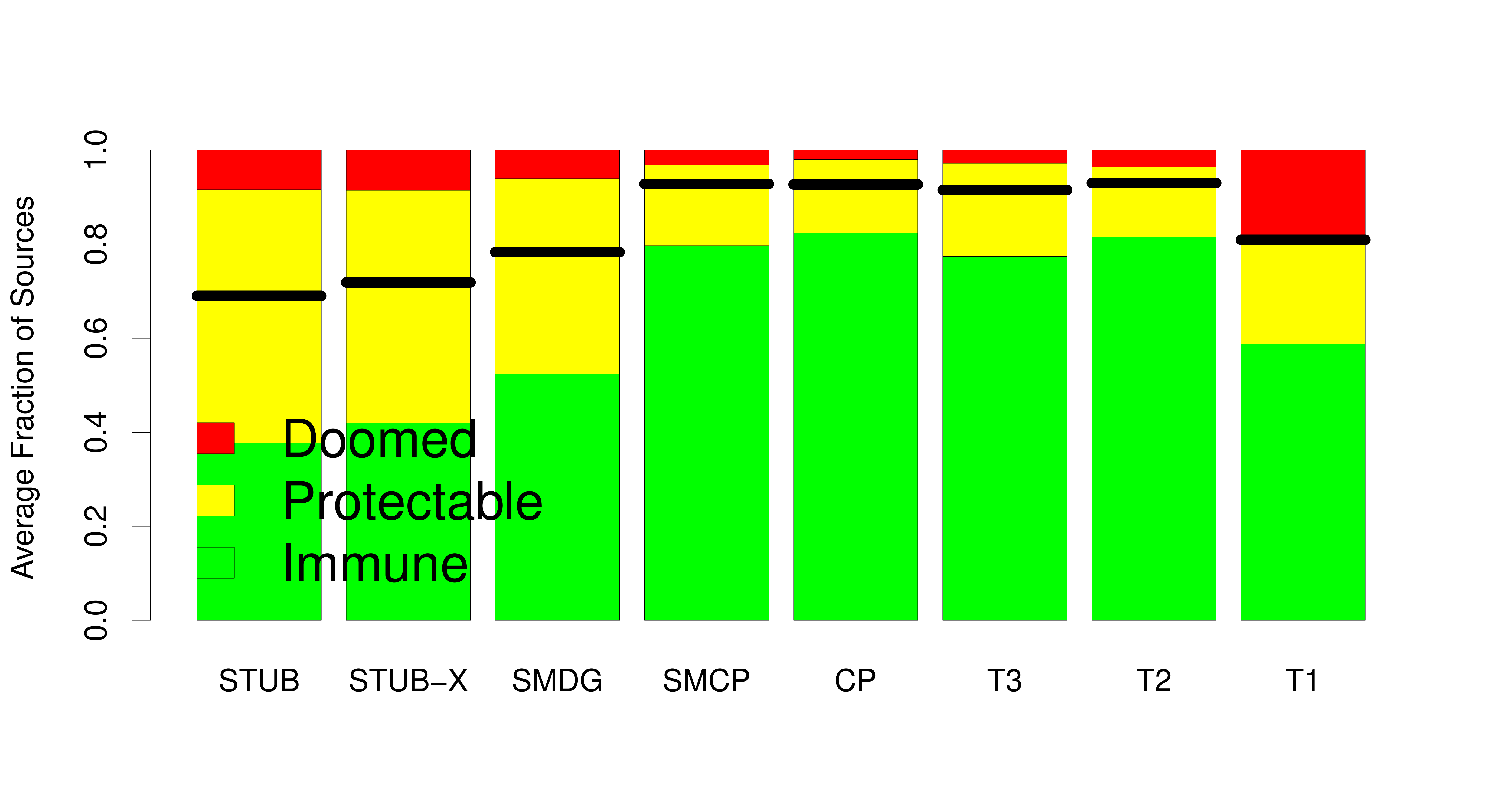} \label{fig:no_ixp_per_dest_ss_stubborn_2}
}
\subfigure[]{
\includegraphics [scale = .06]{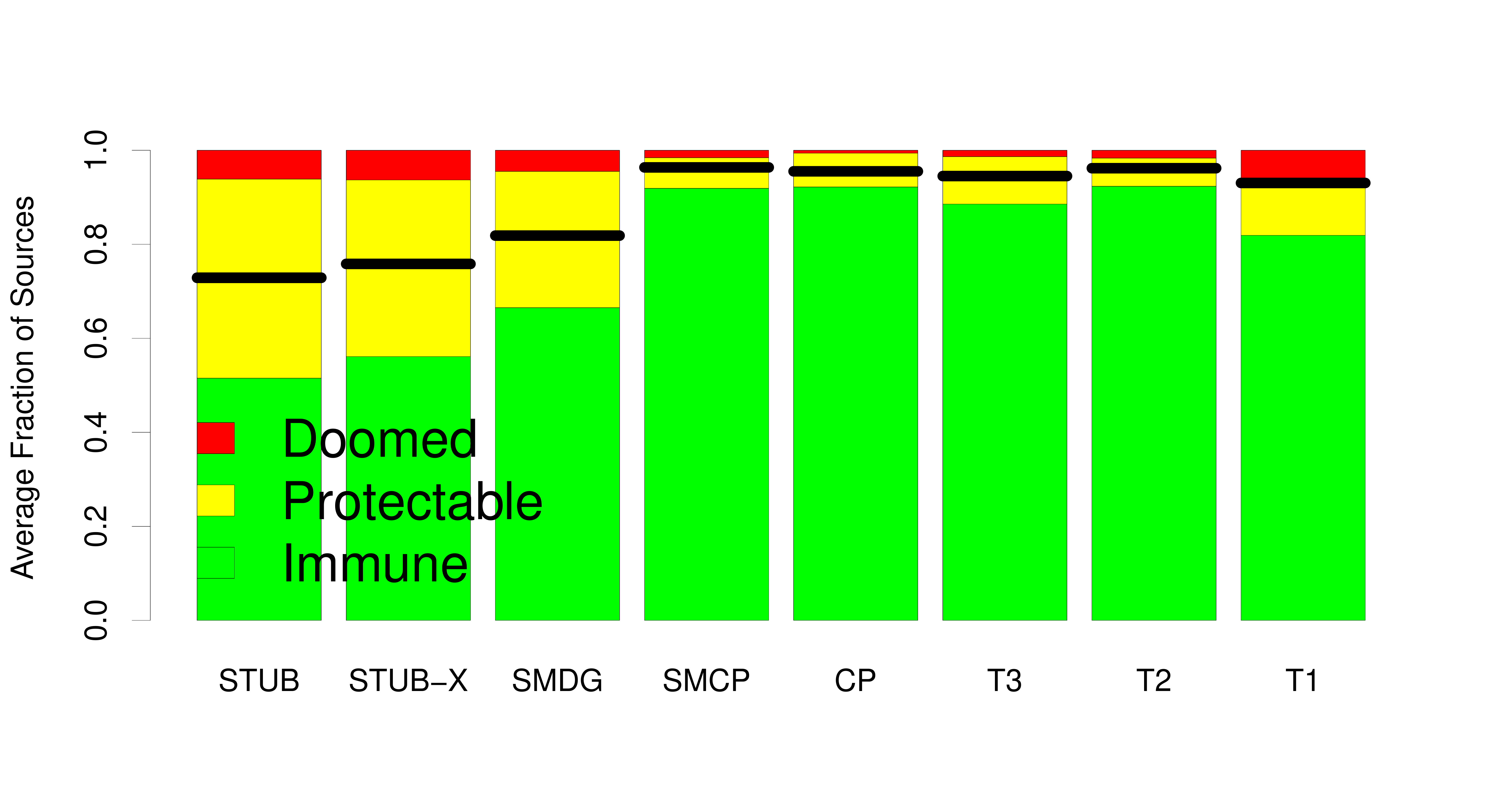} \label{fig:ixp_per_dest_ss_stubborn_2}
}
\caption{Partitions by destination tier for the \LP{2} policy variant.
\subref{fig:no_ixp_per_dest_sl_stubborn_2} UCLA graph, security \third.
\subref{fig:ixp_per_dest_sl_stubborn_2}  IXP-augmented graph, security \third.
\subref{fig:no_ixp_per_dest_ss_stubborn_2} UCLA graph, security \second. \subref{fig:ixp_per_dest_ss_stubborn_2}  IXP-augmented graph, security \second.
}\label{fig:2stubborn:parititions:dest}
\end{center}
\end{figure}

Why is it that high-degree destinations do \emph{not} require protection from S*BGP in the \LP{2} model?  Consider a source AS $s$ that has a long ($>2$ hop) customer route and short ($\leq 2$ hop) peer route to the destination $d$.  In \LP{2}, $s$ will chose the short peer route, so an attacker $m$ that wishes to attract traffic from $s$ must be exactly one hop away from $s$ (so that he can announce the bogus two-hop path ``$m,d$'' directly to $s$, that $s$ will prefer if $m$ is his customer, or if $m$ is a peer that is preferred according to his tiebreak rule). When $m$ is not one hop away from $s$, $s$ is immune. Since $m$ is unlikely to be exactly one-hop away from every source AS that prefers a short peer route in \LP{2} over a long customer route that it would have used in our original routing policy model, we see more immune nodes on average in \LP{2}.  This effect is stronger on the IXP-augmented graph because it contains more peering edges, and therefore more short peering routes.

\mypara{2. } While in Section~\ref{sec:T1suck} we found that most ASes that wish to reach Tier 1 destinations are doomed,  this is no longer the case in the \LP{2} model; while the Tier 1 destinations still do not have quite as many immune ASes as the Tier 2s do, the vast majority of source ASes that wish to reach Tier 1 destinations are immune when security is \third.

What is the reason for this?  Consider the security \second model. Many of the protocol downgrades we saw with the original \LP model resulted from a source AS $s$ preferring (possibly-long) bogus \emph{customer} path to the attacker $m$, over (possibly-short) peer or provider routes to the legitimate destination (\eg Figure~\ref{fig:pda2nd}).  However, in the \LP{2} policy variant, $s$ will only prefer a bogus customer path only if $s$ has no shorter ($\leq 2$ hop) \emph{peer or customer} route to the legitimate destinations; when $s$ has such route, we consider $s$ to be \emph{immune} (\cf Section~\ref{sec:doomedImm}).  For example, while AS~174 in Figure~\ref{fig:pda2nd} was doomed in our original \LP model when security is \second, with the \LP{2} variant and security \second AS~147 is now immune, because it has a one-hop peer route to the legitimate Tier 1 destination!

Our results indicate that this situation is common. Comparing Figure~\ref{fig:2stubborn:parititions:dest} with Figure~\ref{fig:partitions:dest:sl}-\ref{fig:partitions:dest:ss}, suggests that during attacks on Tier 1, 2, and CP destinations, there are many ASes that have short ($\leq2$ hop) peer routes to the legitimate destination $d$, and are therefore choosing those routes instead of long bogus customer routes to the attacker $m$.  Moreover, in the IXP-augmented graph, that are many more ($\approx 4X$) peering edges than in the UCLA graph, which accounts for the increased number of immune nodes we saw for the security \second model in Figure~\ref{fig:2stubborn:parititions}.

While this is good news for the Tier 1s, we point out that in the \LP{2} model this is little need for S*BGP to protect the Tier 1, 2, 3 and CP destinations, since most source ASes that wish to reach these destinations (\ie $>80\%$) are happy in the baseline scenario already!

\fi

\end{document}